\documentclass[twocolumn,notitlepage,nofootinbib,floatfix,superscriptaddress]{revtex4-2}
\usepackage{amsmath, amsfonts, amsthm, amssymb}
\usepackage[margin=1in]{geometry}
\usepackage{graphicx,xcolor,tikz}
\usepackage[inkscapelatex=false]{svg}
\usepackage{hyperref}
\usepackage{qcircuit}
\usepackage{mathrsfs}
\usepackage{soul}
\usepackage[shortlabels]{enumitem}
\usepackage{algorithm}
\usepackage[noend]{algpseudocode}
\usepackage{mathtools}
\usepackage{placeins}

\algblock{REPEAT}{ENDREPEAT}
\algnotext{ENDREPEAT}

\makeatletter 
\renewcommand\onecolumngrid{% <<<<<<
\do@columngrid{one}{\@ne}%
\def\set@footnotewidth{\onecolumngrid}% <<<<<<<<<<<<<<<<
\def\footnoterule{\kern-6pt\hrule width 1.5in\kern6pt}%
}

\renewcommand\twocolumngrid{% <<<<<<
        \def\footnoterule{% restore rule
        \dimen@\skip\footins\divide\dimen@\thr@@
        \kern-\dimen@\hrule width.5in\kern\dimen@}
        \do@columngrid{mlt}{\tw@}
}%
\makeatother  

\newtheorem{theorem}{Theorem}[section]
\newtheorem{lemma}[theorem]{Lemma}
\newtheorem{definition}[theorem]{Definition}
\newtheorem{remark}[theorem]{Remark}
\newtheorem{corollary}[theorem]{Corollary}

\newcommand{\row}{\operatorname{row}}
\newcommand{\floor}[1]{\left\lfloor{#1}\right\rfloor}
\newcommand{\supp}{\operatorname{supp}}
\newcommand{\code}[1]{$[\![ #1 ]\!]$}
\newcommand{\res}[2]{\left. #1 \right|_{#2}}
\newcommand{\mc}{\mathcal}
\def\tr{{\rm Tr}}

\newcommand*{\scrL}{\mathscr{L}}
\newcommand*{\cE}{\mathcal{E}}
\newcommand*{\cG}{\mathcal{G}}
\newcommand*{\cP}{\mathcal{P}}
\newcommand*{\bbZ}{\mathbb{Z}}
\newcommand*{\bx}{\mathbf{x}}

\newcommand{\sL}{\mathscr{L}}
\newcommand{\bp}{\mathbf{p}}
\algnewcommand{\LongState}[1]{\State%
	\parbox[t]{\dimexpr\linewidth -\algorithmicindent}{\strut #1\strut}}
\newcommand{\func}[1]{{\texttt{#1}}}
\DeclareMathOperator*{\ball}{\mathcal{B}}
\newcommand*{\diam}{\text{Diam}}

\makeatletter
\renewcommand{\frontmatter@abstractheading}{\vspace{0.1pt}\par}
\makeatother

\begin{document}

\author{Shouzhen Gu}
\thanks{These authors contributed equally to this work.\\\href{mailto:shouzhen.gu@yale.edu}{shouzhen.gu@yale.edu}\\\href{mailto:cahalibor@me.com}{cahalibor@me.com}}
\affiliation{Yale Quantum Institute \& Department of Applied Physics, Yale University, New Haven, CT, USA}
\author{Libor Caha}
\thanks{These authors contributed equally to this work.\\\href{mailto:shouzhen.gu@yale.edu}{shouzhen.gu@yale.edu}\\\href{mailto:cahalibor@me.com}{cahalibor@me.com}}
\affiliation{School of Computation, Information and Technology, Technical University of Munich, Germany \& \\
Munich Center for Quantum Science and Technology, Munich, Germany.}
\author{Shin Ho Choe}
\affiliation{School of Computation, Information and Technology, Technical University of Munich, Germany \& \\
Munich Center for Quantum Science and Technology, Munich, Germany.}
\affiliation{IQM Quantum Computers, Munich, Germany}
\author{Zhiyang He}
\affiliation{Department of Mathematics, Massachusetts Institute of Technology, Cambridge, MA, USA}
\author{Aleksander Kubica}
\affiliation{Yale Quantum Institute \& Department of Applied Physics, Yale University, New Haven, CT, USA}
\author{Eugene Tang}
\affiliation{Departments of Mathematics \& Physics, Northeastern University, Boston, MA, USA}

\title{Layer codes as partially self-correcting quantum memories}

\begin{abstract}
We investigate layer codes, a family of three-dimensional stabilizer codes that achieves optimal scaling of code parameters and energy barrier,
as candidates for self-correcting quantum memories.
First, we introduce two decoding algorithms for layer codes with provable guarantees for local stochastic and adversarial noise, respectively.
We then prove that layer codes constitute partially self-correcting quantum memories which outperform previously analyzed models such as the cubic and welded solid codes.
Notably, we argue that partial self-correction without the requirement of efficient decoding is more common than expected, as it arises solely from a diverging energy barrier. This draws a sharp distinction between partially self-correcting systems and partially self-correcting memories.
Another novel aspect of our work is an analysis of layer codes constructed from random Calderbank–Shor–Steane codes.
We show that these random layer codes have optimal scaling (up to logarithmic corrections) of code parameters and energy barrier.
Finally, we present numerical studies of their memory times and report behavior consistent with partial self-correction.
\end{abstract}

\maketitle

Self-correcting quantum memories (SCQM) are systems that can passively protect quantum information from noise~\cite{Terhal2015,QMFT}.
Development of SCQM in three or fewer spatial dimensions provides an attractive path to fault-tolerant quantum computers~\cite{Shor1996,Preskill1998}, and also offers fundamental insights into quantum many-body physics~\cite{zeng2019quantum}.
Despite their practical and theoretical importance, the existence of SCQM in 3D remains largely unresolved while in 2D, there are many no-go results~\cite{BTbound,HaahPreskill,Alicki2009}.

One may search for SCQM within stabilizer quantum error-correcting codes~\cite{gottesman1997stabilizercodesquantumerror}.
In that setting, the Hamiltonian describing the quantum memory is specified by the parity checks of the code, with the ground space corresponding to the codespace.
It is typically assumed that the system is weakly and locally coupled to a Markovian thermal bath~\cite{QMFT}; then, an important condition associated with self-correction is an energy barrier that grows with the linear system size $L$~\cite{Komar2016,temme2016thermalizationtimeboundspauli}.
The energy barrier captures the largest energy cost the environment needs to pay in order to change the encoded information by applying local errors.
The first example of a SCQM is the 4D toric code~\cite{dennisTopologicalQuantumMemory2002,alicki2008thermalstabilitytopologicalqubit}, with an energy barrier $\Delta=\Theta(L)$
\footnote{We use a mathematical notation, such as $O, \Theta$ and $\Omega$, to describe the limiting behavior of a function.}.
In 3D, the cubic code~\cite{Haah,BravyiHaah} has an energy barrier $\Delta=\Omega(\log L)$ and provided the first example of a partially SCQM, where the memory time increases polynomially in $L$ but only for $L<L^* = e^{\Theta(\beta)}$, where $\beta=1/T$ is the inverse temperature~\cite{bravyi2011analytic}.
Another example of a partially SCQM is the 3D welded solid code, with the memory time increasing subexponentially in $L$ for $L < L^* = e^{\Theta(\beta)}$~\cite{Michnicki_2014,michnicki_thesis,Siva_2017}.

Recently, Williamson and Baspin~\cite{williamson2023layer} introduced the layer code construction. 
Layer codes have parity checks that are geometrically local in 3D and can saturate the Bravyi-Poulin-Terhal-Haah (BPTH) bounds~\cite{BPTbound,BTbound,HaahBound}, achieving optimal scaling of code parameters and energy barrier for 3D stabilizer codes.
A layer code is formed by stacking layers of the surface code~\cite{Kitaev_2003,bravyi1998quantumcodeslatticeboundary} and intricately coupling them along one-dimensional line defects according to the parity-check matrices of an input Calderbank-Shor-Steane (CSS) code~\cite{cssCalderbankShor,Steane1996}.
Since thermal excitations within each layer can move without incurring energy penalty, one may hastily conclude that layer codes are not SCQM~\cite{Lin2024}.
This would be consistent with the fact that the 2D toric code is not a SCQM~\cite{Alicki2009}.
However, when excitations hit line defects, they split and this process is energetically suppressed.
Given a layer code with many line defects, it is unclear which process is dominant, and one might hope for SCQM.

In this work, we explore layer codes as potential SCQM.
Motivated by an intuition that layer codes with many line defects may exhibit better error suppression, we go beyond the scenario from Ref.~\cite{williamson2023layer} and consider layer codes constructed from random CSS codes with dense parity-check matrices.
By developing new proof techniques, we show that, with high probability and up to logarithmic corrections, these random layer codes have optimal scaling of code parameters and energy barrier.
Next, we introduce two decoding algorithms for layer codes:
(i) the cluster decoder, which has a provable threshold against local stochastic noise, and (ii) the concatenated decoder, which corrects against adversarial noise.
Using the concatenated decoder, we then prove that layer codes constructed from quantum Tanner codes~\cite{leverrier2022quantumtannercodes}, as well as random layer codes exhibit partial self-correction stronger than that of the cubic code. In particular, their memory times increase exponentially in $L$ for $L<L^* = e^{\Theta(\beta)}$.
We also argue that the phenomenon of partial self-correction is more common than previously expected, as it arises from a diverging energy barrier. 
Finally, we use statistical-mechanical methods to numerically study the memory times of random layer codes and observe scaling behavior consistent with partial self-correction.

\section{Layer codes}

To construct a layer code $\scrL(C)$, we specify a pair of parity check matrices $C=(H_X,H_Z)$ of any CSS code.
Then, for each qubit, $X$-check, and $Z$-check of $C$ we introduce one layer of the surface code with appropriate boundaries; see Fig.~\ref{fig_layer}.
For brevity, we refer to them as $Q$-, $X$-, and $Z$-layers, respectively.
The couplings between different intersecting layers are determined by $C$; see Appendix~\ref{sec:appendix_a}.
Intuitively, $\scrL(C)$ can be viewed as a concatenation of $C$ with the surface code, where parity checks of $C$ are implemented via lattice surgery~\cite{Horsman2012}.
In what follows, we sometimes use the linear system size $L= \Theta(n)$.

By construction, layer codes have parity checks of weight at most $6$.
Let $C$ be an \code{n,k,d} input code with $\rho_X n$ $X$-checks, $\rho_Z n$ $Z$-checks, maximum parity-check and qubit degree weight $w$, and energy barrier $\Delta_C$, where $\rho_X,\rho_Z\in\left(0,1/2\right)$ are integer multiples of $1/n$.
Then, as shown in Ref.~\cite{williamson2023layer}, the layer code $\scrL(C)$ is a \code{\Theta(n^3),k,\Omega(nd/w)}
code with energy barrier $\Omega(\Delta_C/w^2)$.
Recently, a series of constructions~\cite{TillichZemorHGP,hastings2020fiber,evra2020decodable,PKAlmostLinear,Breuckmann2020} culminated in proofs that quantum low-density parity-check (qLDPC) codes can have asymptotically good parameters~\cite{PK2022asymptoticallygoodquantumlocally,leverrier2022quantumtannercodes,dinur2023,lin2022goodquantumldpccodes,hsieh2025explicitlosslessvertexexpanders}.
For such input codes, layer codes have code parameters \code{\Theta(L^3),\Theta(L),\Theta(L^2)} and energy barrier $\Theta(L)$ that saturate the BPTH bounds in 3D~\cite{BPTbound,BTbound,HaahBound}.
\begin{figure}[t!]
    \centering
    \includegraphics[width=.9\columnwidth]{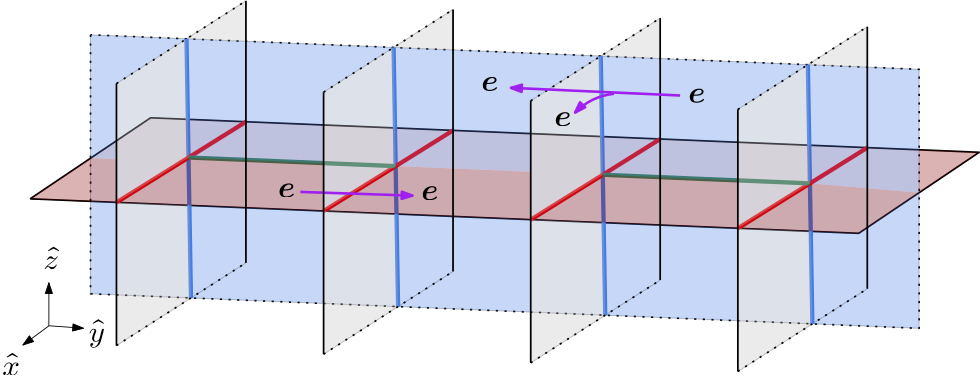}\\
    \caption{
    A layer code for $H_X = H_Z = (1,1,1,1)$.
Gray, red, and blue layers of the surface code correspond to qubits, $X$-checks and $Z$-checks, respectively.
Upon crossing line defects (thick lines), excitations may split, which is determined by the fusion rules; see Appendix~\ref{sec:appendix_a}.
}
    \label{fig_layer}
\end{figure}

\subsection{Quasiparticle excitations}

The memory system associated with a stabilizer code is described by a commuting Hamiltonian $H_\text{mem}$, whose terms correspond to parity checks.
The ground space of $H_\text{mem}$ corresponds to the codespace and excited states correspond to configurations of unsatisfied checks of the stabilizer code.
In the context of the surface code, unsatisfied checks correspond to locations of quasiparticle excitations, $e$ and $m$.
In the bulk, excitations can only be created in pairs, as captured by the fusion rules $e\times e = m\times m = 1$.
For layer codes, we also have $e$ and $m$ excitations.
Within individual layers, the fusion rules are the same as for the surface code except in the vicinity of line defects.
There, excitations can move or split between layers; see Fig.~\ref{fig_layer}.
For layer codes, there is a constant energy penalty to create excitations and their movement within individual layers is not energetically suppressed unless they cross line defects, where they split.

\subsection{Random layer codes}

One may hope that layer codes with many line defects give rise to SCQM, as the movement of excitations can be energetically suppressed.
Constructing such layer codes would require an input code with a dense parity-check matrix, which emphatically is not a qLDPC code.
In what follows, we prove that such layer codes achieve favorable parameters.

Let us define a random CSS code by choosing, uniformly at random, two parity-check matrices $H_X \in \mathbb{F}_2^{\rho_X n\times n}$ and $H_Z \in \mathbb{F}_2^{\rho_Z n\times n}$ satisfying $H_X H_Z^\mathrm{T}=0$.
We denote the resulting ensemble of codes as $\mathrm{CSS}_n(\rho_X,\rho_Z)$.
Note that a random CSS code has maximum check weight $\Theta(n)$ with high probability.
Consequently, the proof technique from Ref.~\cite{williamson2023layer} gives trivial bounds on the parameters of the resulting layer code.
Nevertheless, we are able to show the following theorem; see Appendix~\ref{sec:appendix_d} for details.

\begin{theorem}
\label{thm:random_layer_maintext}
Let $\scrL(C)$ be a layer code with the input code $C \sim \mathrm{CSS}_n(\rho_X,\rho_Z)$.
Then $\scrL(C)$ has distance $d = \Omega(n^2/\log n)$ and energy barrier $\Delta = \Omega(n/\log n)$ with high probability.
\end{theorem}

As a corollary, a random layer code has parameters \code{\Theta(n^3),\Theta(n),\Omega(n^2/\log n)} and energy barrier $\Omega(n/\log n)$, with high probability.
Notably, up to logarithmic corrections, this is optimal scaling that saturates the BPTH bounds and is achieved without using asymptotically good qLDPC codes, in contrast with all previous works~\cite{williamson2023layer,portnoy2023localquantumcodessubdivided,lin2024geometricallylocalquantumclassical}.
This has practical implications, as the codes are smaller, making them amenable to numerical simulations.

\section{Layer code decoders}

We now describe two decoding algorithms for layer codes, the cluster decoder and the concatenated decoder.
The cluster decoder builds on the renormalization group~\cite{bravyi2011analytic,DuclosCianci2010} and union-find~\cite{Delfosse2021almostlineartime,delfosse2022toward} decoders.
It partitions the error syndrome into local clusters and tries to remove them independently.
The concatenated decoder is inspired by the matching~\cite{dennisTopologicalQuantumMemory2002,Brown2022} and color code decoders~\cite{Kubica2023}.
It first removes the low-level error syndrome of individual surface code layers, which results in some high-level error syndrome of the input code that is subsequently decoded using any decoder of the input code.

\begin{figure}[t!]
    \centering
    \includegraphics[width=0.9\linewidth]{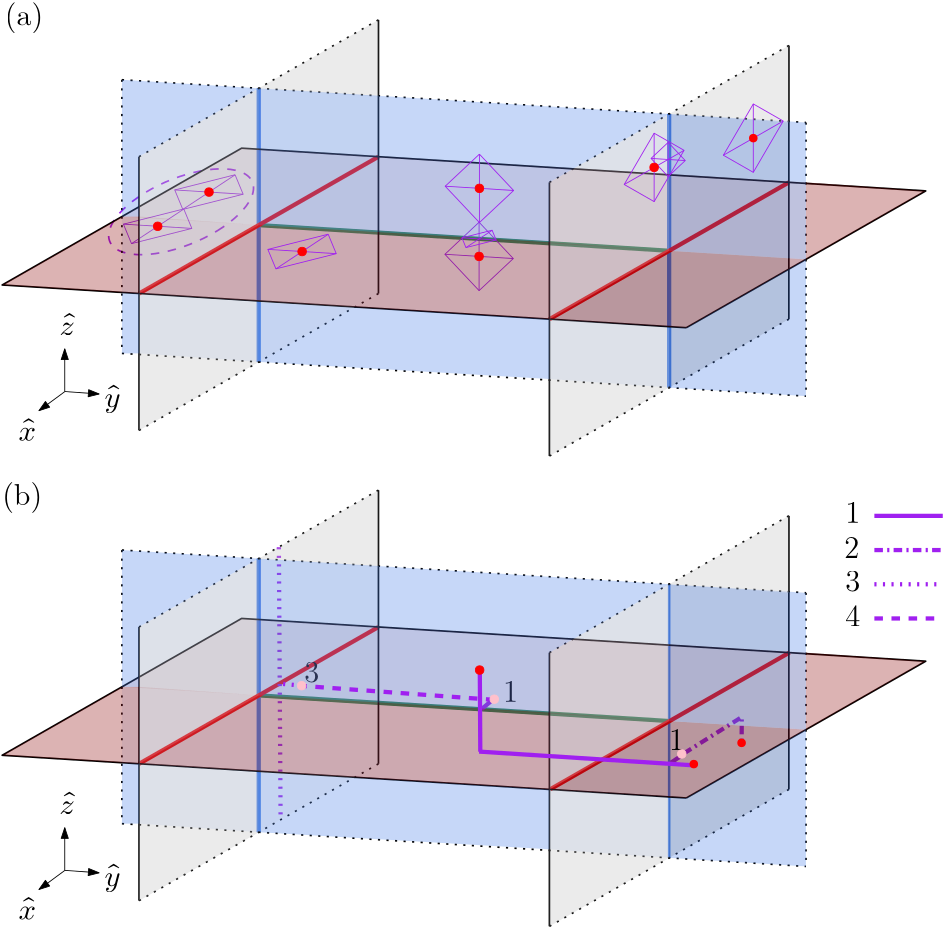}
    \caption{
Decoding algorithms for layer codes.
Red dots represent initial excitations.
(a) The cluster decoder gradually grows and merges clusters of excitations (purple shapes), and removes them if they are neutral (dashed oval); other clusters are not neutral.
(b) The concatenated decoder sequentially matches excitations in the $Z$-, $Q$-, and $X$-layers. At different stages 1--4, the decoder finds Pauli operators (purple lines); new excitations may also be created (pink dots labeled 1 and 3).}
\label{fig_decoders}
\end{figure}

\subsection{Cluster decoder}\label{sec:main_cluster_decoder}

The cluster decoder starts by initializing a list of clusters, each corresponding to the location of an individual excitation.
While the list is nonempty, the cluster decoder:
(i) grows each cluster by adding locations in its neighborhood,
(ii) merges any overlapping clusters, and 
(iii) removes any neutral clusters.
To determine if a cluster is neutral, we check whether a system of linear equations associated with that cluster is solvable.
When the list becomes empty, the cluster decoder returns a recovery operator composed of local operators found in step (iii) (that remove neutral clusters); see Fig.~\ref{fig_decoders}(a).

Using arguments similar to those in Ref.~\cite{Bravyi2013}, we prove that the cluster decoder has a nonzero QEC threshold for locally stochastic noise; that is, for sufficiently small noise strength, the logical error rate goes to zero as the layer code size grows.
This result in particular applies to independent and identically distributed depolarizing noise, viewed as a special case of the locally stochastic model.
See Appendix~\ref{sec:appendix_b}.

\subsection{Concatenated decoder}

The concatenated decoder applies minimum-weight perfect matching (MWPM)~\cite{dennisTopologicalQuantumMemory2002} to different types of surface code layers in stages and matches excitations between them with a decoder of the input code, see Fig.~\ref{fig_decoders}(b).
In particular, to decode $Z$ errors, it proceeds as follows.
\begin{enumerate}[noitemsep, topsep=0pt,parsep=0pt,leftmargin=1.7em,topsep=0.5em] 
\item For each $Z$-layer, apply MWPM to eliminate excitations within this layer; this may create new excitations on the $Q$- and $X$-layers.
\item For each $Q$-layer, apply MWPM to eliminate excitations on this layer; this may create new excitations on the $X$-layers.
\item For each $X$-layer, evaluate the parity of the number of excitations on this layer. Treating the resulting bit string as a syndrome, use a decoder of the input code to obtain a correction, and apply vertical strings on the corresponding $Q$-layers.
\item Match the excitations on the $X$-layers in an arbitrary way.
\end{enumerate}

To analyze partial self-correction of random layer codes, we will also consider a version of the concatenated decoder that matches all excitations on the $Z$- and $Q$-layers to the top boundary in Steps 1 and 2.
For either version, it is important to invoke the decoder of the input code in Step 3 to obtain an even number of excitations on each $X$-layer, as boundaries of the $X$-layers do not condense $e$ excitations. Decoding $X$ errors is analogous.
Note that if the input code decoder is efficient, then so is the concatenated decoder. We prove that the concatenated decoder can correct high-weight adversarial errors.

\begin{theorem}
    Let the input $C$ be an \code{n,k,d} code with parity checks of weight at most $w$.
    Then, given a decoding algorithm for $C$ that corrects all $Z$ errors up to weight $\alpha d$, where $\alpha>0$, the concatenated decoder corrects all $Z$ errors of the layer code $\mathscr L(C)$ of weight $O(\alpha d n /w)$.
\end{theorem}

As a corollary, we obtain that for any layer code constructed from a qLDPC code, the concatenated decoder corrects adversarial errors of weight linear in the layer code distance.

\section{Self-correction and quantum memories}

We now focus on the performance of finite-temperature systems in terms of storing quantum information.
Let $\{C_n\}$ be a family of stabilizer codes.
Each $C_n$ is associated with a memory system described by a Hamiltonian $H^{(n)}_\text{mem}$.
We work in the Davies weak-coupling limit~\cite{Davies}, i.e., each memory system is weakly-coupled to its own thermal bath at inverse temperature $\beta$, which is Markovian and acts locally on the physical degrees of freedom of the memory system.
At time $t=0$, the state $\rho^{(n)}(0)$ of the memory system is in the ground space of $H^{(n)}_{\rm mem}$.
The evolution of $\rho^{(n)}(t)$ is determined by the Lindblad master equation
\begin{align}
\frac{d\rho^{(n)}(t)}{dt}=-i\left[H^{(n)}_{\rm mem}, \rho^{(n)}(t)\right]+\mathcal{L}^{(n)}(\rho(t)),
\end{align}
where $\mathcal{L}^{(n)}$ is a Lindblad generator describing energy dissipation from the memory system to the bath and it is chosen so that the evolution drives $\rho^{(n)}(t)$ to the canonical Gibbs state $\rho^{(n)}_{\beta}\propto \exp\left(-\beta H^{(n)}_{\rm mem}\right)$; see Refs.~\cite{Lindblad,GKS,BreuerPetruccione}.

To analyze how the ability to recover $\rho^{(n)}(0)$ from $\rho^{(n)}(t)$ changes in time $t$, we need a decoder $\Phi$, which is a quantum channel implementing an appropriate recovery and returning
$\Phi\left(\rho^{(n)}(t)\right)$ that, hopefully, is close to $\rho^{(n)}(0)$.
For concreteness, we consider decoders of the form
\begin{align}
\Phi(\cdot)= \sum_\sigma P_\sigma \Pi_\sigma (\cdot) \Pi_\sigma P_\sigma,
\end{align}
where the summation is over different error syndromes, $P_\sigma$ is a Pauli operator with syndrome $\sigma$, and $\Pi_\sigma$ is a projector onto the subspace with syndrome $\sigma$.
Let $\epsilon\in(0,1)$ be a constant.
We define the memory time $t_\text{mem}$ as the maximum time such that for all $t<t_\text{mem}$, there exists $\Phi$ satisfying
\begin{equation}
\left\|\Phi\left(\rho^{(n)}(t)\right) - \rho^{(n)}(0)\right\|_1 \leq \epsilon
\end{equation}
for any initial state $\rho^{(n)}(0)$.
We then say that the family $\{C_n\}$ is self-correcting at inverse temperature $\beta$ if the memory time diverges with $n$.

It may happen that $t_\text{mem}$ grows with $n$ for $n\leq n^*$ and is not guaranteed to increase for $n>n^*$, where $n^*$ is some cutoff size. 
This motivates the following definition: we say that the family $\{C_n\}$ is partially self-correcting if the maximum memory time $t^*_\text{mem}$ scales superexponentially in $\beta$, i.e., $t^*_\text{mem} = \exp(\omega(\beta))$; see Definition~\ref{def:pscqm}. For instance, the 4D toric code is self-correcting with $t_\text{mem} = \Theta(\exp (\beta n^{1/4}))$; the 3D cubic code and 3D welded solid code are partially self-correcting with 
$t^*_\text{mem} = \exp(\Omega(\beta^2))$ and $t^*_\text{mem} = \exp(\exp(\Omega(\beta))$, respectively, and with cutoff size $n^*=\exp(\Omega(\beta))$.

We arrive at the following general result.
\begin{theorem}
\label{thm:partial_self_corr}
Let $\{C_n\}$ be a family of qLDPC codes, where $C_n$ has $n$ physical and $k_n$ logical qubits, and energy barrier $\Delta_n$.
If $\Delta_n = \Omega(\max(k_n,\log n))$, then $\{C_n\}$ is partially self-correcting with maximum memory time $t^*_\mathrm{mem} = \exp[\Omega(\Delta_{n^*}\beta)]$ and cutoff size $n^* = \exp(\beta/2)$.
\end{theorem}

\begin{proof}
We start by introducing the energy barrier decoder $\Phi$.
If $\sigma$ is a syndrome for some Pauli error $E$ with energy barrier smaller than $\Delta_n$, then set $P_\sigma = E$; otherwise, choose as $P_\sigma$ any Pauli operator with syndrome $\sigma$.
We remark that the decoder $\Phi$ is primarily of conceptual interest, as the task of finding an appropriate Pauli error $E$ may be computationally intractable.

By definition of energy barrier, the decoder $\Phi$ successfully corrects all Pauli errors with energy barrier smaller than $\Delta_n$.
Using Lemma~1 and Eq.~(26) in Ref.~\cite{bravyi2011analytic}, we obtain
\begin{equation}
\left\|\Phi\left(\rho^{(n)}(t)\right) - \rho^{(n)}(0)\right\|_1 \le O(t) n 2^{k_n} e^{-\beta(\Delta_n-f)/2}
\end{equation}
for all $n\le n^* = e^{\beta/2}$, where $f$ is a constant depending on locality of $C_n$.
We equate the right-hand side of the inequality to a constant $\epsilon$ and find the maximum memory time
\begin{equation}
t^*_\mathrm{mem} \ge \exp(\beta(\Delta_n - f)/2 - k_n\log 2 - \log n).
\end{equation}
Since $\Delta_n=\Omega(\max(k_n,\log n)) = \omega(f)$, the term $\beta\Delta_n/2$ dominates in the exponent and the result $t^*_\mathrm{mem} = \exp[\Omega(\Delta_{n^*}\beta)]$ follows. In particular, $\Delta_n = \Omega(\log n)$ implies $t^*_\mathrm{mem}=\exp(\Omega(\beta^2))$, so the definition of partial self-correction is satisfied.
\end{proof}

As a corollary, by invoking the energy barriers from Ref.~\cite{williamson2023layer} and Theorem~\ref{thm:random_layer_maintext}, we obtain that quantum Tanner layer codes and random layer codes are partially self-correcting; their maximum memory times are $t^*_\mathrm{mem} = \exp(\exp(\Omega(\beta)))$ with cutoff sizes $n^* = \exp(\beta/2)$.
Such long memory times may obviate the need for full self-correction in practice.
Additionally, Theorem~\ref{thm:partial_self_corr} immediately implies that the codes of Refs.~\cite{portnoy2023localquantumcodessubdivided,lin2024geometricallylocalquantumclassical} are partially self-correcting.

Conventionally, a quantum memory requires the existence of a computationally efficient decoder~\cite{PhysRevA.91.032303,Brell_2016,QMFT}.
Thus, we distinguish between (partially) self-correcting systems and (partially) self-correcting \emph{memories} by the efficiency of the decoder.
The proof of Theorem~\ref{thm:partial_self_corr} uses the energy barrier decoder, which is not necessarily efficient.
Thus, whether the codes in Refs.~\cite{portnoy2023localquantumcodessubdivided,lin2024geometricallylocalquantumclassical} are partially SCQM is still unknown.
By using the concatenated decoder, we obtain the following stronger result about layer codes.

\begin{theorem}[informal]\label{thm:scqm_informal}
Quantum Tanner layer codes and random layer codes are partially SCQM with respect to the concatenated decoder.
\end{theorem}

Layer codes, due to their higher energy barrier, also attain better memory time scaling than the cubic code and welded solid code, the two known examples of partially SCQM in 3D (which also only encode a constant number of logical qubits).  As the existence of a fully SCQM in 3D is an open problem, layer codes may be optimal.
For a precise statement and proofs, see Theorem~\ref{thm:psc_ldpc_layer}, Theorem~\ref{thm:psc_random_layer} and Remark~\ref{rem:psc_random_layer_efficient}.

\section{Numerical simulations}

\begin{figure}[t!]
\centering
\includegraphics[width=.9\columnwidth]{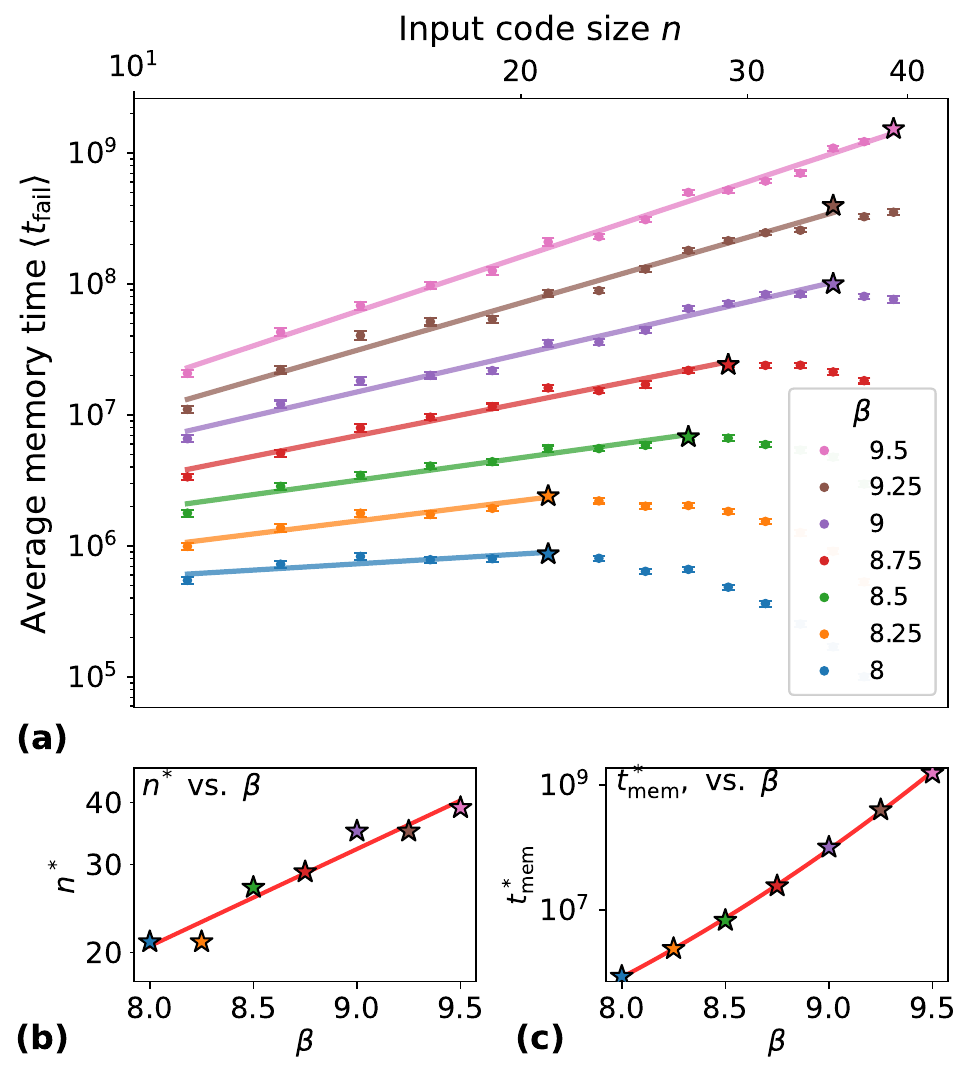}
\caption{
(a) Numerical estimates $\langle t_\text{fail}\rangle$ of the memory time $t_\text{mem}$ for random layer codes and the cluster decoder as a function of the input code size $n$ for various inverse temperatures $\beta$.
For each $\beta$, we find the maximum memory time $t^*_\text{mem}$ at the cutoff size $n^*$ (marked by $\star$).
For $n\le n^*$, we fit the data with a numerical ansatz $\log\langle t_\text{fail}\rangle \approx a \log n + b$, where $a$ and $b$ are fitting parameters.
In (b) and (c), we analyze the scaling of $n^*$ and $t^*_\text{mem}$ as a function of $\beta$, finding $n^* \approx \exp(0.448\beta-0.562)$ and $ t_\text{mem}^* \approx \exp(0.695\beta^2-7.112\beta+26.073)$.
}
\label{fig_memory}
\end{figure}

To complement the analytical results, we perform numerical simulations of the memory time of random layer codes; see Fig.~\ref{fig_memory}.
We consider a family of layer codes constructed from an ensemble $\left\{C_{n,i}\right\}$ of CSS codes, where, for given odd $n$, we sample $R=2000$ times from $\text{CSS}_n\left(\tfrac{n-1}{2},\tfrac{n-1}{2}\right)$, and keep $r=20$ codes with $k=1$ logical qubit and the highest distance.
Since the construction treats Pauli $X$ and $Z$ similarly, for concreteness, we focus on $X$ errors and consider a Hamiltonian $H^{(n,i)}_\text{mem}$, whose terms correspond to $Z$-checks of $\scrL(C_{n,i})$.

To estimate the memory time, we initialize the memory state in a ground state of $H^{(n,i)}_\text{mem}$ and perform Monte Carlo simulations to analyze its thermal evolution.
We model the thermal noise as a sequence of single-qubit $X$ errors,
and use Glauber evolution~\cite{GlauberDynamics} with continuous time rejection free algorithm~\cite{BORTZ1975,Young_1966}.
We then use the cluster decoder at geometrically increasing time intervals to probe the memory state (without implementing any recovery operator) and record the time $t_\text{fail}$ of the first decoder failure.
We repeat the process multiple times and take the sample average $\langle t_\text{fail}\rangle$ to estimate the memory time $t_\text{mem}$.
We observe that the behavior of random layer codes is consistent with partial self-correction (albeit, we find $t^*_\mathrm{mem} = \exp(\Theta(\beta^2))$ instead of the analytically-derived bound for the concatenated decoder $t_{\rm mem}^*=\exp(\exp(\Omega(\beta)))$.
To accentuate the qualitative differences between partial self-correction and self-correction, we also simulate the 3D toric code, which is self-correcting (but only for one type of Pauli errors). 
Additional details on the simulations can be found in Appendix~\ref{sec:numerics}.

\section{Conclusions}

We show that layer codes provide a concrete example of partially SCQM, with code parameters and memory time better than for previous 3D constructions.
We also argue that partial self-correction (without an efficient decoder) is more common than expected---it arises from a diverging energy barrier and its interplay with code parameters, as exemplified by Theorem~\ref{thm:partial_self_corr}.

Fundamentally, layer codes are notable because they saturate the BPTH bounds with an optimal energy barrier.
Practically, though, layer codes constructed from generic families of qLDPC codes are prohibitively qubit-intensive.
To mitigate this issue, we introduce random layer codes and prove their code parameters and energy barrier.
We thus view random layer codes as a step toward more resource-efficient constructions.

The memory time is the most important aspect for the storage of information, as there are generic methods to extract quantum information for computation~\cite{He25extractors}.
Nevertheless, identifying specific families of layer codes, developing better decoders, and designing more efficient fault-tolerant logical gates are relevant near-term challenges.
Progress on these fronts would transform layer codes from a theoretical construct into a practical quantum architecture.

\acknowledgements
We thank D.~Williamson for helpful discussions on proving partial self-correction of layer codes, in particular pointing out Ref.~\cite{Bravyi2013}. We also thank N.~Baspin, S.~Bravyi and J.~Haah for discussions on self-correction.
S.G. and A.K. acknowledge support from the NSF (QLCI, Award No.~OMA-2120757), IARPA and the Army Research Office (ELQ Program, Cooperative Agreement No.~W911NF-23-2-0219).
L.C. and S.C. acknowledge support from the European Research Council (Project EQUIPTNT, Grant No.~101001976).
L.C. is supported by Munich Quantum Valley (supported by the Bavarian state government with funds from the Hightech Agenda Bayern Plus). 
Z.H. is supported by the NSF Graduate Research Fellowship (Grant No. 2141064).
The authors gratefully acknowledge the computational and data resources provided by the Leibniz Supercomputing Centre (\url{www.lrz.de}) and the Research Center for Quantum Information, Slovak Academy of Sciences.
Part of this work was completed in the spring of 2024, while Z.H., A.K and E.T. were visiting the Simons Institute for the Theory of Computing.

\emph{Note added.---}
We would like to bring the reader’s attention to independent and concurrent work by D.~Williamson~\cite{WilliamsonPartial}, which demonstrates that layer codes constructed from quantum Tanner codes are partially SCQM.

While finalizing the manuscript, we became aware of a result by N. Baspin~\cite{baspin2025freeenergybarriereyringpolanyi}, which establishes that layer codes are not SCQM.
This is an asymptotic result (in the system size) for any fixed temperature, as such it does not apply to partial self-correction. The free energy derived is also not applicable to the setting of random layer codes.

\section*{Data availability}
Memory-time simulation data for random layer codes, together with the input-code family used in these simulations and the Python script used to generate layer code stabilizers from an input code, are available at \cite{random_layer_codes_2025}.

\clearpage
\onecolumngrid

\appendix

\section*{APPENDICES}

In what follows, we provide detailed explanations of the results from the main text.
First, in Appendix~\ref{sec:appendix_a}, we review the layer code construction.
Then, in Appendix~\ref{sec:appendix_b}, we introduce two layer code decoders and establish guarantees on their performance.
In Appendix~\ref{sec:appendix_c}, we prove that layer codes constructed from quantum Tanner codes are partially SCQM.
In Appendix~\ref{sec:appendix_d}, we describe random layer codes and prove bounds on their code parameters, demonstrating that they are partially self-correcting.
Finally, in Appendix~\ref{sec:numerics}, we present the results of our numerical simulations of layer codes.

\section{Layer Codes}\label{sec:appendix_a}

In this appendix, we provide a short overview of the layer code construction and associated background. The layer code construction is based on coupling different surface code patches in three-dimensional space. For us, a surface code is defined by a two-dimensional square grid, where qubits are placed on edges, $X$-checks associated with vertices, and $Z$-checks associated with plaquettes. The boundaries of a surface code patch may be rough or smooth, allowing $Z$ or $X$ strings, respectively, to terminate without causing non-trivial syndromes.

A CSS code $C$ is a pair of parity-check matrices $H_X\in \mathbb{F}_2^{n_X \times n}$ and $H_Z\in\mathbb{F}_2^{n_Z \times n}$ satisfying the orthogonality condition $H_XH_Z^\mathrm{T}=0$. The rows of $H_X$ (resp. $H_Z$) define the $X$-type (resp. $Z$-type) stabilizer generators, or checks, of the code. In particular, we will always regard a CSS code as coming equipped with a particular choice of stabilizer generators, and we will generally denote a CSS code explicitly by its parity-check matrices, e.g., $C=(H_X,H_Z)$. The \emph{sparsity} $w$ of a CSS code is defined to be the maximum number of non-zero entries in a row or a column of $H_X$ or $H_Z$. An infinite family of codes with constant sparsity $w=O(1)$ is said to be a quantum low-density parity-check (qLDPC) code family.

Given an arbitrary CSS code $C=(H_X,H_Z)$, the layer code $\mathscr{L}(C)$ is a 3-dimensional topological CSS code defined as follows: 
\begin{enumerate}
\item Fix an arbitrary integer $K\ge 2$, which is an implicit constant specifying the distance between surface code layers.

\item To each physical qubit of the input code $C$ we associate a ``qubit layer'', or ``$Q$-layer''. A $Q$-layer is a surface code patch parallel to the $xz$-plane with rough boundaries along the $x$-direction and smooth boundaries along the $z$-direction. The size of the surface code lattice is defined so that the distance between the smooth boundaries (resp. the rough boundaries) is $(n_Z + 1)K$ (resp. $(n_X + 1)K$).

If the input code $C$ has $n$ physical qubits, then the layer code $\scrL(C)$ consists of $n$ distinct qubit layers which we will take to be located at $y$-coordinates $y = K, 2K, \cdots, nK$. The $Q$-layer located at $y= jK$ is called the $j$-th $Q$-layer. The $x$ and $z$ coordinates of each $Q$-layer are $0\le x\le (n_Z+1)K$ and $0\le z\le (n_X+1)K$.

\item For each $X$-type stabilizer generator of the input code $C$ we associate an ``$X$-check layer,'' or ``$X$-layer.'' An $X$-layer is a surface code patch parallel to the $xy$-plane with all smooth boundaries. For $i\in [n_X]$, the $i$-th $X$-layer is located at $z = iK$ and has $x$ and $y$ coordinates in the range $0\le x\le (n_Z+1)K$ and $0\le y\le (n+1)K$.

\item For each $Z$-type stabilizer generator of the input code $C$ we associate a ``$Z$-check layer,'' or ``$Z$-layer.'' A $Z$-layer is a surface code patch parallel to the $yz$-plane with all rough boundaries. For $i\in [n_Z]$, the $i$-th $Z$-layer is located at $x = iK$ and has $y$ and $z$ coordinates in the range $0\le y\le (n+1)K$ and $0\le z\le (n_X+1)K$.

\item The arrangement of $Q$-, $X$-, and $Z$-layers defined above will intersect at various junctions. Not all intersections will be considered non-trivial:
\begin{itemize}
\item Each $X$-layer will only have non-trivial intersection with $Q$-layers corresponding to qubits in the support of the $X$-check that defines it.

\item Likewise, each $Z$-layer will only have non-trivial intersection with $Q$-layers in its support.

\item The intersection of $X$- and $Z$-layers is more intricate. Fix an arbitrary $X$-layer and $Z$-layer. The fact that their respective check operators in $C$ commute implies that there exists an even number of $Q$-layers which belong to the support of both layers. Pair off the $Q$-layers in the common support so that the first is paired with the second, the third with the fourth, and so on. Then the $X$-layer and $Z$-layer will have non-trivial intersection only along the line segments joining paired $Q$-layers.
\end{itemize}

The line (segments) defined by non-trivial intersections are called \emph{line defects}. Multiple line defects meet at \emph{point defects}.

\item The physical qubits of the layer code are defined as the union of all physical qubits in the surface codes defining the $Q$-, $X$-, and $Z$-layers. The stabilizers of the layer code coincide with ordinary surface code stabilizers everywhere away from the defects. The stabilizers are modified at the defects so as to support non-trivial fusion rules. The exact form of the stabilizers themselves will not be of great importance to us; we will interface with the stabilizers primarily through the fusion rules that they define (see section~\ref{sec:fusion}). We refer to the original paper~\cite{williamson2023layer} for detailed descriptions of defects and the modified stabilizers. 
\end{enumerate}

Fig.~\ref{fig_layer} of the main text illustrates an example of the layer code constructed from the \code{4,2,2} input code. In the figure, solid (dashed) lines on the boundaries of layers indicate smooth (rough) boundaries. Red, blue, and green lines on the intersections of the layers are $X$-$Q$, $Z$-$Q$, and $X$-$Z$ defect lines, respectively.

It will often be convenient to think of the layer code as a map $\scrL:\mathrm{CSS}\rightarrow \mathrm{CSS}$ which takes an arbitrary input CSS code and returns the associated 3D topological code defined above.

Note that different choices of parity-check matrices defining the same codespace define distinct layer codes, so the input is really the pair of matrices $(H_X,H_Z)$. The linear lengths of $\scrL(C)$ are $(n_Z+1)K\times (n+1)K\times (n_X+1)K$, so each length is $\Theta(n)$, assuming $n_X, n_Z = \Omega(n)$. For stating some results, we also define $L=K\min(n_X+1, n_Z+1)=\Theta(n)$ to be the minimum length.

We remark that our definition of layer codes deviates slightly from the one in Ref.~\cite{williamson2023layer}. In particular, the $X$- and $Z$-layers of the original construction terminate on the first and last $Q$-layers that intersect with them non-trivially. Extending the layers so that all surface code patches have comparable size does not change the essential properties of the code but is needed for our concatenated decoder to work properly and is convenient for stating our other results. In particular, the results in the rest of this section, which are proven in Ref.~\cite{williamson2023layer} for the original construction, also hold after extending the $X$- and $Z$-layers.

Many properties of the layer code are related to that of the input code, including the number logical qubits, the distance, and the energy barrier.
Let us formally define energy barrier of Pauli operators on a CSS code and of the code itself.

\begin{definition}[Energy Barrier]\label{def:energy_barrier}
Let $C=(H_X,H_Z)$ be a qLDPC CSS code. We say that $\mathcal{P}=\{P(t)\}_{t=0}^T$ is a \emph{Pauli-$Z$ path} if each $P(t)$ is a Pauli-$Z$ operator with $P(0)=I$ and $|P(t+1)P(t)| = 1$ for all $t$. 

The \emph{energy} $\Delta_Z(\cP)$ of a Pauli-$Z$ path $\cP=\{P(t)\}_{t=0}^T$ is the weight of the largest syndrome associated with the operators in the path, i.e.,
\begin{align}
\Delta_Z(\cP) = \max_{0 \le t \le T}|\sigma(P(t))|,
\end{align}
where $\sigma(Q)=H_XQ$ denotes the \emph{syndrome} of the operator $Q\in P_Z(C)$ regarded as a vector in $\mathbb{F}_2^n$.

The $Z$-\emph{energy barrier of an operator} $P_0\in P_Z(C)$, denoted $\Delta_Z(P_0)$, is the minimum energy of any Pauli-$Z$ path terminating on $P_0$, i.e.,
\begin{align}
\Delta_Z(P_0) = \min_{\cP: P(T)=P_0}\Delta_Z(\cP).
\end{align}
The $Z$-\emph{energy barrier} of the code $C$, denoted $\Delta_Z$, is the minimum energy barrier of any non-trivial logical $Z$ operator.
\begin{align}
    \Delta_Z(C) = \min_{\overline{Z}\in L_Z(C)} \Delta_Z(\overline{Z}).
\end{align}

All the notions defined here for Pauli-$Z$ operators have associated $X$ versions, and the energy barrier $\Delta_C$ of the code $C$ is the minimum of the $X$- and $Z$-energy barriers, i.e., $\Delta_C = \min(\Delta_X(C),\Delta_Z(C))$.
\end{definition}

The key properties of the layer code construction are summarized below.

\begin{theorem}[Layer Codes~\cite{williamson2023layer}]\label{thm:layer_codes}
The layer code construction, applied to an input \code{n,k,d} CSS code $C$ with:
\begin{itemize}
\item $n_X$ $X$-checks, 
\item $n_Z$ $Z$-checks, 
\item max check weight and qubit degree $w$,
\item energy barrier $\Delta_C$,
\end{itemize}
produces a $3$-dimensional topological CSS code with parameters \code{\Theta(nn_Xn_Z),k,\Omega(d/w\cdot\min(n_X,n_Z))}, max check weight $6$, and energy barrier $\Delta_{\scrL(C)}\ge 2\Delta_C/w^2$.
\end{theorem}

\begin{remark}Note that Ref.~\cite{williamson2023layer} uses a slightly different definition for the energy barrier of the input code $\Delta_C'$, which allows consecutive operators in the Pauli-$Z$ path to differ by constant weight. They obtain a lower bound $\Delta_{\scrL(C)}\ge 4\Delta_C'/w^2$. A simple modification to their argument gives $\Delta_{\scrL(C)}\ge 4\Delta_C/w^2 - 1 \ge 2\Delta_C/w^2$.
\end{remark}

When applied to asymptotically good qLDPC codes, the layer code construction produces a family of \code{\Theta(L^3),\Theta(L),\Theta(L^2)} codes with energy barrier $\Theta(L)$. These codes saturate the BPT bound~\cite{BPTbound} for topological codes in 3 dimensions.

\subsection{Fusion Rules}\label{sec:fusion}

We will work with the layer code primarily through its fusion (or condensation) rules. Similar to the surface code, the layer code supports topological excitations, which are also called anyons. A violation of an $X$-type (resp. $Z$-type) check defines an $e$-type (resp. $m$-type) excitation. These excitations are topological in that a single one cannot be created or annihilated by a local operators except in the vicinity of a boundary.

Fusion rules characterize the local structure of the layer code by identifying the types of excitations that can be locally annihilated -- or \emph{condensed} -- at a given location. Fig.~\ref{fig:fusionrules}(b) shows the fusion rules at defect lines of layer codes, where we define the abstract $\bullet$- and $\oplus$-junction in Fig.~\ref{fig:fusionrules}(a).
The regions away from defect lines, which include trivial intersections between layers, support the usual fusion rules $e\times e = m\times m = 1$ for the surface code. The fusion rules at the boundaries of layers are also the same as for the surface code, i.e., rough boundaries condense $e$ anyons and smooth boundaries condense $m$ anyons.

\begin{figure}[H]
	\centering
    \includegraphics[width=0.9\linewidth]{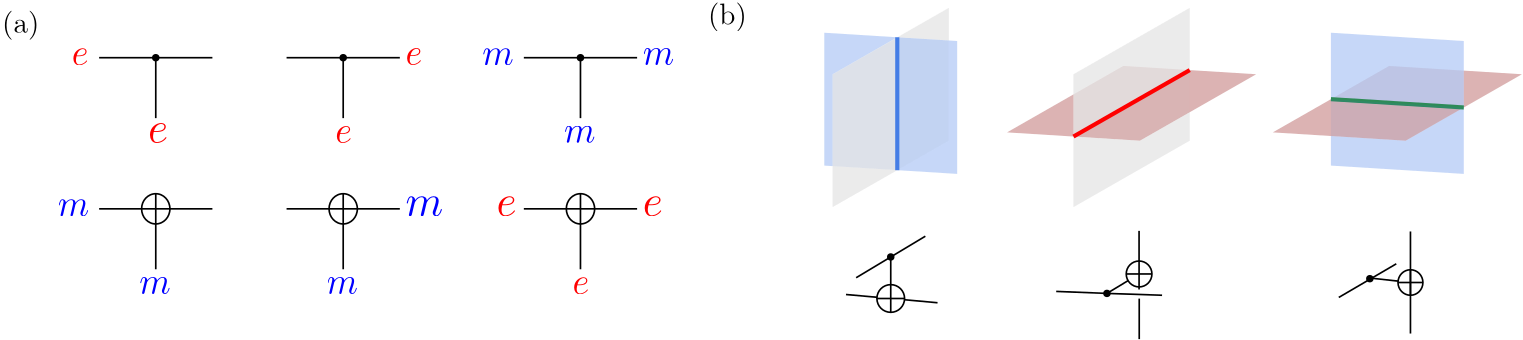}
    \caption{(a) Fusion rules for $e$ and $m$ excitations at $\bullet$- and $\oplus$-junctions. (b) Fusion rules at defect lines of layer codes defined in terms of $\bullet$- and $\oplus$-junctions.
    }
    \label{fig:fusionrules}
\end{figure}

\subsection{Logical Operators}\label{sec:LC:logicals}

In this section we collect some basic facts about the logical operators of the layer code.

Given a CSS code $C$, let $P_Z(C)$ be the set of all Pauli-$Z$ operators on $C$, and we denote the set of all non-trivial $Z$-type logical operators by $L_Z(C)$. Similarly, the $X$-type Pauli operators and non-trivial logical operators are denoted $P_X(C)$ and $L_X(C)$, respectively. As $\scrL(C)$ is also a CSS code, $P_Z(\scrL(C))$, $P_X(\scrL(C))$, $L_Z(\scrL(C))$, and $L_X(\scrL(C))$ are defined in the same way.

The logical operators of a layer code are essentially string operators whose intersections with certain surfaces are characterized by logical operators the input code.

To make this statement more precise, we first define a partition of the layer code into \emph{slabs} as well as the corresponding configuration of excitations associated with these slabs. A slab is a horizontal slice containing a single $X$-layer.

Let $\{A_1,\cdots,A_{n_X}\}$ be the partition the qubits of $\mathscr{L}(C)$, where each $A_i$ is the subset of qubits in the slab containing the $i$-th $X$-layer. Given an operator $\hat Z \in P_Z(\mathscr{L}(C))$, we denote the restriction of $\hat Z$ to $A_i$ by $\hat Z_i$.

The restriction $\hat Z_i$ defines a set of point excitations on the top boundary of $A_i$. We can identify all possible sets of such excitations with the abelian group $\bbZ^n\oplus \bbZ^{(n+1)n_Z}$ by keeping track of the number of point excitations on each $Q$-layer, and on each of the $n+1$ segments of every $Z$-layer between consecutive $Q$-layers (or the boundary).

Since excitations are free to move and merge within each $Q$-layer and within each segment of a $Z$-layer, we usually only need to keep track of their parities. This motivates the next definition.

\begin{definition}
    An \emph{$e$-configuration} is an equivalence class of $e$-excitations on slab boundaries modulo $2$. The set of all possible $e$-configurations will be denoted $\mathcal{E} = \mathbb{F}_2^n\oplus \mathbb{F}_2^{(n+1)n_Z}$. We will also denote the subspace of $e$-configurations supported on the $Q$-layers (resp. $Z$-layers) by $\mathcal{E}_Q=\mathbb{F}_2^n$ (resp. $\mathcal{E}_Z=\mathbb{F}_2^{(n+1)n_Z}$).
\end{definition}

The fusion rules for $e$-excitations induce analogous fusion rules on the set of $e$-configurations. We say that two $e$-configurations $E_1,E_2 \in \cE$ are \emph{boundary-equivalent}, or equivalent for short, denoted $E_1\approx E_2$, if one can be obtained from the other by fusion rules. We will denote the $e$-configuration associated with $\hat Z_i$ by $E_i(\hat Z)$. Equivalent $e$-configurations may arise in the following way.

\begin{lemma}[Lemma 3 of Supplementary information of Ref.~\cite{williamson2023layer}]\label{lem:equivconfigurations}
If $\hat Z\in P_Z(\scrL(C))$ leaves no excitations in the $Q$- or $Z$-layers within the $i$-th slab, then $E_{i-1}(\hat Z)\approx E_i(\hat Z)$.
\end{lemma}

We can now discuss the relation between logical operators and their $e$-configurations.

\begin{lemma}[Corollary 1 of Supplementary information of Ref.~\cite{williamson2023layer}]\label{lem:layer_slab_support}
Let $\mathscr{L}(C)$ be a layer code. Any $\overline{Z}\in L_Z(\mathscr{L}(C))$ has non-trivial support on all slabs and we have $E_i(\overline{Z}) \approx E_j(\overline{Z})$ for all $i,j\in [n_X]$.
Moreover, $\overline{Z}$ is not boundary-equivalent to the trivial, i.e., empty, configuration.
\end{lemma}

It follows that we may associate to each logical operator of the layer code a non-trivial equivalence class of $\mathcal{E}/\approx$. Since we may always use the fusion rules to move all the excitations from $Z$-layers onto the $Q$-layers, there always exist representatives in $\cE_Q$ for any equivalence class of $e$-configurations. We can identify each $E \in \cE_Q$ with a corresponding Pauli-$Z$ operator on the input code $C$ in the obvious way. Let $P(E)$ denote the Pauli operator defined by $E$.

\begin{lemma}[Remark 2 of Supplementary information of Ref.~\cite{williamson2023layer}]\label{lem:layer_input_map}
Given $E_1,E_2 \in \cE_Q$, we have $E_1\approx E_2$ if and only if $P(E_1)$ and $P(E_2)$ are stabilizer-equivalent in $C$. If $E\in \cE_Q$ is boundary-equivalent to the $e$-configuration of some $\overline{Z}\in L_Z(\scrL(C))$, then $P(E)$ is stabilizer-equivalent to a non-trivial logical operator in $C$.
\end{lemma}

Finally, given any $g\in L_Z(C)$, we can define a canonical logical operator $\overline{Z}_g \in L_Z(\scrL(C))$ such that $g=P(E_i(\overline{Z}_g))$ for all $i\in [n_X]$. The operator $\overline{Z}_g$, called the \emph{quasiconcatented representative} associated with $g$, is defined by replacing each Pauli-$Z$ in the support of $g$ with a full-length vertical $Z$-string operator in the corresponding $Q$-layer. Any excitations created at defect crossings are then merged using horizontal string operators in the affected $X$-layers.

Lemmas~\ref{lem:layer_slab_support} and \ref{lem:layer_input_map}, together with the existence of a quasiconcatented representative for any $g\in L_Z(C)$, imply that there exists a bijection between stabilizer equivalence classes of $L_Z(\scrL(C))$ and $L_Z(C)$. The forward map follows by taking the $e$-configuration of a logical operator $\overline{Z} \in L_Z(\scrL(C))$. The inverse map constructs the quasiconcatenated representative associated with $g\in L_Z(C)$. See Fig.~\ref{fig:econfiguration} for an illustration of stabilizer-equivalent logical operators of $\scrL(C)$ and their corresponding boundary-equivalent $e$-configurations.

\begin{figure}[ht]
	\centering
	(a)\includegraphics[width=.45\linewidth]{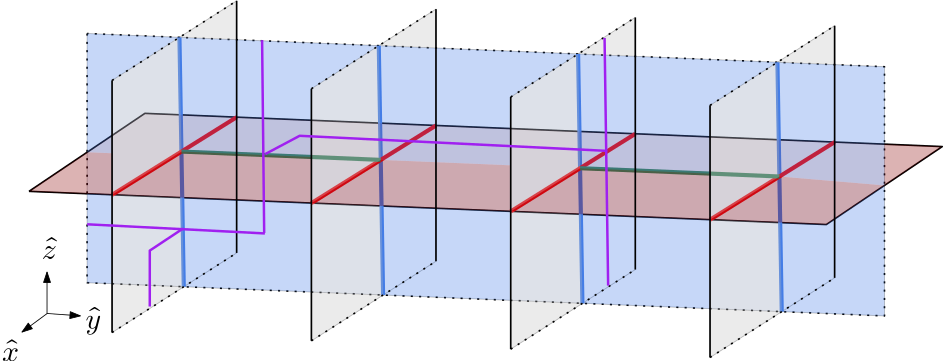}
	(b)\includegraphics[width=.45\linewidth]{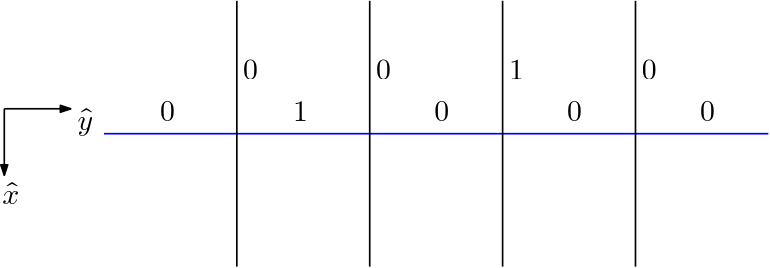}
	(c)\includegraphics[width=.45\linewidth]{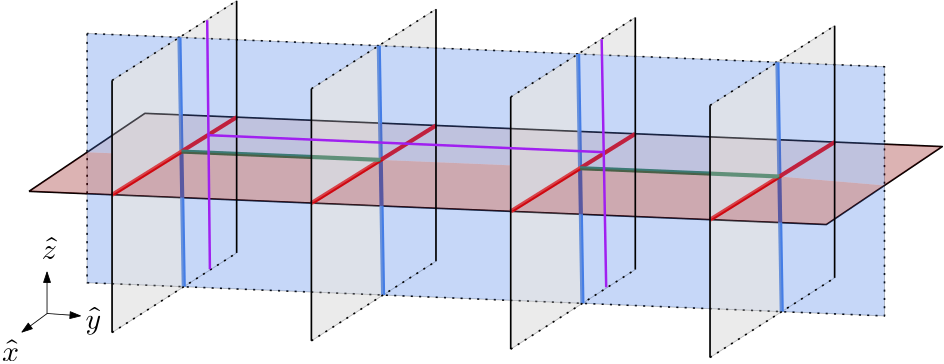}
	(d)\includegraphics[width=.45\linewidth]{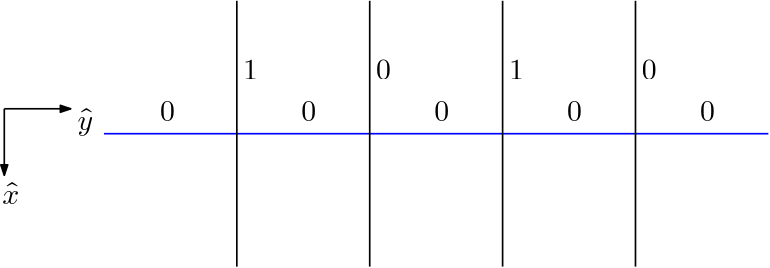}
    \caption{(a) A $Z$ logical operator of a layer code, illustrated as the purple strings. (b) A top-down view of the $e$-configuration at a slab boundary right above the $X$-layer. (c) A quasi-concatenated representative of the logical operator in (a). (d) The corresponding $e$-configuration, which is boundary-equivalent to the one in (b).
    }\label{fig:econfiguration}
\end{figure}

\section{Decoders for Layer Codes}\label{sec:appendix_b}

In this appendix, we define and prove properties of two decoders for the layer code construction. The first, called the cluster decoder, is a renormalization-group (RG) based decoder which has provable threshold against stochastic noise. The second, called the concatenated decoder, is suitable for adversarial noise when the layer code is defined with a qLDPC code as input. We will also use the concatenated decoder as a key ingredient in the proof that layer codes are partially self-correcting.
For the cluster decoder, we consider $X$ errors as the description is cleaner compared to that with $Z$ errors. For the concatenated decoder we analyze $Z$ errors. The cases for the opposite bases are mostly symmetric. Since layer codes are CSS codes, we can correct $X$ and $Z$ errors independently. Consequently, our threshold result is applicable to the setting of the depolarizing noise too.

\subsection{Cluster Decoder}\label{sec:decoder_cluster}

The cluster decoder builds upon the RG decoder of Ref.~\cite{bravyi2011analytic} and the union-find decoder of Ref.~\cite{Delfosse2021almostlineartime}. The key idea is to partition the syndrome into local, bounded-diameter clusters and then annihilate these clusters independently with a local decoder.  We follow this general idea to develop a \textit{cluster decoder} and invoke the 
arguments in Appendix~B of Ref.~\cite{bravyi2011analytic} to prove the existence of a threshold under independent stochastic noise.

Consider a layer code $\sL$
with $Z$-stabilizer generators $S_Z$. 
To decode an $X$ error with $Z$-syndrome $\sigma$, we begin by building the decoding hypergraph $\mathcal{G} = (V, E)$. 
Every vertex of $\cG$ corresponds to a $Z$-check in $S_Z$ (which are plaquettes in layer codes), and every hyperedge of $\cG$ corresponds to a qubit in $\sL$, connecting the checks (vertices) it participates in. 
The syndrome $\sigma$ then corresponds to a set of \textit{excited vertices} $V(\sigma)$, sometimes referred to as \textit{excitations}. 
We define a \textit{cluster} to be a set of vertices, some of which are excited.

\begin{definition}[Correctable Clusters]
    For a cluster $T$, let $E(T)$ denote the set of all hyperedges completely supported on vertices in $T$.
    A cluster $T$ is correctable if there is a set of hyperedges in $E(T)$ which is a matching for the set of excited vertices in $T$. In other words, this collection of hyperedges touch every excited vertex in $T$ an odd number of times, and all other vertices an even number of times. 
\end{definition}

At a high level, the cluster decoder initializes a size-one cluster at each excitation, and then iterates the following routine: 
at every iteration, grow every cluster by radius one (measured in terms of graph distance) and merge any clusters that overlap. For every cluster, check whether it is correctable; if it is, correct the excitations within the cluster and remove them. 
This main routine is described in Algorithm~\ref{cluster:main}.

\begin{algorithm}[H]
	\caption{Cluster Decoder: Main Routine}
    \label{cluster:main}
		 \textbf{Input:} A decoding hypergraph $\cG$, and a set of excitations \\
		\textbf{Output:} A correction that removes all excitation.
		\begin{algorithmic}[1]
			\State Initialization: for every excited vertex $v$, create cluster $T$ that contains only $v$. 
            Let $L$ be the list of active clusters, which currently contains all initialized size-one clusters.
            Set time $t = 0$.
            \While{there are active clusters}
                \For{each active cluster $T$}
                \State \func{Grow} $T$ by radius~1.
                \EndFor
                \State \func{Merge} all active clusters that overlap. Replace overlapping clusters by their union in $L$.~\label{mer}
            \For{each active cluster $T$}
                \State \func{Check correctability} of $T$. If $T$ is correctable,  \func{Correct} $T$ and remove it from $L$.
            \EndFor
            Set $t = t+1$.
            \EndWhile           
			\State \Return All corrections accumulated.
		\end{algorithmic}
\end{algorithm}
We now detail these steps formally. 
For clarity, we omit implementation details such as data structures as well as certain straightforward optimizations.\footnote{For instance, there is no need to check correctability of every active cluster---only those that grow to a layer boundary or are the result of merging need to be checked.} The most important definition we need is that of a \textit{region}.

\begin{definition}[Region]
A region is a set of $Z$-check plaquettes enclosed by layer boundaries and line defects. 
We call the enclosing layer boundaries and line defects the sides of the region.
The regions of $\sL$ form a partition of its $Z$-checks.
\end{definition}
Let us enumerate the regions of $\sL$ as $R_1, \cdots, R_r$. 
For a $Z$-check $v$ (which is a vertex in $\mc G$), we let $R(v)$ denote the region it belongs to.
We observe that based on the layer code construction, qubits in $\sL$ (hyperedges in $\cG$) can be placed into three categories based on the regions they interact with, i.e., the regions with a $Z$-check that the qubit participates in.
\begin{enumerate}[(1)]
    \item A qubit $q$ on a layer boundary interacts with exactly one region. In terms of $\cG$, the hyperedge $h_q$ contains one or two vertices, depending on whether the boundary is smooth or rough, and all vertices are from the same region. 
    We call these qubits (hyperedges) \textit{smooth/rough boundary qubits (boundary hyperedges)}.
    
    \item Another category of qubits (hyperedges) is the \textit{regional qubits (regional hyperedges)}, which is the subset of the remaining qubits that correspond to all hyperedges $h$ which contain exactly two vertices, both from the same region.
    
    \item All remaining qubits (hyperedges), which we call \textit{defect qubits (defect hyperedges)}, interact with multiple regions. In particular, a defect hyperedge contains at least 2 vertices, each belonging to a distinct region.
\end{enumerate}

The motivation behind these definitions is as follows: the $Z$-check plaquettes in a region $R$ form a surface code patch with boundary conditions prescribed by the sides of $R$.
Therefore, any even number of excitations in the same region $R$ (and in the case where one side of $R$ is a smooth boundary, any number of excitations in $R$) can be annihilated by corrections supported on boundary and regional qubits of $R$. 
For this reason, for a cluster $T$ and a region $R$, we often study the intersection of $T$ and $R$, denoted $T(R)$, which is simply the set of vertices of $T$ which belongs to $R$. We call $T(R)$ a \textit{subregion}.

Before we discuss \func{Grow} or \func{Merge}, let us discuss \func{Check correctability} and \func{Correct} as they imposes conditions on what we need to keep track of in \func{Grow} and \func{Merge}. 
Consider a cluster $T$ supported on regions $R_1, \cdots, R_s$. 
If $T(R_i)$ contains an even number of excitations for all $i\in [s]$, then clearly $T$ is correctable, and the correction is supported on regional and boundary qubits.
In general, when some subregions $T(R_i)$ contain an odd number of excitations, we may apply corrections to defect qubits contained in $T$, which changes the excitation parities of subregions in $T$. 
Therefore, the correctability of $T$ depends on the defect qubits contained in $T$.
We formalize these discussions with the following definitions.

\begin{definition}[Type]
    For a defect hyperedge $h = (v_1, \cdots, v_w)$, where $v_j$ belongs to region $R_{i_j}$ for $i_j\in [r]$, the type of $h$ is the tuple $(i_1, \cdots, i_w)\in [r]^w$. 
\end{definition}

As we \func{Grow} and \func{Merge} clusters, for each cluster $T$, we keep a list $D(T)$ of defect qubits contained in $T$, each of a distinct type, such that for every defect hyperedge $h$ in $T$, there exists a defect qubit $h'\in D(T)$ such that $h, h'$ have the same type.
We call $D(T)$ the list of \textit{defect representatives} of $T$.

\begin{definition}[Correctability Equations]
    Consider a cluster $T$ supported on regions $R_1, \cdots, R_s$, with defect representatives $D(T)$. 
    We construct a system of linear equations $M\bx = \bp$ as follows. 
    \begin{enumerate}[itemsep = 0pt]
        \item Initialize a binary $|D(T)|\times s$ matrix $M'$. Enumerate the defect hyperedges in $D(T)$ as $h_1, \cdots, h_{|D(T)|}$. For $\ell\in [|D(T)|]$ and $i\in [s]$, if $h_\ell$ contains a vertex in region $R_i$, set $M'_{\ell,i} = 1$. Otherwise set $M'_{\ell,i} = 0$.

        \item For every region $R_i$, if $T(R_i)$ contains a smooth boundary qubit (a boundary hyperedge of weight 1), remove the $i$-th column of $M'$. 
        Without loss of generality, suppose the remaining columns correspond to regions $R_1, \cdots, R_{s'}$. 

        \item Remove empty or redundant rows from the modified matrix, and let $M$ be the final matrix.

        \item Initialize binary vector $\bp$ of length $s'$. For $i\in [s']$, set $\bp_i$ to be the parity of excitations in $T(R_i)$.
    \end{enumerate}
    We call $M\bx = \bp$ the correctability equations of $T$.
\end{definition}

\begin{lemma}~\label{lem:correctability}
The cluster $T$ is correctable if and only if there is a solution to its correctability equations. 
\end{lemma}
\begin{proof}
As argued above, for a subregion $T(R)$, if $T(R)$ contains a smooth boundary, then the excitations contained within the subregion can all be matched to the boundary qubit. If it contains an even number of excitations then these excitations can be matched inside the subregion. 
Therefore, the hypergraph matching problem representing correctability of a cluster can be reduced to choosing defect qubits to flip such that every subregion without a smooth boundary has even number of excitations.
This is precisely captured by the correctability equations.
Given a solution $\bx$ to the correctability equations, we will simply add the defect qubits indicated by $\bx$ to the correction, and perform matching within each (correctable) subregion. 
\end{proof}

\begin{remark}
    Note that when a cluster touches two opposite smooth boundaries of a layer, we may simply declare decoding failure as the cluster has percolated. In our implementation we continue the decoding process to completion.
\end{remark}

Our \func{Check correctability} function would therefore construct the correctability equations of $T$ and try to find a solution. The \func{Correct} function will output the correction detailed in Lemma~\ref{lem:correctability}. 
In later discussions, we sometimes use the word \textit{neutral} in place of correctable.

\func{Growing} of a cluster $T$ can be done by breadth-first search (BFS). 
At any time $t$, for every active cluster $T$, we keep a list of its outmost vertices.\footnote{In typical BFS these vertices are called boundary vertices. We use ``outmost'' to prevent overloading the word ``boundary.''} 
We visit the hyperedges of these outmost vertices, add new vertices to $T$, and update the list of outmost vertices. 
If we reach vertices inside another active cluster $T'$, or vertices reached (at time $t$) by another active cluster $T'$, we know that $T$ and $T'$ need to be merged. 
If we reach a defect hyperedge, we add it to $D(T)$ (which we keep as a sorted list), unless a defect hyperedge of the same type is already present.

\func{Merging} of clusters $T_1, \cdots, T_i$ can be done by merging the list of vertices, excitations, and defect representatives of the clusters.

By Lemma~\ref{lem:correctability}, we see that either the cluster decoder will successfully annihilate all syndromes, or we will eventually grow to a cluster that contains the entire decoding graph $\cG$ (unless we intentionally terminate the decoder), which would be correctable. 
Therefore the cluster decoder is a valid decoder. 

Our cluster decoder shares high-level similarities with the decoder of Ref.~\cite{eggerickx2024almost}, which was designed for the geometrically local codes constructed in Ref.~\cite{lin2024geometricallylocalquantumclassical}. 
Indeed, our decoders can both be viewed as adaptations of the union-find decoder of Ref.~\cite{Delfosse2021almostlineartime,delfosse2022toward}.

\subsubsection{Existence of Threshold} 

In the previous section, we described the decoder as growing the radius of clusters by one in every timestep. 
For proof of threshold, we consider a slightly different version of the decoder, where the clusters grows exponentially: at time $t$, we grow the radius of the clusters to $4^t$ before merging, checking correctability and correcting. 
In practice, we expect growing the radius by one per timestep to perform better.

In Appendix B of Ref.~\cite{bravyi2011analytic}, Bravyi and Haah showed that for a stabilizer code defined over a 3D lattice with topological order, an RG decoder has a threshold against locally stochastic noise. 
While their definitions do not exactly capture our setting, we adapt their arguments to prove a threshold against locally stochastic noise for our cluster decoder.

We use two different metrics in the argument.
The first metric is the \textit{lattice metric} $d_L$, which for two $Z$-checks (or two qubits) measures the $\ell_\infty$ distance between the checks (qubits) on the 3D lattice.
The second metric is the \textit{graph metric} $d_G$, which for two $Z$-checks $u,v$, measures the distance between $u,v$ in the hypergraph $G$.
For two qubits $e$ and $f$, $d_G$ measures the distance between $e$ and $f$ in the dual hypergraph\footnote{For a hypergraph $G = (V, E)$, the dual hypergraph $G^\top$ has one vertex $x_e$ for every hyperedge $e\in E$, and one hyperedge $h_v$ for every vertex $v\in V$. The hyperedge $h_v$ is incident to all vertices $x_e$ where $v\in e$ in $G$.} $G^\top$ of $G$.

\begin{definition}
    A collection $S$ of $Z$-checks (or qubits) is $r$-connected if $S$ cannot be partitioned into two disjoint subsets that are distance $r$ apart. 
    A $r$-connected component is a maximal $r$-connected set. 
\end{definition}
If the distance is measured with the lattice metric, we say that $S$ is $r$-lattice-connected or $S$ is a $r$-lattice-connected component. 
Similarly if the distance is measured with the graph metric, we say that $S$ is $r$-graph-connected or $S$ is a $r$-graph-connected component. 
It is evident from the construction of layer codes that the lattice metric lower bounds the graph metric.
Therefore, an $r$-graph-connected set is also $r$-lattice-connected. This observation will be used tacitly throughout this proof.

The threshold proof of Ref.~\cite{bravyi2011analytic} utilizes the notion of \textit{chunks} to bound the support of a locally stochastic error.
We include their definitions and key lemmas here.
\begin{definition}
    Let $B$ be a collection of qubits, which should be thought of as the support of an error $P$. 
    A qubit in $B$ is a level-$0$ chunk. 
    Choose an integer $Q \gg 1$.
    A subset of $B$ is a level-$n$ chunk if it is a disjoint union of two level-$(n-1)$ chunks and its diameter is at most $Q^n/2$. 
    Note that a level-$n$ chunk contains exactly $2^n$ qubits. 
\end{definition}
Similar to above, we denote the chunks lattice-chunks and graph-chunks depending on the metrics used. 
Note that a level-$n$ graph-chunk is a level-$n$ lattice-chunk.

We now partition the error support $B$ according to the levels. 
Let $B_n$ denote the union of all level-$n$ chunks. Since a level-$n$ chunk is a union of two level-$(n-1)$ chunks, we see that
\begin{align}
    B = B_0 \supseteq B_1\cdots \supseteq B_m,
\end{align}
where $m$ is the smallest integer such that $B_{m+1} = \varnothing$.
Let $F_i = B_i\setminus B_{i-1}$, then $B = F_0\sqcup F_1\cdots \sqcup F_m$ is a partition of $B$. 
Bravyi and Haah called this the \textit{chunk decomposition} of $B$. 
Using a percolation argument, they proved that the probability for a locally stochastic error to include a level-$n$ chunk is doubly exponentially decaying in $n$.
\begin{lemma}[Proof of Theorem~2, Appendix~B of Ref.~\cite{bravyi2011analytic}]\label{lem:BH-chunk-probability}
    Suppose the error support $B$ is locally stochastically distributed\footnote{We note that Bravyi and Haah assumed the error is independently and identically distributed, but their arguments apply to local stochastic noise as well.}
    with strength $\epsilon$.
    On a $D$-dimensional lattice\footnote{In their proof, Bravyi and Haah assumed that the lattice is translationally invariant. Their arguments can be easily adapted to remove this assumption, which enables us to apply this lemma to layer codes.}
    of linear size $L$, the probability that $B$ includes a level-$n$ lattice-chunk is bounded by $O(L^D)((3Q)^{2D}\epsilon)^{2^n}$.
\end{lemma}
Note that this is also an upper bound on the probability that $B$ includes a level-$n$ graph-chunk.
To complete their threshold proof, Bravyi and Haah showed that the RG decoder will correct all errors that do not include high-level chunks.
Specifically, they assumed that the stabilizer code being decoded has topological order at scale $L_{tqo}$ (see Definition~1,~\cite{bravyi2011analytic}), where $L_{tqo} \geq L^{\gamma}$, and proved the following.
\begin{lemma}[Lemma~5 of Ref.~\cite{bravyi2011analytic}]
    Let $Q\geq 10$. If the length $m$ of the lattice-chunk decomposition of $B$ satisfies $Q^{m+1}\le L_{tqo}$, then $B$ is corrected by the RG decoder.
\end{lemma}
The existence of threshold follows from the above two lemmas. 
While layer codes do not strictly satisfy their definition of topological order\footnote{Specifically, their definition only captures codes with periodic boundary conditions (see footnote~32 of Ref.~\cite{bravyi2011energy-arxiv}), while layer codes have many boundaries.},
a similar condition holds: any Pauli operator with no syndrome which is supported within a lattice-box of diameter $L/2$ is a product of stabilizers. 
This condition is a direct consequence of the fact that any logical operator of the layer code must cross an entire $Q$-layer, which has length at least $L$.
Utilizing this condition, we prove the following correctness lemma for the cluster decoder, which will imply a threshold.
\begin{lemma}\label{lem:cluster-decoder-correct}
    Let $Q\ge 180$, and let $P$ be an error supported on $B$.
    If the length $m$ of the graph-chunk decomposition of $B$ satisfies $Q^{m}\le L/60$, then $P$ is corrected by the cluster decoder.
\end{lemma}
To prove this lemma, we characterize the behavior of the cluster decoder.

For a $Z$-check $v$, let $\ball(v, r)$ denote the ball of radius $r$, measured in graph metric $d_G$, around $v$. Similarly denote a ball around a qubit $e$. 
We emphasize that for all of this proof, balls are only defined with respect to the graph metric.
For a collection $S$ of checks (qubits), let
\begin{equation}
\ball(S, r) = \bigcup_{x\in S}\ball(x, r).
\end{equation}
We can then bound the support of output correction of the cluster decoder. 

\begin{lemma}\label{lem:cluster-correction-radius}
    Consider running the cluster decoder with an input error $P$ supported on $B$. Let $C_t$ denote the output correction accumulated up to time $t$. Then 
    $P\cdot C_t$ is supported within $\ball(B, 4^t+1)$.
\end{lemma}
\begin{proof}
    Recall that we are considering a version of the cluster decoder where the cluster graph-radius grows exponentially with $t$.
    Note that all syndromes caused by $P$ must be adjacent to qubits in $B$. The cluster decoder grows graph-radius $t$ balls around these syndromes, and annihilate clusters of syndromes when possible. Since the decoder never introduces new syndromes, any accumulated correction must be supported within balls of radius $t$ centered at the input syndrome $\sigma(P)$.
\end{proof}

We also make use of the following lemma. Let $\diam(\cdot)$ denote the diameter of sets, with subscripts $G, L$ when we are considering the graph or lattice metrics, respectively.
\begin{lemma}[Proposition~7 of Ref.~\cite{bravyi2011analytic}]\label{lem:BH-chunk-diameter}
    Let $Q \ge 6$ and let $M$ be a $Q^n$-connected component of $F_n$. Then $\diam(M)\leq Q^n$ and $M$ is separated from $E_n\setminus M$ by distance more than $\frac{1}{3}Q^{n+1}$. 
\end{lemma}

\begin{proof}[Proof of Lemma~\ref{lem:cluster-decoder-correct}]
    We follow the high-level ideas of the proof in Ref.~\cite{bravyi2011analytic}.
    Let $B = F_0\sqcup F_1\cdots \sqcup F_m$ be the graph-chunk decomposition of $B$. Let $F_{j,\alpha}$ enumerate the $Q^j$-graph-connected component in $F_j$.
    Note that $F_{j,\alpha}$ form a partition of $F_j$.

    We prove the following claim by induction on $t$: $P\cdot C_t$ is stabilizer-equivalent to an operator supported in the union of all balls $\ball(F_{j,\alpha}, 4^t+1)$ for all $j$ such that $Q^j \geq 4^t/3$.

    Note that the base case where $t = 0$ is a direct consequence of Lemma~\ref{lem:cluster-correction-radius}.

    Suppose the claim holds for some $t$. From Lemma~\ref{lem:BH-chunk-diameter}, we know that
    \begin{align}
        \diam_G(F_{j,\alpha})\le Q^j, \quad d_G(F_{j,\alpha}, F_{k,\alpha}) \geq \frac{1}{3}Q^{1+\min(j,k)}.
    \end{align}
    Consequently, the distance between any pair of higher level balls is large. For $j$, $k$ such that $Q^j\ge 4^t/3$ and $Q^k\ge 4^t/3$, we have
    \begin{align}
        d_G(\ball(F_{j,\alpha}, 4^t+1), \ball(F_{k,\beta}, 4^t+1))
        \geq \frac{1}{3}Q^{1+\min(j,k)} - 2(4^t+1) \geq 60Q^{\min(j,k)} - 3\cdot 4^t > 2\cdot 4^{t+1}.
    \end{align}
    At timestep $t+1$, the cluster decoder considers all $(2\cdot 4^{t+1})$-graph-connected components of syndromes caused by $P\cdot C_t$.
    Since higher level balls are far apart, any such component of syndromes must be caused by 
    the support of $P\cdot C_t$ within a single ball $\ball(F_{j,\alpha}, 4^t+1)$.
    
    Let $j$ be such that $4^{t}/3 \leq Q^j < 4^{t+1}/3$. 
    We will show that the support of $P\cdot C_t$ in all balls $\ball(F_{j,\alpha}, 4^t+1)$ are corrected up to stabilizers by the cluster decoder at time $t+1$.
    Fix one of these balls, we consider its diameter. 
    \begin{align}
        \diam_G(\ball(F_{j,\alpha}, 4^t+1))
        \leq Q^j + 2(4^t+1) < 4^{t+1}.
    \end{align}
    Since the diameters are less than $4^{t+1}$, the syndrome $S$ caused by errors supported in this ball is $(2\cdot 4^{t+1})$-connected, which means it will be examined by the cluster decoder at time $t+1$. 
    The cluster decoder will further determine that this cluster of syndromes is correctable, as a valid correction (namely, the error itself) is supported within the ball $\ball(F_{j,\alpha}, 4^t+1)$, which is contained in the cluster. 
    Therefore, the decoder will correct $S$. 
    The correction must be supported within $\ball(F_{j,\alpha}, 4^{t+1}+1)$, because any correction must be within radius $4^{t+1}$ of the syndromes $S$ and $S$ must be in the 1-neighborhood of $F_{j,\alpha}$. 
    To see this last claim, note that the cluster decoder never introduce new syndromes. Therefore, $S$ must be a subset of the input syndrome, which means it must be in the 1-neighborhood of the input error $F_{j,\alpha}$.

    The diameter of these correction balls are again small.
    \begin{align}
        \diam_G(\ball(F_{j,\alpha}, 4^{t+1}))\le Q^j + 2\cdot 4^{t+1} < 30Q^{j} < L/2.
    \end{align}
    Therefore the support of $P\cdot C_t$ in all balls $\ball(F_{j,\alpha}, 4^t+1)$ are corrected up to stabilizers.
    Consequently, by Lemma~\ref{lem:cluster-correction-radius}, $P\cdot C_{t+1}$ will be stabilizer-equivalent to an operator supported on the union of all balls $\ball(F_{k,\beta}, 4^{t+1}+1)$ where $k > j$.
    This completes our induction.
    By taking $t$ such that $4^t > 3Q^m$, we see that the cluster decoder corrects $P$ up to stabilizers. 
\end{proof}

Combining the above lemmas, we conclude that the cluster decoder has a threshold.
\begin{theorem}\label{thm:cluster-decoder-threshold}
    There exists a constant threshold $\epsilon_*$ such that for all $\epsilon < \epsilon_*$, for a layer code of linear size $L$ experiencing locally stochastic noise of strength $\epsilon$, the cluster decoder fails with probability upper bounded by $\exp(-\Theta(L^{\eta}))$ for some constant $\eta$.
\end{theorem}
\begin{proof}
    By Lemma~\ref{lem:BH-chunk-probability}, the probability that the error support includes a level-$n$ lattice-chunk is at most $O(L^D)((3Q)^{6}\epsilon)^{2^n}$. Since a level-$n$ graph-chunk is a level-$n$ lattice-chunk, the same upper bound applies. 
    Set $\epsilon_* = (3Q)^{-6}$.
    Choose $m = O(\log L/\log Q)$ such that $Q^m < L/60$. 
    By Lemma~\ref{lem:cluster-decoder-correct}, we see that if the error does not include level-$m$ graph-chunks, then the cluster decoder succeeds in correcting the error. 
    Therefore, the decoding failure probability is bounded by $O(L^D)\exp(-\Theta(2^m)) = \exp(-\Theta(L^{\eta}))$ for $\eta = 1/\log Q$.
\end{proof}

\subsection{Concatenated Decoder}\label{sec:decoder_concatenated} 

The concatenated decoder proceeds by first decoding ``low-level'' syndromes on the surface code patches in order to obtain a ``high-level'' syndrome which is a valid syndrome of the input code. We then apply a decoder of the input code to find a high-level correction, which, when lifted to a corresponding correction on the surface code layers, annihilates all syndromes on the layer code. We prove that the concatenated decoder succeeds against adversarial errors of weight at most $O(d)$, where $d$ is the distance of the layer code defined with a qLDPC input code (cf. Theorem~\ref{thm:concatenateddecoder}). When the layer code is instantiated with an asymptotically good qLDPC code, the concatenated decoder also succeeds against errors with sufficiently small energy barrier, which is the key property required for partial self-correction (cf. Appendix~\ref{sec:appendix_c}).

Let $e_0$ be a $Z$ error affecting the layer code $\mathscr{L}(C)$, and let $\sigma_0$ be its syndrome (decoding $X$ errors is analogous). We define $R_Z$, $R_Q$, and $R_X$ to be the $Z$-, $Q$-, and $X$-layers of the code, respectively.
The decoder operates by applying minimum-weight perfect matching (MWPM) to these different layers and using a decoder $\operatorname{Dec}_C^Z$ for $Z$ errors of the input code when it gets stuck.
There are four main steps to the decoding procedure:
\begin{enumerate}
	\item Considering only the $Z$-layers, apply MWPM on $R_Z$ to eliminate all excitations on those layers. Because the boundaries of the $Z$-layers are all rough, we may match to any of the four boundaries. Note that we ignore defect lines with the other layers when performing the matching. In particular, this step may create additional excitations on the $Q$- and $X$-layers.
	\item Considering only the $Q$-layers, apply MWPM on $R_Q$ to eliminate all excitations on those layers (both the original excitations from $\sigma_0$ and the new ones from the correction in Step 1). We may match to the top and bottom boundaries, which condense $e$-excitations. We again ignore defect lines when performing the matching.
	\item Let $\sigma_C$ be the indices of all $X$-layers that now contain an odd number of excitations. Viewing $\sigma_C$ as a syndrome of the input code, we apply the decoder $\operatorname{Dec}_C^Z$ to find a $Z$ correction $\hat f_C$ with that syndrome. For every qubit in $\hat f_C$, apply a vertical string operator from the top to the bottom of the corresponding $Q$-layer of the layer code, creating additional excitations on the $X$-layers.
	\item Apply MWPM on $R_X$ to eliminate the remaining excitations on the $X$-layers (from $\sigma_0$ and the three previous steps). We may not match to any of the boundaries because $X$-layer boundaries do not condense $e$-excitations.
\end{enumerate}

The first step eliminates all excitations on $R_Z$. The second eliminates all excitations on $R_Q$. Since the second and third steps only apply corrections within $R_Q$, the fusion rules imply that no new excitations are created in $R_Z$. Finally, the fourth step eliminates all excitations in $R_X$ while not creating additional excitations in $R_Z$ or $R_Q$. Therefore, as long as all the steps are achievable, the decoder will output a correction that returns the state back to the codespace. The first and second steps are always achievable because we may match to a boundary. The non-trivial result is to show that the input decoder $\operatorname{Dec}_C^Z$ can correct the syndrome $\sigma_C$ in the third step and that the resulting number of excitations on every $X$-layer is even so that MWPM can be applied in the fourth step.

Pseudocode for the concatenated decoder is presented in Algorithm~\ref{alg:concatenateddecoder}.

\begin{algorithm}[H]
\caption{Concatenated decoder}
\label{alg:concatenateddecoder}
    {\textbf{Require:}} A decoder $\operatorname{Dec}_C^Z$ for the input code and the $\operatorname{MWPM}$ decoder.\\
	\textbf{Input:} The syndrome $\sigma_0$ of a $Z$ error on a layer code $\mathscr{L}(C)$.\\
	\textbf{Output:} A $Z$ correction $\hat f$ of syndrome $\sigma_0$.
	\begin{algorithmic}[1]
		\State $\sigma_Z \gets \res{\sigma_0}{R_Z}$
		\State $\hat f_Z \gets \operatorname{MWPM}(\sigma_Z, R_Z)$ \label{algstep:fZmatching}
		\State $\sigma_Q \gets \res{\left(\sigma_0 + \sigma(\hat f_Z)\right)}{R_Q}$ \label{algstep:sigmaQ}
		\State $\hat f_{Q}^1 \gets \operatorname{MWPM}(\sigma_Q, R_Q)$ \label{algstep:fQmatching}
		\State $\sigma_C \gets$ indices of $X$-layers such that $\sigma_0 + \sigma(\hat f_Z + \hat f_{Q}^1)$ has an odd number of excitations \label{algstep:sdef}
		\State $\hat f_C \gets \operatorname{Dec}_C^Z(\sigma_C)$ \label{algstep:classicaldecoder}
		\State $\hat f_{Q}^2 \gets$ vertical strings of $Z$ operators on all $Q$-layers specified by $\hat f_C$
		\State $\sigma_X \gets \res{\left(\sigma_0 + \sigma(\hat f_Z + \hat f_{Q}^1 + \hat f_{Q}^2)\right)}{R_X}$
		\State $\hat f_X \gets \operatorname{MWPM}(\sigma_Q, R_X)$ \label{algstep:fXmatching}
		\State \Return $\hat f_Z + \hat f_{Q}^1 + \hat f_{Q}^2 + \hat f_X$
	\end{algorithmic}
\end{algorithm}

To analyze the decoder, we first prove two lemmas about cleaning the support of the residual error from $Z$- and $Q$-layers after applying the matching. 
Illustrations of the process are presented in Fig.~\ref{fig:coloringintersection}.

\begin{figure}[htpb]
	\centering
    \includegraphics[width=0.9\linewidth]{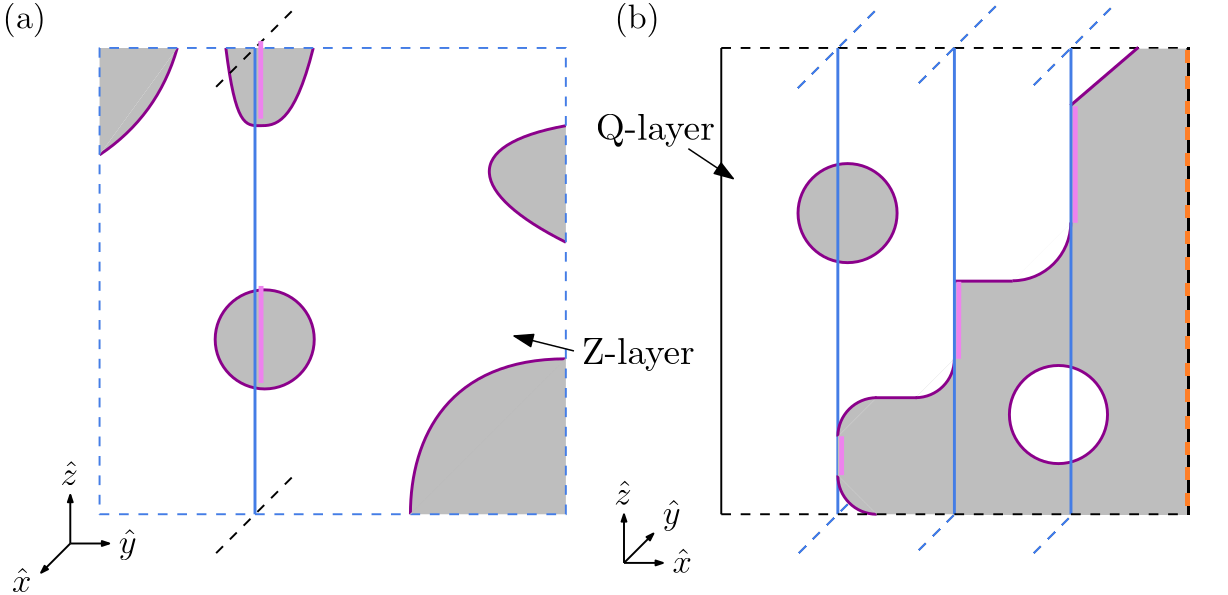}
	\caption{Cleaning the support of operators from $Z$- and $Q$-layers.
    (a) An example of how $e_z+\hat f_z$ is cleaned from the $Z$-layer $z$ in Lemma~\ref{lem:Zlayercleaning}. The dark purple lines represent the error and correction $e_z + \hat f_z$, the gray regions represent the stabilizer $s_z$, and the light purple lines represent the cleaned operator $(e_z + \hat f_z + s_z)|_{q} = \res{s_z}{q}$ for a $Q$-layer $q$.
    (b) Placing the residual error on a $Q$-layer $q$ into the form in Lemma~\ref{lem:Qlayercleaning}. The light purple lines along the defect lines with the $Z$-layers $\mathbf z_q$ is the result of cleaning the support of $e_0+\hat f_Z$ from $\mathbf z_q$ to $q$ in (a). The dark purple lines represent $(e_0 + \hat f_Q^1)\big|_{q}$. The shaded regions indicate the stabilizer $s_Q$. After multiplication by $s_Q$, what remains is the dashed orange string spanning the height of the layer (drawn to coincide with the right boundary here). This string is only present when $\big|e_0|_{q \cup \mathbf z_{q}}\big| \ge L/4$.
    }
	\label{fig:coloringintersection}
\end{figure}

\begin{lemma}
	\label{lem:Zlayercleaning}
	Let $z$ be a $Z$-layer and $\mathbf q$ be the $Q$-layers with non-trivial intersection with $z$. There exists a stabilizer $s_z$ such that $e_z + \hat f_z + s_z$ is supported on $\mathbf q$ and $R_X$, where $e_z = \res{e_0}{z}$ and $\hat f_z = \hat f_Z\big|_z$. Furthermore, if $|e_z| < L/2$, the stabilizer can be chosen so that $\big|(e_z + \hat f_z + s_z)|_{q}\big| = |\res{s_z}{q}| \le 2|e_z|$ for all $q\in \mathbf q$.
\end{lemma}
\begin{proof}
	The fact that $e_z + \hat f_z$ leaves no excitation on $z$ means that it consists of closed loops and strings that terminate on the boundaries of $z$. Coloring the areas in $z$ bounded by the strings and the surface code boundaries using two colors, we let $s_z$ be the product of all stabilizers on the plaquettes of one color. Because all boundaries of $z$ condense $e$ anyons, the operator $e_z + \hat f_z + s_z$ has no support on $z$ and is only supported on $\mathbf q$ and $R_X$ (stabilizers defined by a plaquette in $z$ will also have support on $Q$- and $X$-layers that intersect $z$ non-trivially).
	
    If $|e_z| < L/2$, then the minimum-weight condition of the matching in Line~\ref{algstep:fZmatching} of Algorithm~\ref{alg:concatenateddecoder} implies that $|\hat f_z| \le |e_z|$, so $|e_z + \hat f_z| < L$. Thus, no string from $e_z + \hat f_z$ crosses the entire layer. We choose $s_z$ to be defined by the color of the smaller areas. Since $(e_z + \hat f_z + s_z)|_q$ has support which is the intersection of the area with the $qz$-defect line, its size is at most $|e_z + \hat f_z| \le 2|e_z|$.
\end{proof}

\begin{lemma}
	\label{lem:Qlayercleaning}
	For any $Q$-layer $q$, let $\mathbf z_q$ denote the $Z$-layers that intersect it non-trivially. Let $R_Q'$ denote the set of $Q$-layers $q$ such that $|\res{e_0}{q \cup \mathbf z_{q}}| \ge L/4$. Then $e_0 + \hat f_Z + \hat f_{Q}^1$ is stabilizer-equivalent to an operator $p$ supported only on $R_Q'$ and $R_X$, and $\res{p}{R_Q'}$ consists of single strings going from the top to the bottom of certain $Q$-layers in $R_Q'$.
\end{lemma}
\begin{proof}
Define $s_z$ as in Lemma~\ref{lem:Zlayercleaning}, and let $s_Z = \sum_z s_z$. We apply similar reasoning as in the previous proof. The fact that $e_0 + \hat f_Z + s_Z + \hat f_Q^1$ leaves no excitations on the $Q$-layers means that it consists of closed loops and strings that terminate on the top or bottom boundaries of the $Q$-layers. For each $Q$-layer $q$, color the enclosed areas gray and white, with the area touching the front boundary of each layer being white. Let $s_q$ be the product of all stabilizers on the plaquettes colored gray, and define $s_Q = \sum_q s_q$.
	
Let us now analyze the operator $p = e_0 + \hat f_Z + s_Z + \hat f_Q^1 + s_Q$, which is stabilizer-equivalent to $e_0 + \hat f_Z + \hat f_Q^1$. By construction, $p$ has no support on $R_Z$. For a given $Q$-layer $q$, there is no support on $q$ if the area touching the back boundary is also white; otherwise, the support is a single vertical string along the back boundary of $q$. In particular, if $(e_0 + \hat f_Z + s_Z + \hat f_Q^1)|_q$ does not contain any strings that stretch the whole layer vertically, then $p$ does not have support on $q$.
	
Now suppose $q\notin R_Q'$, i.e., $|\res{e_0}{q\cup \mathbf z_q}| < L/4$. Then $|e_z| < L/4$ for all $z\in \mathbf z_q$, so Lemma~\ref{lem:Zlayercleaning} implies that $\big|(e_z + \hat f_z + s_z)|_{q}\big| \le 2|e_z|$. Thus, we have
\begin{align}
\big|(e_0+\hat f_Z + s_Z)|_{q}\big| &= \bigg|\res{e_0}{q} + \sum_{z\in \mathbf z_q}(e_z + \hat f_z + s_z)|_{q}\bigg|\\
		&\le \big|\res{e_0}{q}\big| + 2\sum_{z\in \mathbf z_q}|e_z|\\
		&\le 2|\res{e_0}{q\cup \mathbf z_q}|\\
		&< L/2.
\end{align}
	
From the minimum-weight condition of the matching in Line~\ref{algstep:fQmatching} of Algorithm~\ref{alg:concatenateddecoder} as well as the fact that $(e_0 + \hat f_Z + s_Z)|_{q}$ is an operator supported on $q$ with the syndrome $\res{\sigma_Q}{q}$, we also have $\big|\hat f_Q^1|_{q}\big| < L/2$. Therefore, $(e_0 + \hat f_Z + s_Z + \hat f_Q^1)|_{q}$ has weight less than $L$ and cannot contain a string stretching the whole layer vertically. This shows that the only $Q$-layers where $p$ can have a vertical string are the ones in $R_Q'$.
\end{proof}

Now, we are ready to prove that the input code decoder can be applied when the initial error is sufficiently small.

\begin{lemma}
	\label{lem:inputdecodervalid}
	Let $\alpha > 0 $ be a constant. If $4w|e_0| \le \alpha d_{\mathrm{in}}L$, then the error $e_C$ of the input code $C$ corresponding to the $Q$-layers on which $p$ from Lemma~\ref{lem:Qlayercleaning} is supported has weight at most $\alpha d_{\mathrm{in}}$, and syndrome $\sigma(e_C)=\sigma_C$ in Line~\ref{algstep:sdef} of Algorithm~\ref{alg:concatenateddecoder}.
\end{lemma}
\begin{proof}
First, note that $R_Q'$ consists of at most $\alpha d_{\mathrm{in}}$ layers, since
\begin{equation}
	\frac{L}{4}|R_Q'| \le \sum_{q\in R_Q'} \left|\res{e_0}{q\cup \mathbf z_q}\right| \le \left|\res{e_0}{R_Q}\right| + w\left|\res{e_0}{R_Z}\right| \le w|e_0| \le \frac{\alpha d_{\mathrm{in}} L}{4}.
\end{equation}
The syndrome $\sigma_C$ is defined as the $X$-layers with an odd number of excitations for the operator $e_0 + \hat f_Z + \hat f_Q^1$, which is stabilizer-equivalent to $p$ as defined in Lemma~\ref{lem:Qlayercleaning}. Because an operator supported on $R_X$ cannot change the parity of the number of excitations on $X$-layers, we only need to consider $p|_{R_Q'}$. Let $e_C$ be the $Q$-layers in $R_Q'$ where $p$ contains a vertical string. The set $e_C$ contains at most $\alpha d_{\mathrm{in}}$ layers. The string from $p$ in each $q\in e_C$ creates an excitation in each $X$-layer that intersects $q$ non-trivially. Therefore, the $X$-layers with an odd number of excitations is exactly the syndrome of $e_C$, viewed as an error of the input code.
\end{proof}

Next, we clean the remaining error onto the $X$-layers and analyze the resulting operator.

\begin{lemma}
	\label{lem:Qlayerstringsstabilizer}
	If $s_C$ is a $Z$ stabilizer of the input code $C$, then the operator $\ell$ on $\mathscr L(C)$ consisting of a vertical string on all $Q$-layers corresponding to $s_C$ is stabilizer-equivalent to an operator supported on $R_X$.
\end{lemma}
\begin{proof}
	Let $s_C = \sum_{g\in S} g$ for a set of $Z$ stabilizer generators $S\subseteq S_Z(C)$. Define $s$ to be the product of all stabilizers on the $Z$-layers corresponding to elements in $S$. Then $\ell + s$ consists of an even number of vertical strings on every $Q$-layer and has no support on $R_Z$. Pair up consecutive strings on the $Q$-layers, and let $s'$ be the product of the stabilizers between the paired up strings. Then $\ell + s + s'$ is supported on $R_X$.
\end{proof}

\begin{lemma}
	\label{lem:RXopisstab}
	Let $p$ be a Pauli $Z$ operator supported on $R_X$ with zero syndrome. Then $p$ is a stabilizer.
\end{lemma}
\begin{proof}
	Since the boundaries of $X$-layers do not condense $e$-excitations, $p$ must consist of closed loops. Thus, $p$ is the product of all stabilizers enclosed by these loops ($Z$-type stabilizers on $X$-layers do not have support on other layers).
\end{proof}

Combining the lemmas, we conclude that the decoder successfully corrects adversarial errors.

\begin{theorem}
	\label{thm:concatenateddecoder}
	Let $C$ be an \code{n,k,d_{in}} input CSS code with sparsity $w$ and let $\operatorname{Dec}_C^Z$ be a decoder that successfully corrects all $Z$ errors of $C$ up to weight $\alpha d_{\mathrm{in}}$ for some $\alpha>0$. Then the concatenated decoder in Algorithm~\ref{alg:concatenateddecoder} successfully corrects all $Z$ errors of the layer code $\mathscr L(C)$ of weight at most $\alpha d_{\mathrm{in}}L/(4w)$.
\end{theorem}
\begin{proof}
	Let $p = e_0 + \hat f_Z + s_Z + \hat f_Q^1 + s_Q$ be the operator defined in Lemma~\ref{lem:Qlayercleaning}. It follows by Lemma~\ref{lem:inputdecodervalid} and the correctness of the input decoder that $\operatorname{Dec}_C^Z$ returns a correction $\hat{f}_C$ which is stabilizer-equivalent to $e_C$. Lemma~\ref{lem:Qlayerstringsstabilizer} therefore implies that $p+\hat{f}_Q^2$ is stabilizer-equivalent, say through the stabilizer $s$, to an operator supported on $R_X$. By the correctness of $\operatorname{Dec}_C^Z$, it also follows that $p+\hat{f}_{Q}^2+s$ has an even number of excitations on every $X$-layer (the parity is well-defined since $X$-layers do not condense $e$-excitations) so that MWPM can be applied in Line~\ref{algstep:fXmatching} of Algorithm~\ref{alg:concatenateddecoder}. The final result $p + \hat f_Q^2 + \hat f_X + s$ is therefore an operator supported on $R_X$ with zero syndrome, so it is a stabilizer by Lemma~\ref{lem:RXopisstab}. This completes the proof.
\end{proof}

\section{Partial Self-Correction of Layer Codes with Quantum Tanner Code Input}\label{sec:appendix_c}

\newcommand{\mem}{{\rm mem}}
\newcommand{\bath}{{\rm bath}}
\newcommand{\inter}{{\rm int}}

\subsection{Definition of Partial Self-Correction}

Unlike active memories, which rely on frequent feedback-based error-correction loops, passive (self-correcting) quantum systems preserve information through the intrinsic dynamics of the physical system and its coupling to a thermal bath at sufficiently low temperature $T$. The joint unitary evolution of such a system with a thermal bath is generated by a Hamiltonian acting on the composite Hilbert space $\mathcal{H}_{\mem}\otimes \mathcal{H}_{\bath}$ of the form
\begin{align}
    H=H_{\mem}\otimes I_{\bath}+ I_{\mem}\otimes H_{\bath} + H_{\inter},\label{eq:memory_bath_hamiltonian}
\end{align}
where $H_{\mem}$ is the memory Hamiltonian, $H_{\bath}$ is a heat bath Hamiltonian, and $H_{\inter}$ is a memory-bath interaction term.

We will consider memory Hamiltonians which are constructed from the checks of a stabilizer code.
For a CSS code $C=(H_X,H_Z)$, let $S_X$ be the set of $X$-type stabilizer generators and $S_Z$ the $Z$-type stabilizer generators. Then the associated memory Hamiltonian is defined by
\begin{align}
    H_{\mem} = -\sum_{g_x\in S_X} g_x -\sum_{g_z\in S_Z} g_z . \label{eq:H_mem}
\end{align}
This Hamiltonian is designed so that the ground space is exactly the codespace of $C$.

At time $t=0$, we initialize the state of the memory in a ground state $\rho(0)$ of the Hamiltonian $H_{\mem}$. The thermal evolution generated by the joint Hamiltonian $H$ in Eq.~\eqref{eq:memory_bath_hamiltonian} drives the quantum state of the memory subsystem towards the canonical Gibbs state $\rho_{\beta}=e^{-\beta H_{\mem}}/\tr(e^{-\beta H_{\mem}})$, where $\beta=1/T$ is the inverse temperature (we set $k_B=1$).

For the study of self-correcting quantum memories, the dynamics generated by $H$ is typically modeled using the Davies weak-coupling limit~\cite{Davies}, which makes a few natural assumptions about the bath and the coupling term. The bath is assumed to be Markovian, and the interaction Hamiltonian is weakly-coupled $\|H_{\inter}\|\ll \|H_{\mem}\|$ and of the form
\begin{align}
H_{\inter} = \sum_{\alpha} A_\alpha\otimes B_\alpha,
\end{align}
where the jump operators $A_\alpha$ are $O(1)$-local, meaning they act on a constant number of qubits of the memory system. We assume without loss of generality that the jump operators $A_\alpha$ are Hermitian operators with operator norm $\|A_\alpha\|\le 1$. The Davies weak-coupling limit~\cite{Davies} leads to system dynamics governed by the Lindblad equation
\begin{align}
\dot{\rho}(t)=-i[H_{\mem}, \rho(t)]+\sum_{\alpha,\omega} h(\alpha,\omega)\left(A_{\alpha}(\omega)\rho A^\dagger_{\alpha}(\omega) - \frac12\{\rho,A_{\alpha}^\dagger(\omega) A_{\alpha}(\omega)\}\right),\label{eq:lindblad}
\end{align}
where $\omega$ are the Bohr frequencies and $A_{\alpha}(\omega)$ are Fourier components of $A_\alpha$. The coefficients $h(\alpha,\omega)$ capture the rates of these jumps and satisfy $\max_{\alpha,\omega} h(\alpha,\omega)=O(1)$. The rate function $h(\alpha, \omega)$ satisfies the detailed balance condition
\begin{align}
    h(\alpha,-\omega)=e^{-\beta \omega}h(\alpha,\omega),
\end{align}
which ensures that the canonical Gibbs state $\rho_{\beta}$ is a fixed-point of the dynamics. If, additionally, the Lindbladian satisfies certain ergodicity conditions, $\rho_{\beta}$ is the unique fix point of the dynamics~\cite{Spohn} and the system thermalizes $\lim_{t\to\infty}\rho(t)=\rho_\beta$.

Under this setup, we show that layer codes are partially self-correcting quantum memories (SCQM). 

\begin{remark}
While the scaling of the memory time with system size for a genuine self-correcting quantum system is well defined asymptotically, the breakdown of a partially self-correcting system at finite system size means that there is no precise way to talk about the functional scaling of the memory time; it is meaningless to say, for example, that a function ``scales exponentially'' when restricted to a finite interval. Due to this fact, we adopt a more technical definition of partial self-correction, as explained in the main text.
\end{remark}

\begin{definition}[Partially Self-Correcting Quantum System and Memory]\label{def:pscqm}
Fix a constant $\epsilon \in (0,1)$. Let $\{C_n\}$ be a family of quantum codes. We define the memory time $t_{\mathrm{mem}}(n,\beta)$ as the maximum time such that, for all $t<t_{\mathrm{mem}}$, there exists some decoder $\Phi$ satisfying
\begin{align}
\left\|\Phi\left(\rho^{(n)}(t)\right) - \rho^{(n)}(0)\right\|_1 \le \epsilon
\end{align}
for all initial states $\rho^{(n)}(0)$ in the codespace $C_n$, where $\rho^{(n)}(t)$ is the time-evolved state under the Lindblad equation~\eqref{eq:lindblad} with bath temperature $T=1/\beta$. Let $t^*_{\mathrm{mem}}(\beta)$ denote the maximum value of $t_{\mathrm{mem}}(n,\beta)$ over $n$. Then we say that $\{C_n\}$ is a \emph{partially self-correcting quantum system} if $t^*_{\mathrm{mem}}(\beta)$ scales superexponentially with $\beta$, i.e., $t^*_{\mathrm{mem}}(\beta)=\exp(\omega(\beta))$. Furthermore, we say that $\{C_n\}$ is a partially self-correcting quantum \emph{memory} (SCQM) if there exists an efficient, i.e., $\mathrm{poly}(n)$ runtime, decoder $\Phi$ with respect to which the system is partially self-correcting.
\end{definition}

The definition of partial self-correction is therefore associated with the scaling behavior of the optimal memory time as a function of $\beta$, which we require to grow sufficiently quickly to rule out trivial systems.
For the purposes of proving partial self-correction, we often consider a cutoff system size $n^*(\beta)$, which provides a lower bound $t^*_{\mathrm{mem}}(\beta) \ge t_{\mathrm{mem}}(n^*(\beta), \beta)$.
Note that we also distinguish partially self-correcting \emph{systems} and partially self-correcting \emph{memories}. The former is an intrinsic statement about the memory time scaling of a quantum system, whereas the latter has the additional requirement of admitting an efficient decoder with respect to which it is self-correcting.

In the remainder of this appendix, we prove that layer codes instantiated with quantum Tanner codes form a partially self-correcting memory.

In Appendix~\ref{sec:partial_self_corr_HDPC} we prove an analogous result for random layer codes.

\subsection{Proof of Partial Self-Correction}

We will make use of the following energy barrier preserving map of Pauli operators between a layer code and its input code; see the ``Energy barrier'' section of Ref.~\cite{williamson2023layer} for more details.

\begin{lemma}[\cite{williamson2023layer}]\label{lem:energy_map}
Let $\mathscr{L}(C)$ be a layer code. There exists a map $P\mapsto P'$ taking a Pauli-$Z$ operator $P\in P_Z(\mathscr{L}(C))$ of the layer code to a Pauli-$Z$ operator $P'\in P_Z(C)$ of its input code such that 
\begin{align}
2\Delta(P') \le w^2\Delta(P),
\end{align}
where $w$ is the maximum sparsity of the input code $C$. Furthermore, if $P$ has no support on $Z$- or $Q$- layers, then $P'$ is stabilizer-equivalent to any of the $e$-configurations defined by $P$.
\end{lemma}

Lemma~\ref{lem:energy_map} allows us to control the energy barrier of corresponding operators of the input code as we decode the layer code. The concatenated decoder does not increase the syndrome weight significantly, so if we start with an error of the layer code with low energy barrier, the input code decoder also receives an error of low energy barrier. Such an error can be corrected, leading to successful decoding of the layer code.

\begin{lemma}\label{lem:ldpc_decoder_barrier}
Let $\mathscr{L}$ be the layer code family defined with quantum Tanner code input $C$, which has parameters \code{\Theta(L^3),\Theta(L),\Theta(L^2)}. There exists a constant $c>0$ such that all errors with energy barrier $\Delta(e) \le cL$ are correctable by the concatenated decoder (using the potential function decoder of the input code~\cite{gu2022efficient}).
\end{lemma}
\begin{proof}
Let $w = O(1)$ be the maximum sparsity of the quantum Tanner code. Let $e_0$ be an initial error of energy barrier $\Delta(e_0)$. We may without loss of generality assume that $e$ is an $Z$-type error, with $X$-type errors treated analogously. Suppose that $e_0$ has $\delta_0$ total excitations. Note that $\delta_0 \le \Delta(e_0)$ by assumption.

Running the concatenated decoder, we can track the evolution and generation of excitations throughout the decoding process and upper bound their weight:
\begin{enumerate}
\item In the first phase of the concatenated decoder the excitations on the $Z$-layers are condensed using MWPM, possibly producing additional excitations on the $Q$ and $X$-layers. Note that the matching string produced by MWPM consists of a single vertical segment and a single horizontal segment. Each $Z$-layer has non-trivial intersection with at most $w$ $Q$-layers and with at most $w^2$ $X$-layers.\footnote{For a given $Z$-layer, each non-trivial $XZ$ defect line must begin at one of the $w$ intersecting $Q$-layers, and each $Q$-layer supports at most $w$ intersecting $X$-layers. It follows that each $Z$-layer supports at most $w^2$ distinct non-trivial $XZ$ line defects.} It follows that each excitation needs to cross at most $w+w^2\le 2w^2$ defect lines during the matching process. The maximal syndrome weight after the completion of the first phase is therefore $\delta_1 \le 2w^2\delta_0$.

\item The second phase proceeds like the first, with the matching done on $Q$-layers. Again, additional excitations may be produced on the $X$-layers. Since every $Q$-layer intersects at most $w$ $X$-layers non-trivially, each excitation needs to cross at most $w$ defect lines during the matching process. The maximal syndrome weight after the completion of the second phase is then $\delta_2 \le w\delta_1 \le 2w^3\delta_0$.
\end{enumerate}

In the third phase, we run the decoder of the input quantum Tanner code to obtain a correction $\hat{f}$. The correctness of this step relies on the success of the decoder for the quantum Tanner code. For the correctness of the input decoder, we must ensure that the input syndrome to the quantum Tanner code decoder is associated with a correctable error, i.e., we must obtain an energy-barrier version of Lemma~\ref{lem:inputdecodervalid}. We will show that the input syndrome is associated with a correctable error of the quantum Tanner code, provided that $\Delta(e_0) \le cL$ for some sufficiently small constant $c > 0$. 

Let $e_2$ denote the error obtained from $e_0$ after the completion of the first two phases of the concatenated decoder. By Lemma~\ref{lem:Qlayercleaning}, $e_2$ is stabilizer-equivalent to an operator supported on $X$-layers as well as vertical strings on a subset of $Q$-layers. From the map defined in Lemma~\ref{lem:energy_map}, this subset, when considered as a Pauli-$Z$ error of $C$, is stabilizer-equivalent to some $e_2'$ with $2\Delta(e_2')\le w^2\Delta(e_2)$. Note that since the excitations of $e_2$ are supported solely on the $X$-layers, the syndrome of $e_2'$ in the quantum Tanner code is precisely the input to the concatenated decoder. It suffices to prove that $e_2'$ is correctable for the quantum Tanner code decoder if $\Delta(e_0) \le cL$. From Lemma~\ref{lem:energy_map}, we find
\begin{align}
\Delta(e_2') \le \frac{w^2}{2}\Delta(e_2) \le \frac{w^2}{2}(\Delta(e_0)+\delta_2) \le 2w^5\Delta(e_0),  
\end{align}
where the penultimate inequality relates $\Delta(e_2)$ and $\Delta(e_0)$ by the fact that at most $\delta_2$ excitations are produced on the path between them. By the soundness of the quantum Tanner code (see Corollary 14 of Ref.~\cite{gu2022efficient}), there exists a constant $c'>0$ such that $\Delta(e'_2)\le c'L$ implies the correctability of $e_2'$ (indeed, any Pauli path ending on $e_2'$ must then consist of correctable operators). Taking $c = c'/(2w^5)$ gives the desired constant.

Since the concatenated decoder is guaranteed to succeed provided the input decoder succeeds, this implies that the concatenated decoder will successfully correct all errors of energy barrier $\Delta(e_0)\le cL$, as required.
\end{proof}

We now combine Lemma~\ref{lem:ldpc_decoder_barrier} with Lemma 1 and Eq.~(26) in Ref.~\cite{bravyi2011analytic} to prove that the layer codes $\scrL$ form partially SCQM. 

\begin{lemma}[Lemma 1 of Ref.~\cite{bravyi2011analytic}]\label{lem:BH}
Let $C$ be a qLDPC stabilizer code with parameters \code{n,k}. Let $f$ be the maximum energy barrier of Pauli operators that appear in the expansion of the jump operators $A_\alpha(\omega)$ and $A^\dagger_\alpha(\omega)A_\alpha(\omega)$. Note that an $O(1)$-local Lindbladian implies that $f=O(1)$. Suppose a decoder $\Phi$ corrects any Pauli error $e$ with energy barrier $\Delta(e) \le m+2f$, where $m$ is an arbitrary energy cutoff\footnote{In the proof of Theorem~\ref{thm:partial_self_corr} of the main text, $f$ was redefined to absorb the factor $2$.}. Then for all $a\in (0,1)$, we have
\begin{align}
\|\Phi(\rho(t)) - \rho(0)\|_1 \le O(t)2^kne^{-am\beta} \label{eq:BHtracedistbound}
\end{align}
for all $n\le e^{(1-a)\beta}$ and all initial state $\rho(0)$ in the codespace.
\end{lemma}

\begin{corollary}\label{cor:partial_self_corr_decoder}
    Let $C$, $\Phi$, and $m$ be defined as in Lemma~\ref{lem:BH}. If $m=\Omega(\max(k,\log n))$, then at cutoff size $n^* = \exp(\beta/2)$ with energy cutoff $m^*$, maximum memory time $t^*_\mathrm{mem} = \exp[\Omega(m^*\beta)]$ can be achieved with the decoder $\Phi$.
\end{corollary}
\begin{proof}
    The proof is identical to that of Theorem~3, replacing the energy barrier decoder with $\Phi$.
\end{proof}

\begin{theorem}[Partially SCQM of quantum Tanner layer codes]\label{thm:psc_ldpc_layer}
    Let $\mathscr{L}$ be the layer code family defined with quantum Tanner code input, which has parameters \code{\Theta(L^3),\Theta(L),\Theta(L^2)}. Let $\Phi$ be the concatenated decoder. Then $\scrL$ with the decoder $\Phi$ is a partially SCQM. In particular, maximum memory time $t^*_\mathrm{mem} = \exp(\exp(\Omega(\beta)))$ can be achieved with cutoff length $L^*=\exp(\Theta(\beta))$.
\end{theorem}
\begin{proof}
    Lemma~\ref{lem:ldpc_decoder_barrier} shows that the $\Phi$ can correct errors of energy barrier up to $m=\Theta(L)$. At cutoff size $\exp(\beta/2)$, the linear system size is $L^*=\exp(\Theta(\beta))$ and the energy cutoff is $m^*=\Theta(L^*)=\exp(\Theta(\beta))$. The result then follows from Corollary~\ref{cor:partial_self_corr_decoder}.
\end{proof}

\begin{remark}\label{rem:Lscaling}
    Although scaling is not rigorously defined for finite system sizes, we may consider the memory time for $L\le L^* = \exp(\Theta(\beta))$. To achieve constant error, Lemma~\ref{lem:BH}, combined with Lemma~\ref{lem:ldpc_decoder_barrier} and the parameters of $\scrL(C)$, gives the bound
    \begin{equation}
        t_{\mathrm{mem}} \ge \exp(a\beta(cL - 2f) - \Theta(L)\log 2 - \Theta(\log L)).
    \end{equation}
    Therefore, for sufficiently large $\beta$ and $L$ (but still satisfying $L<L^*$), we conclude that $t_{\mathrm{mem}} = \exp(\Omega(\beta L))$.
\end{remark}

\section{Random Layer Codes}\label{sec:appendix_d}

Theorem~\ref{thm:layer_codes} implies that the layer code construction applied to qLDPC inputs provides outputs with favorable parameters. In this section, we establish analogous results for \emph{random} CSS inputs. We first formalize what we mean by a random CSS code.

\begin{definition}[Random CSS Code]
A random CSS code with block length $n \in \mathbb{N}$ and rate parameters $\rho_X,\rho_Z \in (0,1/2)$ is defined by first choosing a uniformly random $\rho_Zn \times n$ binary matrix $H_Z \in \mathbb{F}_2^{\rho_Z n\times n}$, and then a uniformly random $\rho_Xn\times n$ binary matrix $H_X \in \mathbb{F}_2^{\rho_Xn\times n}$ from the subspace of matrices orthogonal to $H_Z$, i.e., satisfying the condition $H_XH_Z^\mathrm{T}=0$. We denote the ensemble of random CSS codes defined this way by $\mathrm{CSS}_n(\rho_X,\rho_Z)$.
\end{definition}

In the remainder of this appendix we assume without loss of generality that every $Q$-layer intersects at least one $Z$-layer and one $X$-layer non-trivially. The probability of sampling some $(H_X,H_Z)$ for which this condition fails to hold is $2^{-\Omega(n)}$, which is negligible in the limit of large $n$ and can be excluded from the probabilistic arguments that follow.

Note that a random CSS code is equivalently defined by choosing, uniformly at random, a pair of binary matrices $(H_X,H_Z)$ satisfying $H_XH_Z^\mathrm{T}=0$, so there is no asymmetry in the definition. A random CSS code has maximum check weight $\Theta(n)$ with high probability, so the distance and energy barrier guarantee of Theorem~\ref{thm:layer_codes} are reduced by a factor of $n$ and $n^2$, respectively, rendering them trivial. Nevertheless, we show that the parameter scaling of Theorem~\ref{thm:layer_codes} still hold when using a random CSS input code with high probability, up to a logarithmic reduction.

\begin{theorem}[Random Layer Codes]\label{thm:random_layer_appx}
Let $\scrL(C)$ be a layer code with random CSS input $C \sim \mathrm{CSS}_n(\rho_X,\rho_Z)$. Then there exists a constant $c_0>0$ such that $\scrL(C)$ has distance $d\ge c_0n^2/\log n$ and energy barrier $\Delta_{\scrL(C)} \ge c_0n/\log n$ with probability $1 - 2^{-\Omega(n)}$ as $n\rightarrow \infty$.
\end{theorem}

We will call a layer code with random input a \emph{random layer code}. Note that a random layer code has block length $N=\Theta(n^3)$, so Theorem~\ref{thm:random_layer_appx} implies that a random layer code has parameters \code{\Theta(n^3),\Theta(n),\Omega(n^2/\log n)} and energy barrier $\Delta_{\scrL(C)}=\Omega(n/\log n)$ with high probability. This matches the results of Theorem~\ref{thm:layer_codes} up to logarithmic factors.

We separate the proof of Theorem~\ref{thm:random_layer_appx} into two parts. In Section~\ref{sec:layer_distance} we establish the $\Theta(n^2/\log n)$ distance property, culminating in Theorem~\ref{thm:layer_distance}. In Section~\ref{sec:layer_energy} we establish the $\Theta(n/\log n)$ energy barrier in Theorem~\ref{thm:layer_energy}. Together, Theorems~\ref{thm:layer_distance} and \ref{thm:layer_energy} imply Theorem~\ref{thm:random_layer_appx}.

\subsection{Distance of Random Layer Codes}\label{sec:layer_distance}

In this section, we prove that random layer codes have distance $d=\Omega(n^2/\log n)$ with high probability.

\begin{definition} \label{def:YC}
Let $C = (H_X,H_Z)$ be a CSS code of length $n$. Let $S_Z$ be the set of $Z$-type stabilizer generators for $C$, i.e., the rows of $H_Z$. For each $i \in [n]$, let $Y_C^i$ denote the set of stabilizer generators with all entries of index greater than $i$ set to zero, i.e., 
\begin{align}
Y_{C}^i \coloneq \{y \in \mathbb{F}_2^n \mid  \exists s\in S_Z:\ (\forall j \le i:\  y_j = s_j\ \text{ and }\ \forall j > i:\  y_j = 0)\}.
\end{align}
Let $Y_C$ be the union of the $Y_C^i$ for all $i\in [n]$.
We define the \emph{$Y$-weight} of a string $w\in \mathbb{F}_2^n$ as a modification of the Hamming weight with respect to the spanning set $Y_C$:
\begin{align}
|w|_{Y} \coloneq \min\left\{|S|: S\subseteq Y_C \text{ and } x = \sum_{s\in S} s\right\},
\end{align}
For $f \in (0,1)$, let $B_C(f) \subseteq \mathbb{F}^n_2$ denote the set of all elements having $Y$-weight at most $fn/\log n$, i.e.,
\begin{align}
    B_C(f)\coloneq \{w\in \mathbb{F}^n_2 : |w|_Y\le fn/\log n\} .
\end{align}
\end{definition}

Note that the sets $Y_C$ and $B_C(f)$ depends on $C$ only through $H_Z$. We will write $Y_{H_Z}$ and $B_{H_Z}(f)$ when we want to emphasize independence from $H_X$. For motivation on the definition of $Y_C$, see Fig.~\ref{fig:distance1} and the accompanying caption.

\begin{figure}[H]
	\centering
	\includegraphics[width=.8\linewidth]{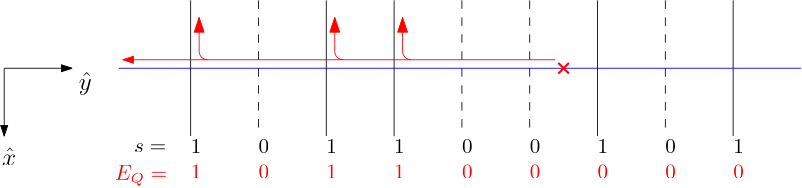}
	\caption{
    A top-down view of converting an elementary $e$-configuration $E_Z\in\mc E_Z$ consisting of a single excitation to an equivalent $e$-configuration $E_Q\in\mc E_Q$. The horizontal line represents the $Z$-layer with associated $Z$ stabilizer $s\in S_Z$ viewed from above, and solid (dashed) vertical lines depict the $Q$-layers which intersect non-trivially (trivially) with the the $Z$-layer.    
    By the fusion rules of $e$ excitations, $E_Q$ can be obtained by pushing the excitation to the left until it condenses on the boundary, where it will branch at each intersecting $Q$-layer. Thus, $E_Q$ is precisely $s$ with all entries of index greater than $i$ set to zero. The elements of $Y_C$ therefore represent elementary vectors in $\mc E_Z$ when converted to equivalent $e$-configurations in $\mc E_Q$. 
    }
\label{fig:distance1}
\end{figure}

\begin{lemma}\label{lem:f_bound}
Let $C \sim \mathrm{CSS}_n(\rho_X,\rho_Z)$ be a random CSS code. For any constant $f\in(0,1)$, we have
\begin{align}
\Pr_{C}\left(B_C(f) \cap L_Z(C) \neq \emptyset\right) \le 2^{-n(\rho_X-6f)}
\end{align}
for all sufficiently large $n$.
\end{lemma}
\begin{proof}
First, consider a fixed $H_Z$ defining the $Z$-checks. There are a total of $\rho_Zn$ stabilizer generators, so the size of $Y_{H_Z}$ is bounded above by $\rho_Zn^2$. The size of $B_{H_Z}(f)$ is therefore bounded above by
\begin{align}
|B_{H_Z}(f)| \le \sum_{k=0}^{fn/\log n}\binom{\rho_Zn^2}{k} \le 2^{\rho_Zn^2H_2\left(\frac{f}{\rho_Zn\log n}\right)},
\end{align}
where $H_2(p) = -p\log_2(p) - (1-p)\log_2(1-p)$ is the binary entropy function. We can further bound the binary entropy function by
\begin{align}
H_2\left(\frac{f}{\rho_Zn\log n}\right) &\le \frac{f}{\rho_Zn\log(2)\log n}\left(1-\log\left(\frac{f}{\rho_Zn\log n}\right)\right)\\
&\le \frac{f}{\rho_Zn\log(2)\log n}\log\left(\frac{e\rho_Zn\log n}{f}\right)\\
&\le \frac{6f}{\rho_Zn},
\end{align}
where the last inequality holds for $n\log n \ge e\rho_Z/f$. Therefore it follows that
\begin{align}
|B_{H_Z}(f)| \le 2^{\rho_Zn^2H_2\left(\frac{f}{\rho_Zn\log n}\right)} \le 2^{6fn}
\end{align}
for sufficiently large $n$.

Now let us sample $H_X$ uniformly at random from the space of matrices orthogonal to $H_Z$ and let $C=(H_X,H_Z)$ denote the resulting CSS code. Then $\ker(H_X)$ is uniformly random among the subspaces of dimension $(1-\rho_X)n$ containing $\row(H_Z)$. It follows that for any fixed vector $z\in \mathbb{F}^n$, we either have $z \in L_Z(C)$ with probability zero if $z \in \row(H_Z)$, or else
\begin{align}
\Pr_{H_X}(z \in L_Z(C) \mid H_Z) &= \frac{|\text{\{subspaces of dim }(1-\rho_X)n \text{ containing }z \text{ and }\row(H_Z)\}|}{|\text{\{subspaces of dim }(1-\rho_X)n \text{ containing }\row(H_Z)\}|}\\
&=\frac{\binom{(1-\rho_Z)n-1}{(1-\rho_X-\rho_Z)n-1}_2}{\binom{(1-\rho_Z)n}{(1-\rho_X-\rho_X)n}_2}\\
&=\frac{2^{(1-\rho_X-\rho_Z)n}-1}{2^{(1-\rho_Z)n}-1}
\end{align}
for $z\in\mathbb{F}^n\backslash \mathrm{row}(H_Z)$, where the second equality is the $2$-binomial coefficient. Taking a union bound, we find
\begin{align}
\Pr_{H_X}(B_{H_Z}(f) \cap L_Z(C) \neq \emptyset\mid H_Z) &\le \sum_{z \in B_{H_Z}(f)}\Pr_{H_X}(z \in L_Z(C) \mid H_Z)\\
&\le |B_{H_Z}(f)|\cdot \frac{2^{(1-\rho_X-\rho_Z)n}-1}{2^{(1-\rho_Z)n}-1}\\
&\le 2^{6fn}\cdot \frac{2^{(1-\rho_X-\rho_Z)n}-1}{2^{(1-\rho_Z)n}-1}\\
&\le 2^{-n(\rho_X-6f)}.
\end{align}
Since this bounds holds independently of the choice of $H_Z$, it follows that
\begin{align}
\Pr_{C}(B_C(f) \cap L_Z(C)\neq\emptyset) \le 2^{-n(\rho_X-6f)}.
\end{align}
\end{proof}

\begin{theorem}[Random Layer Code Distance]\label{thm:layer_distance}
Let $\mathscr{L}(C)$ be a layer code with random CSS input $C \sim \mathrm{CSS}_n(\rho_X,\rho_Z)$. There exists a constant $c_0 > 0$ such that $\mathscr{L}(C)$ has distance $d \ge c_0n^2/\log n$ with probability $1 - 2^{-\Omega(n)}$ as $n\rightarrow \infty$.  
\end{theorem}
\begin{proof}
We will prove that the $Z$-distance of $\mathscr{L}(C)$ satisfies $d_Z \ge c_0n^2/\log n$ with high probability. A completely analogous argument will then imply that the $X$-distance satisfies the same bound with high probability, and the full result then follows by the union bound.\\

First, consider a fixed input code $C$ of length $n$. Let $\overline{Z} \in L_Z(\mathscr{L}(C))$ be a minimal weight logical $Z$ operator for $\mathscr{L}(C)$. For any slab $A_i$, let $\overline{Z}_i$ denote the restriction of $\overline{Z}$ to $A_i$. Then we have
\begin{align}
\left|\overline{Z}_i\right| \ge \left|E_i(\overline{Z})\right|.
\end{align}
Applying the fusion rules to convert any excitation on a $Q$-layer into two excitations on the adjacent $Z$-layer segments for an intersecting $Z$-layer, it follows that there exists $E' \in \cE_Z$ such that $E_i(\overline{Z}) \approx E'$ and $|E'| \le 2|E_i(\overline{Z})|$. Applying the fusion rules one more time to move any excitations on a $Z$-layer segment onto all intersecting $Q$-layers to its left, it follows from Lemma~\ref{lem:layer_input_map} and Fig.~\ref{fig:distance1} that there exists some $E'' \in \cE_Q$ such that:
\begin{enumerate}
\item $P(E'') \in L_Z(C)$,
\item $P(E'')$ is a linear combination of at most $|E'|$ elements of $Y_C$, i.e., $P(E'') \in B_C(|E'|\log n/n)$.
\end{enumerate}
In particular, it follows that
\begin{align}
B_C(|E'|\log n/n) \cap L_Z(C) \neq \emptyset.
\end{align}
Let us now bound the probability over $C$ that there exists some $\overline{Z} \in L_Z(\mathscr{L}(C))$ such that $|\overline{Z}_i| \le fn/\log n$ for some constant $f$ and index $i \in [\rho_Xn]$. Since $|\overline{Z}_i| \le fn/\log n$ implies that $|E'| \le 2fn/\log n$, it follows from the discussion above and Lemma~\ref{lem:f_bound} that
\begin{align}
\Pr_C(\exists \overline{Z} :\ |\overline{Z}_i| \le fn/\log n \text{ for some }i) \le \Pr_C(B_C(2f) \cap L_Z(C)\neq \emptyset) \le 2^{-n(\rho_X-12f)} 
\end{align}
for sufficiently large $n$. For any fixed $\rho_X$, we may choose $f$ such that $12f < \rho_X$. For such a choice of $f$, it follows that we will have $|\overline{Z}_i| > fn/\log n$ for all $i\in [\rho_Xn]$ with probability $1-2^{-\Omega(n)}$. Therefore we have
\begin{align}
|\overline{Z}| = \sum_{i = 1}^{\rho_Xn}|\overline{Z}_i| > \rho_Xfn^2/\log n \label{eq:distlowerboundlogicalZ}
\end{align}
for all $\overline{Z} \in L_Z(\mathscr{L}(C))$ with probability $1-2^{-\Omega(n)}$. Taking $c_0 = \rho_Xf$, it follows that $d_Z > c_0n^2/\log n$ with probability $1-2^{-\Omega(n)}$.
\end{proof}

\subsection{Energy Barrier of Random Layer Codes}\label{sec:layer_energy}

In this section, we prove that random layer codes have energy barrier (see Definition~\ref{def:energy_barrier}) $\Delta_{\scrL(C)} = \Omega(n/\log n)$ with high probability.

It will be convenient throughout this section and the next to work with $Z$-type Pauli operators that have been put into a canonical form. Let $P\in P_Z(\scrL(C))$ be a Pauli-$Z$ operator on the layer code $\scrL(C)$. By applying string operators, we will define an associated operator $\tilde{P} \in P_Z(\scrL(C))$. The precise definition of $\tilde P$ and its properties are stated in the following lemma and accompanying figures. 

\begin{lemma}[Mapping Lemma]\label{lem:mapping}
Let $C = (H_X,H_Z)$ be a CSS code and let $\scrL(C)$ be its layer code. By an application of string operators, we can associate to each $P\in P_Z(\scrL(C))$ another operator $\tilde{P} \in P_Z(\scrL(C))$ such that:
\begin{enumerate}
\item $\tilde{P}$ is supported entirely on $X$ and $Z$ layers,
\item\label{itm:mapping2} $\tilde{P}$ has the same excitations as $P$ on $X$ layers,
\item Restricted to the $Z$ layers, $\tilde{P}$ is a union of vertical strings located half-way between successive $Q$-layers. Each region between successive $Q$-layers contains at most one excitation---located at its center---and the vertical strings connect successive excitations. More explicitly, in each region, the vertical strings either: 
\begin{enumerate}
\item connects the top $X$-layer of the region to the bottom $X$-layer,
\item connects the top $X$-layer to an excitation located at the center of the region, or
\item connects the bottom $X$-layer to an excitation located at the center of the region.
\end{enumerate}
\item\label{itm:mapping4} The restriction of $\tilde{P}$ to $Z$-layers is a function of only $P$ and $H_Z$. In particular, it is independent of $H_X$, i.e., as long as the number of $X$ layers ($\rho_Xn$) is the same, the matrix $H_X$ does not matter. This will be important in the argument to follow. Note that the parts of $\tilde{P}$ on the $X$ layers, including its set of excitations, \emph{will} generally depend on $H_X$.
\item \label{itm:mapping5} The syndrome weights are related by $|\sigma(\tilde{P})| \le 2|\sigma(P)|$.
\item\label{itm:mapping6} If $P$ is a logical operator, then $\tilde{P}$ is stabilizer-equivalent to $P$.
\item\label{itm:mapping7} If $|P_0P_1| = 1$, then the $e$-configurations defined by $\tilde{P}_0$ and $\tilde{P}_1$ differ on at most a single slab boundary. 
\end{enumerate}
\end{lemma}

\begin{proof}
\textbf{Definition of $\tilde{P}$.}

For each $Q$-layer, fix some $Z$-layer which has non-trivial intersection with $Q$. We construct $\tilde{P}$ in three stages, which we call cleaning, branching, and merging: 
\begin{enumerate}
\item In the cleaning stage (Fig.~\ref{fig:cleaning}), we use string operators and stabilizers to move all excitations on each $Q$-layer horizontally to the defect line defined by the chosen $Z$-layer. This process does not modify the excitations on the $X$-layers, although the support of the operator on the $X$-layers may be modified by the addition of stabilizers. Any string operators remaining on the $Q$-layers can be cleaned away by stabilizers so as to be supported solely on the chosen $QZ$ line defect. This leaves us with an operator whose $Q$-support consists entirely of vertical string segments supported on the chosen $QZ$ line defects. This process may merge some excitations on the $Q$-layers but will not increase the number.
\begin{figure}[H]
	\centering
	\includegraphics[width=.8\linewidth]{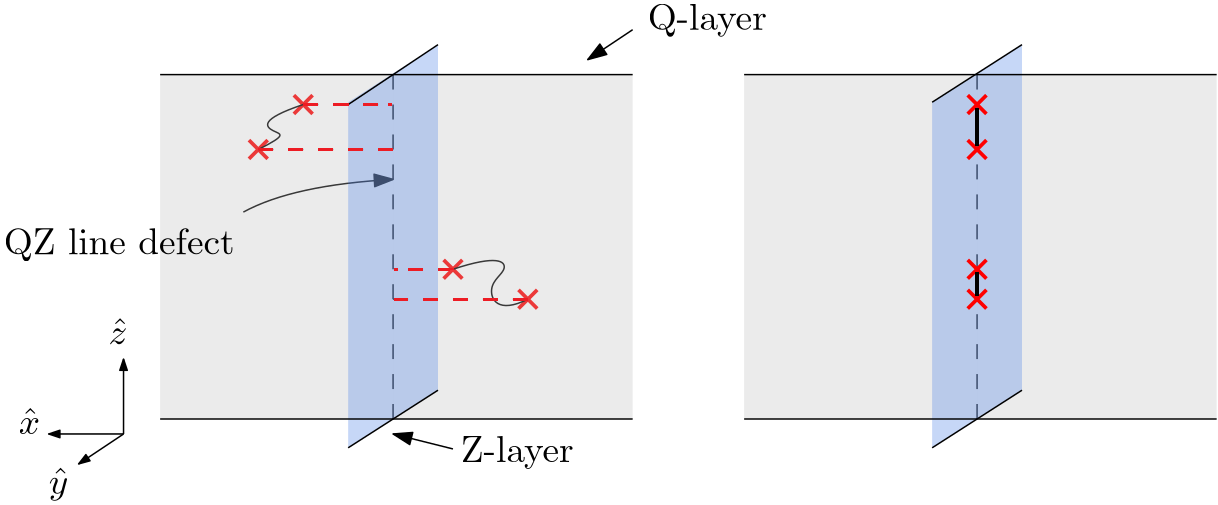}
	\caption{Note that this figure is rotated with respect to the other figures. Illustration of the ``cleaning'' step of the map $P\mapsto \tilde{P}$. In the left figure, a number of strings and excitations are extant on the $Q$-layer. The excitations can be moved by a horizontal string (red dashed) to the chosen $QZ$ line defect (dashed black). The right figure illustrates the end result of the cleaning step. The remaining operator on the $Q$-layer, after cleaning away any closed loops, consists of vertical strings supported on the $QZ$ line defect.
    }
	\label{fig:cleaning}
\end{figure}

\item In the branching stage (Fig.~\ref{fig:branchingproof}), we push the vertical strings on the $QZ$-defect from the $Q$-layer onto the $Z$-layers using the fusion rules. This will create two vertical string segments for each original one. We then move all strings on $Z$-layers until they are vertical and half-way in between successive $Q$-layers. Again, this step can be accomplished entirely using horizontal strings and stabilizers, which does not affect the excitations on the $X$-layers. Any excitations and string segments on the first and last boundary $Z$-layer regions can be simply condensed away at the boundaries. An example of both cleaning and branching steps applied to an operator is depicted in Fig.~\ref{fig:cleaning_branchingproof}.

\begin{figure}[H]
\centering
	\includegraphics[width=.9\linewidth]{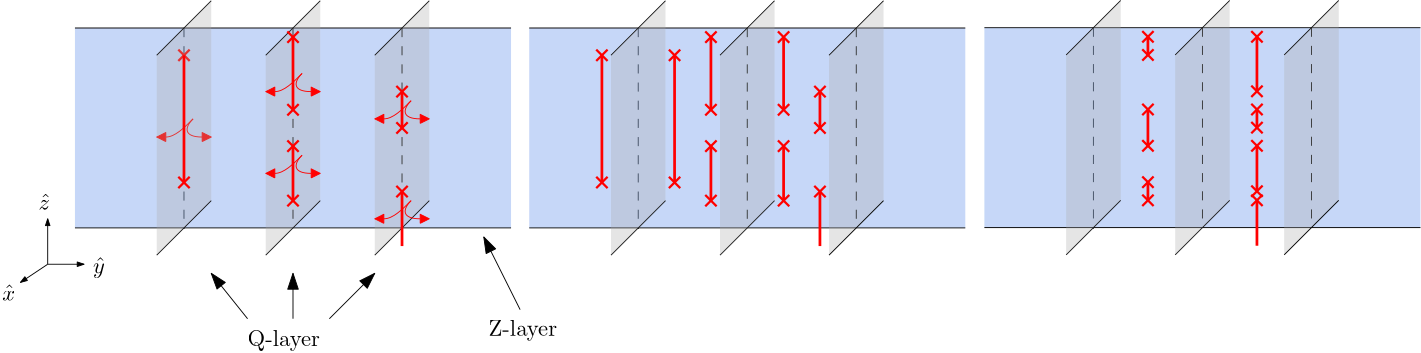}
	\caption{Illustration of the ``branching'' step of the map $P\mapsto \tilde{P}$. The first (left most) figure illustrates the end result of the ``cleaning'' step, a number of vertical string operators supported on $Q$-layers, lying on top of the $QZ$-defect. We can push these operators from the $Q$-layer onto the $Z$-layer using the fusion rules of the layer code. The result (middle figure) splits each vertical string into a pair of identical strings supported on the $Z$-layer, on either side of the defect. Finally, as illustrated on the last (right most) figure, we can move all of the branched strings to be supported in the middle of each region. This can be done solely through multiplication by horizontal strings or stabilizers. 
    }
	\label{fig:branchingproof}
\end{figure}

\begin{figure}[ht]
	\centering
	\includegraphics[width=.95\linewidth]{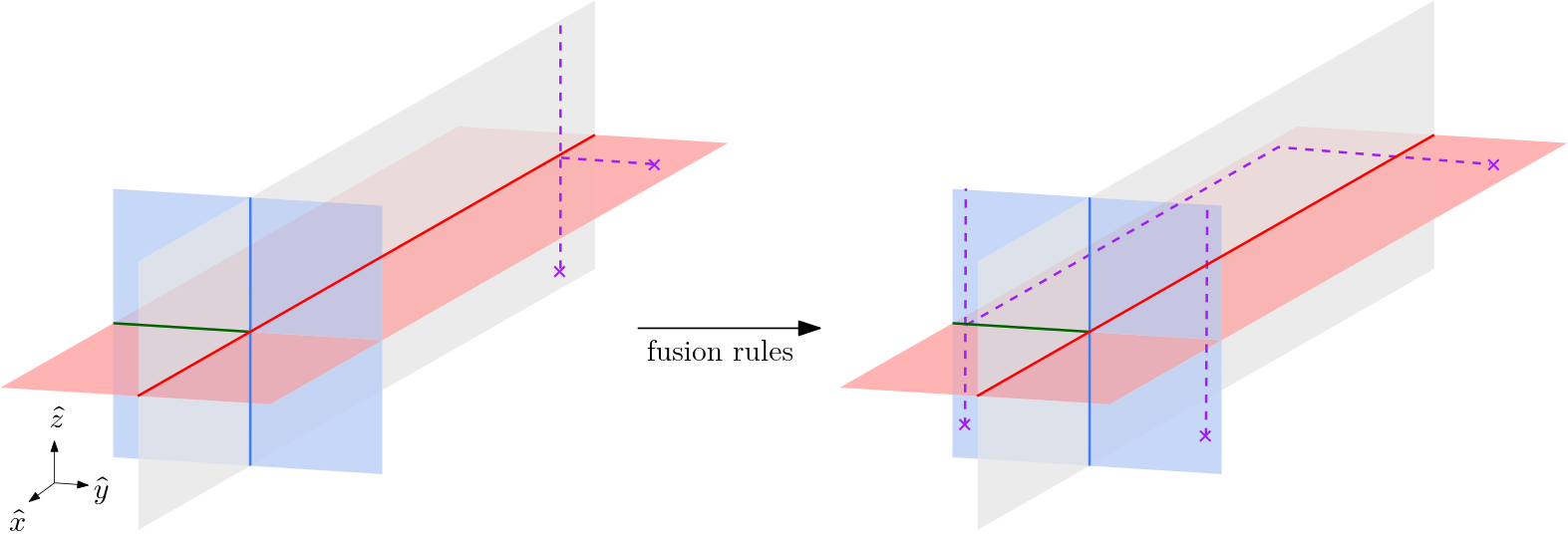}
	\caption{An illustration of the cleaning and branching steps together. The left figure illustrates an initial operator $P$ (dashed purple line) with its associated syndrome (purple crosses). The right figure illustrates the end result of the cleaning and branching steps. 
    The original operator can be transformed through multiplication by horizontal string operators and the fusion rules to the new operator.
    The vertical string that was originally on the $Q$-layer now branches into two segments on the $Z$-layer according to the fusion rules. 
    The string attached to the excitation on the $X$-layer must be carried along; for the illustrated junction, the $XZ$ defect line is to the left of the $Q$-layer, so the string attaches to the $Z$-layer on the left.
    }
	\label{fig:cleaning_branchingproof}
\end{figure}

\item Finally, in the merging stage (Fig.~\ref{fig:merging}), we combine the excitations on the $Z$-layers pairwise within each region so that each region contains at most one excitation. We can then translate the remaining excitation---if present---vertically until it is located at the center of the region. This step does not modify the support on the $X$-layers at all.

\begin{figure}[H]
	\centering
	\includegraphics[width=.9\linewidth]{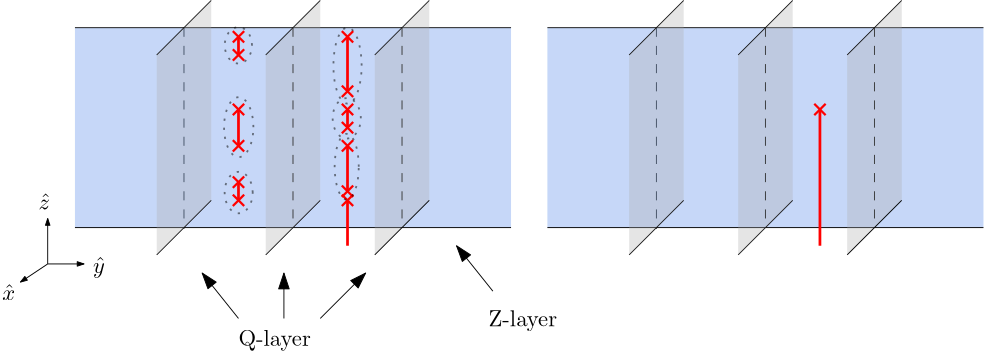}
	\caption{Illustration of the final ``merging'' step of the map $P\mapsto \tilde{P}$. The figure on the left illustrates the end result of the previous ``branching'' step where a number of vertical strings and excitations are located in the middle of each region on the $Z$-layer. Excitations within each region can be paired off and eliminated. This leaves at most one excitation remaining in each region, which can be shifted to the very center of the region (the right figure).
    }
	\label{fig:merging}
\end{figure}
\end{enumerate}

\textbf{Proof of the required properties.}

\begin{enumerate}
\item By the end of the cleaning and branching stages, the operator is already supported solely on $X$ and $Z$ layers. The merging stage does not change this fact.

\item The cleaning and branching stages can potentially change the existing string operators on the $X$-layers. The merging stage does not touch the $X$-layers at all. Note that $\tilde{P}$ and $P$ will differ on the $X$-layers in general; only the excitations (the endpoints of the string operators) remain unchanged.

\item This is essentially by definition of the mapping $P\mapsto \tilde{P}$.

\item The support on the $Z$-layers is entirely a function of the fusion rules between $Z$- and $Q$-layers, which does not depend on the choice of $H_X$. Note however, the mapping itself, and in particular the strings on the $X$-layers and where/whether they attach to the vertical strings on the $Z$-layers \emph{will} depend on the choice of $H_X$. The set of excitations on the $X$-layers will also depend on the choice of $H_X$.

\item The cleaning and merging stages can only decrease the number of excitations. The branching stage at most doubles the number of excitations. Therefore $|\sigma(\tilde{P})| \le 2|\sigma(P)|$.

\item The cleaning and branching steps are accomplished solely using horizontal strings to move excitations and stabilizers. If $P$ is a logical operator, then it has no excitations, and hence these two steps are accomplished solely using stabilizers. Since there are no excitations, the merging step is trivial. Therefore $\tilde{P}$ is stabilizer-equivalent to $P$ when $P$ is a logical operator.

\item Finally, suppose $P_0$ differs from $P_1$ on a single edge. If this edge is on an $X$-layer then it does not affect the $e$-configuration of $\tilde{P}_0$ vs. $\tilde{P}_1$ at all. If the edge is on a $Q$ or $Z$ layer, the only thing that it can affect is the merging process by extending a string segment initially terminating at the boundary between two regions into the neighboring region. This can happen on at most one slab boundary.
\end{enumerate}
\end{proof}

\begin{remark}
Note that the merging step in the construction of the operator $\tilde{P}$ is not essential for our arguments; its main purpose is to simplify things by putting the operator in a canonical form.
\end{remark}

Let $\tilde{P}$ be an operator satisfying the assumptions of Lemma~\ref{lem:mapping}. To each such operator we will associate a tensor which indexes the positions of the vertical strings in $\tilde{P}$.

\begin{definition}[String Indicator and Boundary]
Let $\scrL(C)$ be a layer code, let $P \in P_Z(\scrL(C))$ be a Pauli $Z$ operator, and let $\tilde{P}$ be the associated operator obtained by applying Lemma~\ref{lem:mapping}. To $\tilde{P}$ we associate an tensor $c(\tilde{P}) \in \mathbb{F}_2^{\rho_Zn}\otimes \mathbb{F}_2^{n+1}\otimes \mathbb{F}_2^{\rho_Xn+2}$ with components denoted by $c_{ij}^{\ \ k}$. The components are defined such that $c_{ij}^{\ \ k} = 1$ if and only if there exists a vertical string in $\tilde{P}$ located on the $i$-th $Z$-layer, in the region between $Q$-layers $j$ and $j+1$, which crosses the $k$-th $X$-layer. By convention, we will define $k=0$ and $k=\rho_Xn+1$ to be the indices corresponding to be the bottom and top boundaries, respectively, of the layer code, which are not actual $X$-layers. Similarly, $j=0$ and $j=n$ will be the left and right boundaries, respectively, which are not actual $Q$-layers. We call the tensor $c(\tilde{P})$ the \emph{string indicator} of $\tilde{P}$.

To each string indicator $c$ we may associate a boundary $b:= \partial c \in \mathbb{F}_2^{\rho_Zn}\otimes \mathbb{F}_2^{n+1}\otimes \mathbb{F}_2^{\rho_Xn+1}$ defined by components
\begin{align}
b_{ij}^{\ \ k} = c_{ij}^{\ \ k+1} - c_{ij}^{\ \ k}.
\end{align}
The string boundary $b = \partial c$ is an indicator variable for the endpoints of the strings defined by $c$. More precisely, $b_{ij}^{\ \ k} = 1$ if and only if there exists an excitation in $\tilde{P}$ on the $i$-th $Z$-layer, in the region between $Q$-layers $j$ and $j+1$ and $X$ layers $k$ and $k+1$. See Fig.~\ref{fig:bc_def} for an illustration of $c$ and $b$ for a given operator $\tilde P$.
\end{definition}

\begin{figure}[htpb]
	\centering
	\includegraphics[width=.7\linewidth]{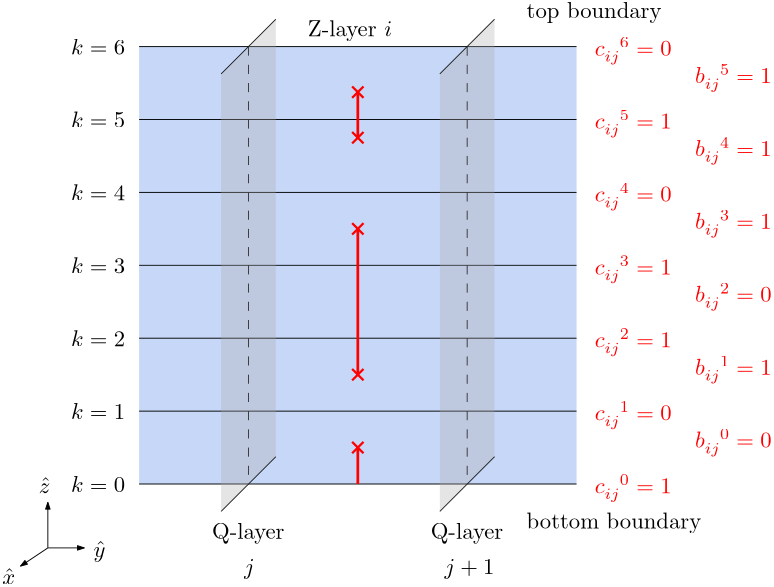}
	\caption{Definition of the indicators $c_{ij}^{\ \ k}$ and $b_{ij}^{\ \ k}$. The figure illustrates the regions on the $i$-th $Z$-layer, between $Q$-layers $j$ and $j+1$. The horizontal lines define the positions of the $X$-layers (regardless of whether or not they actually intersect the $Z$-layer), which are labeled from $k=1$ to $k=5$ (in this figure, $\rho_Xn = 5$). The values $k=0$ and $k=6$ define the bottom and top boundaries of the layer code. A given configuration for an operator $\tilde{P}$ is illustrated, and the corresponding values of $c_{ij}^{\ \ k}$ and $b_{ij}^{\ \ k}$ are provided on the right. The string indicator $c$ defines the presence of strings crossing some $X$-layer. The string boundary $b=\partial c$ defines the presence of excitations within a given bulk region.}
	\label{fig:bc_def}
\end{figure}

\begin{remark}
By Property~\ref{itm:mapping4} of Lemma~\ref{lem:mapping}, it follows that $c(\tilde{P})$ is equivalently a function of $P$ and $H_Z$. In other words, given a Pauli $Z$ operator $P\in \scrL(C)$, we may uniquely associate a string indicator $c(P,H_Z)=c(\tilde{P})$ to $P$ without first choosing $H_X$, i.e., as long as the number of $X$ layers ($\rho_Xn$) is the same, the matrix $H_X$ does not matter. This will be important in the argument to follow.
\end{remark}

\begin{definition}
In this section, we will adopt the convention that our slab boundaries are located right above each $X$-layer. Both $e$-excitations and $e$-configurations will be defined as usual with respect to these boundaries. With this convention in place, we may recover the $e$-configurations of $\tilde{P}$ from its string indicator $c(\tilde{P})$. More precisely, the $e$-configuration of $\tilde{P}$ at the slab boundary above the $k_0$-th $X$-layer is given by the slice of $c$ evaluated at $k_0$. We will denote evaluation of the third index at $k_0$ by an underline $\underline{k_0}$, i.e., we write $c_{ij}^{\ \ \underline{k_0}} \in \mathbb{F}_2^{\rho_Zn}\otimes \mathbb{F}_2^{n+1} \cong \mathcal{E}_Z$ for the tensor obtained from $c_{ij}^{\ \ k}$ by evaluating $k$ at the specific value $k_0$. See Fig.~\ref{fig:excitation1} for an example.   
\end{definition}

\begin{figure}[H]
	\centering
	\includegraphics[width=.7\linewidth]{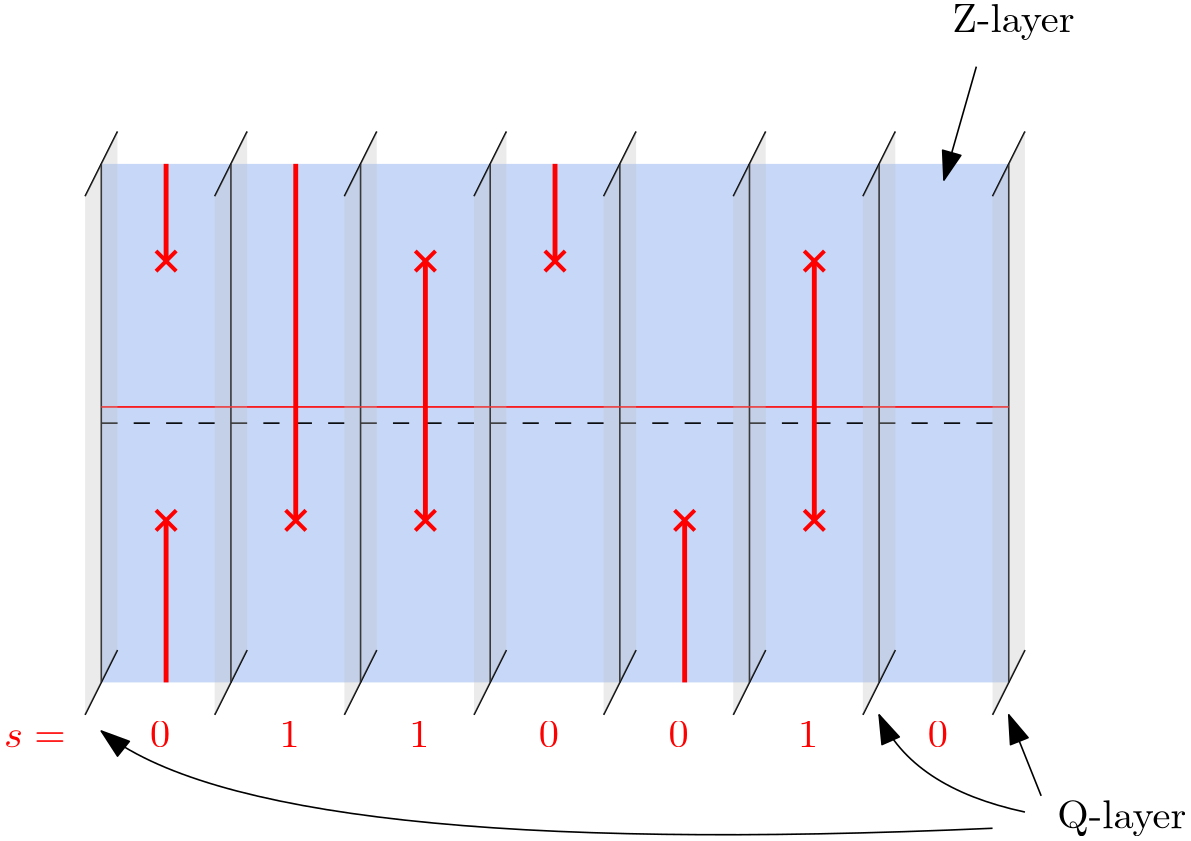}
	\caption{The $e$-configurations defined by the operator $\tilde{P}$. In the top figure, a portion of the operator $\tilde{P}$ (red) is shown on a given $Z$-layer. The dotted horizontal line denotes an $X$-layer location, and the solid horizontal red line denotes the location of a slab boundary, which we take to be just above the $X$-layer. The vertical black lines denote defect lines with $Q$-layers. The intersections of red string operators with the slab boundary define an $e$-configuration supported on the $Z$-layer, i.e., an element of $\cE_Z$, because $\tilde P$ has no support on $Q$-layers.
    The corresponding $e$-configuration of the given $Z$-layer is depicted at the bottom as a binary string $s$. If we take the $Z$-layer and $X$-layer depicted to be the $i$-th and $k$-th, respectively, the entries of the binary string would be $s_j = c_{ij}^{\ \ \underline{k}}$.}
	\label{fig:excitation1}
\end{figure}

\begin{definition}\label{def:setBf}
Let $f \in (0,1)$. Define $\mathfrak{B}_f \subseteq \mathbb{F}_2^{\rho_Zn}\otimes \mathbb{F}_2^{n+1}\otimes \mathbb{F}_2^{\rho_Xn+1}$ to be the set of all string boundaries~$b$ such that:
\begin{enumerate}
\item $|b| \le 2fn/\log n$, where $|b|$ denotes the Hamming weight of the tensor $b$, i.e., the total number of non-zero components,

\item $b = \partial c$ for some $c$ such that $c^{\underline{k}}$ defines a non-trivial $e$-configuration for exactly $\floor{(\rho_Xn-1)/2}$ values of $k \in \{1,2,\cdots,\rho_Xn\}.$
\end{enumerate}
\end{definition}

\begin{lemma}\label{lem:Bsize}
The size of the set $\mathfrak{B}_f$ is bounded above by
\begin{align}
|\mathfrak{B}_f| \le 2^{10fn}
\end{align}
for sufficiently large $n$.
\end{lemma}
\begin{proof}
The total number of $b \in \mathbb{F}_2^{\rho_Zn}\otimes \mathbb{F}_2^{n+1}\otimes \mathbb{F}_2^{\rho_Xn+1}$ satisfying $|b| \le 2fn/\log n$ is bounded above by
\begin{align}
\sum_{\ell = 0}^{2fn/\log n}\binom{\rho_Zn(n+1)(\rho_Xn+1)}{\ell} \le \sum_{\ell = 0}^{2fn/\log n}\binom{2\rho_X\rho_Zn^3}{\ell}\le 2^{2\rho_X\rho_Zn^3H_2\left(\frac{f}{\rho_X\rho_Zn^2\log n}\right)},
\end{align}
where the first inequality above hold for all $n$ sufficiently large. The binary entropy function satisfies
\begin{align}
H_2\left(\frac{f}{\rho_X\rho_Zn^2\log n}\right) &\le \frac{f}{\rho_X\rho_Zn^2\log(2)\log n}\left(1-\log\left(\frac{f}{2\rho_X\rho_Zn^2\log n}\right)\right)\\
&= \frac{f}{\rho_X\rho_Zn^2\log(2)\log n}\log\left(\frac{2e\rho_X\rho_Zn^2\log n}{f}\right)\\
&\le \frac{f}{\rho_X\rho_Zn^2\log(2)\log n}\log\left(n^3\right)\\
&\le \frac{3f}{\log(2)\rho_X\rho_Zn^2},
\end{align}
where the penultimate inequality holds for all $fn/\log n  \ge 2e\rho_X\rho_Z$. It follows that we have
\begin{align}
|\mathfrak{B}_f| \le 2^{2n^3\rho_X\rho_ZH_2\left(\frac{f}{\rho_X\rho_Zn^2\log n}\right)} \le 2^{2n^3\rho_X\rho_Z\frac{3f}{\log(2)\rho_X\rho_Zn^2}} \le 2^{10fn}
\end{align}
for all sufficiently large $n$.
\end{proof}

\begin{lemma}[Configuration]\label{lem:config}
Let $C = (H_X,H_Z)$ be a CSS code and let $\cP=\{P(t)\}_{t=0}^T$ be a Pauli $Z$ path on the layer code $\scrL(C)$ ending on a non-trivial logical operator $P(T)$. If $\Delta_Z(\cP) \le fn/\log n$, then there exists $t' \in [T]$ such that $\partial c(\tilde{P}(t')) \in \mathfrak{B}_f$. 
\end{lemma}
\begin{proof}
We will write $c(t) = c(\tilde{P}(t))$ and $c^{\underline{k}}(t)$ for the $e$-configurations defined by $\tilde{P}(t)$. At $t=0$ we begin with identity operator $P(0)=I$, so $c^{\underline{k}}(0)$ is trivial for all $k \in [\rho_Xn]$. At the last step $t=T$, we have a non-trivial logical operator $P(T)$, and so $\tilde{P}(T)$ is a stabilizer-equivalent logical operator by Property~\ref{itm:mapping6} of Lemma~\ref{lem:mapping}. It follows from Lemma~\ref{lem:layer_slab_support} that $c^{\underline{k}}(T)$ are non-trivial and boundary-equivalent to each other for all $k$.

At each step, $P(t)$ and $P(t+1)$ differ by weight $1$, and so by Property~\ref{itm:mapping7} of Lemma~\ref{lem:mapping}, the operators $\tilde{P}(t)$ and $\tilde{P}(t+1)$ can define boundary inequivalent $e$-configurations on at most one slab boundary, i.e., $c^{\underline{k}}(t)$ and $c^{\underline{k}}(t+1)$ can differ on at most a single value of $k$.  Since $c^{\underline{k}}(t)$ starts trivial for all $k$ at $t=0$, ends non-trivial for all $k$ at $t=T$, and differ on at most a single index $k$  on each time step, it follows that there must exist some $t' \in [T]$ such that $c^{\underline{k}}(t')$ is non-trivial for precisely $\floor{\frac{\rho_X-1}{2}}$ values of $k$.

Finally, since $\Delta_Z(\cP) \le fn/\log n$, it follows that
\begin{align}
|\partial c(t')| \le |\sigma(\tilde{P}(t'))| \le 2|\sigma(P(t'))| \le 2\Delta_Z(\cP) \le 2fn/\log n,
\end{align}
where the first inequality follows from the fact that $|\partial c(t')|$ is syndrome weight of $\tilde{P}(t')$ restricted to the $Z$-layers, and the second inequality follows from Property~\ref{itm:mapping5} of Lemma~\ref{lem:mapping}. Therefore we have $\partial c(t') \in \mathfrak{B}_f$. 
\end{proof}

\begin{lemma}\label{lem:equiv}
Suppose that $c_1,c_2 \in \mathbb{F}_2^{\rho_Zn}\otimes \mathbb{F}_2^{n+1}\otimes \mathbb{F}_2^{\rho_Xn+2}$ are two string indicators that share a common boundary $\partial c_1 = \partial c_2 = b$. If $b \in \mathfrak{B}_f$, then $c_1^{\underline{k}} \approx c_2^{\underline{k}}$ are boundary-equivalent configurations for all $k\in [\rho_Xn]$.
\end{lemma}
\begin{proof}
If $\partial (c_1-c_2)= 0$, then the strings defined by $c_1-c_2$ contain no excitations in the bulk. It follows that they must stretch the entire height of the layer code. This implies that the $e$-configurations defined by $c_1-c_2$ are $k$-independent. Since $b \in \mathfrak{B}_f$, both $c_1$ and $c_2$ define precisely $\floor{(\rho_Xn-1)/2}$ non-trivial configurations. Since the sum of trivial configurations remain trivial, it follows that $c_1-c_2$ can define at most $2\floor{(\rho_Xn-1)/2} \le \rho_Xn-1 < \rho_Xn$ non-trivial $e$-configurations. In particular, $(c_1-c_2)^{\underline{k}}$ must be trivial for some $k \in [\rho_Xn]$, and since its $e$-configurations are $k$-independent, it follows that the $e$-configurations defined by $c_1-c_2$ must be identically trivial. Hence $c_1^{\underline{k}} \approx c_2^{\underline{k}}$ for all $k \in [\rho_Xn]$.
\end{proof}

Given an operator $\tilde{P}$ obtained from Lemma~\ref{lem:mapping}, we can create another associated operator $\hat{P}$ by pushing every vertical string on the $Z$-layer all the way to the left (Fig.~\ref{fig:excitation2}). In this process, each vertical string will split as it passes an intersecting $Q$-layer according to the fusion rules. This process does not create any new excitations on the $X$-layers, and cleans the operator away from the $Z$-layers and onto the $Q$-layers. We then merge any duplicate strings on the $Q$-layers. Note that this is effectively the same process described in Fig.~\ref{fig:distance1}, but now at the operator level. The original $e$-configurations in $\cE_Z$ defined by $\tilde{P}$ will be mapped to a boundary-equivalent $e$-configuration in $\cE_Q$ defined by $\hat{P}$. The original $e$-configurations of $\tilde{P}$ are given by string indicator $c(\tilde{P})$, so we now define an associated string indicator $a(\hat{P})$ for $\hat{P}$.

\begin{figure}[H]
	\centering
	\includegraphics[width=.7\linewidth]{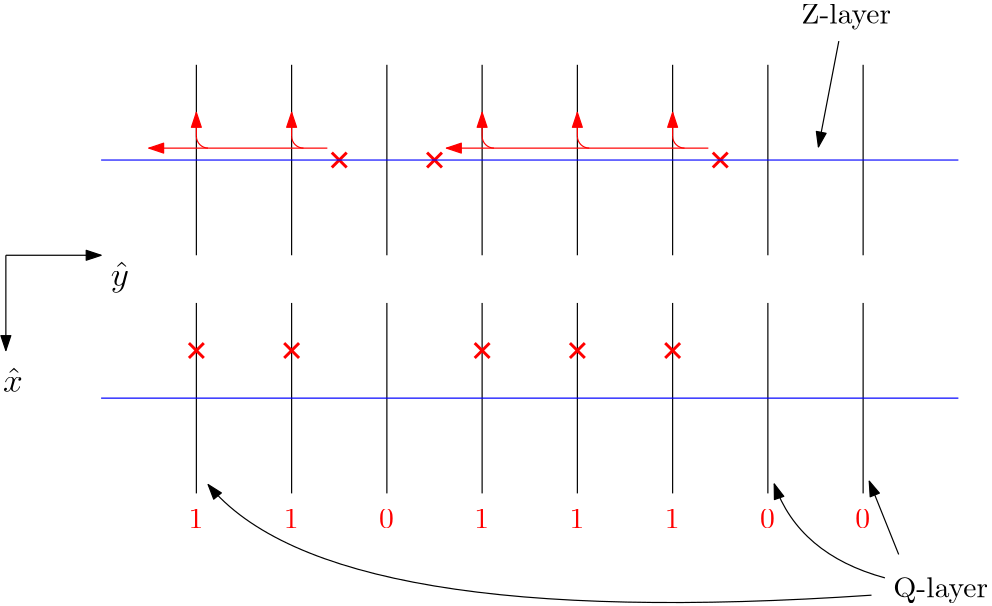}
	\caption{(Top panel) A top-down view of the same $e$-configuration depicted in Fig.~\ref{fig:excitation1}. The vertical (horizontal) lines are $Q$-layers ($Z$-layer) viewed from above. The original $Z$-layer-supported $e$-configuration can be converted into a $Q$-layer-supported $e$-configuration by applying the fusion rules, pushing each excitation on the $Z$-layer towards the left, where it will split at each intersecting $Q$-layer. This process continues until the excitation either annihilates with another excitation down the line, or it reaches the left-most boundary. Note that this is essentially the same process as previously depicted in Fig.~\ref{fig:distance1} (for simplicity, only the non-trivially intersecting $Q$-layers are depicted here). The same fusion rules also define the map from the operator $\tilde{P}$ to $\hat{P}$.
    (bottom panel) The end result of applying the fusion rules to obtain a boundary-equivalent $e$-configuration supported on the $Q$-layers, i.e., an element of $\cE_Q$. The new $e$-configuration is depicted at the bottom as as binary string. If the $Z$-layer depicted is the $k$-th layer, then this is the vector $a^{\underline{k}}$.} 
	\label{fig:excitation2}
\end{figure}

Formally, $a(\hat{P}) \in \mathbb{F}_2^n\otimes \mathbb{F}_2^{\rho_Xn+2}$ is defined to be
\begin{align}
a_j^{\ k}(\hat{P}) = \sum_{i=1}^{\rho_Zn}(H_Z)_{ij}\sum_{j'\ge j}c_{ij'}^{\ \ k}(\tilde{P}).\label{eq:ajk}
\end{align}
In the equation above, each summand $(H_Z)_{ij}\sum_{j'\ge j}c_{ij'}^{\ \ k}$ keeps tracks of the number (or more precisely, the parity) of strings that $Q$-layer $j$ picks up from the $i$-th $Z$-layer as the strings are moved right to left. The presence of the term $(H_Z)_{ij}$ simply encodes the fusion rules, so that the strings branch if $Q$-layer $j$ intersects non-trivially with $Z$-layer $i$, i.e., if $(H_Z)_{ij} = 1$, and that they do not branch in the case of a trivial intersection $(H_Z)_{ij}=0$. Then $a_j^{\ k}$ is obtained by summing over the contributions of all $Z$-layers. In this way, $a_j^{\ k} = 1$ if and only if there exists a vertical string in $\hat{P}$ on the $j$-th $Q$-layer passing through the $k$-th $X$-layer. The $e$-configurations of $\hat{P}$ are then given by $a_j^{\ \underline{k}}(\hat{P})$, with $a_j^{\ \underline{k}}(\hat{P}) \approx c_{ij}^{\ \ \underline{k}}(\tilde{P})$ for all $k\in [\rho_Xn]$ by construction.

In what follows, we will write $\sigma|_X(P,H_X)$ to denote the syndrome, \emph{restricted to the $X$-layers}, of an operator $P \in \cP_Z(\scrL)$. We write the dependence on $H_X$ explicitly to emphasize the dependence of the syndrome on the choice of the input $H_X$ to the layer code. The dependence on $H_Z$ is implicit.

\begin{lemma}\label{lem:Xsyndrome}
The syndrome of an operator $P\in P_Z(\scrL(C))$ on the $X$-layers has weight
\begin{align}
\bigg|\sigma|_X(P,H_X)\bigg| = \bigg|\sigma|_X(\tilde{P},H_X)\bigg| = \bigg|\sigma|_X(\hat{P},H_X)\bigg| \ge \sum_{k=1}^{\rho_Xn}\left|\sum_{j=1}^n(H_X)_{kj}a_j^{\ k}(\hat{P})\right|.
\end{align}
\end{lemma}
\begin{proof}
The mapping $P\mapsto \tilde{P}$ preserves the excitations on $X$-layers by Property~\ref{itm:mapping2} of Lemma~\ref{lem:mapping}, so $\sigma|_X(P,H_X)=\sigma|_X(\tilde{P},H_X)$. Likewise, the map $\tilde{P}\mapsto \hat{P}$ also preserves the excitations on the $X$-layers, so $\sigma|_X(\tilde{P},H_X) = \sigma|_X(\hat{P},H_X)$.

Now, fix an $X$-layer $\ell$ (say the $k$-th layer) and consider the excitations of $\hat{P}$ on this layer. Every vertical string that crosses $\ell$ on a $Q$-layer that intersects $\ell$ non-trivially will create a potential excitation in the bulk which can only be condensed by another such crossing. The parity of these crossings, given by
\begin{align}
\sum_{j=1}^n(H_X)_{kj}a_j^{\ k}(\hat{P}),
\end{align}
therefore lower bounds the number of excitations on $\ell$. The syndrome weight on all $X$-layers is therefore lower bounded by the sum of the parities above over all $X$-layers, i.e.,
\begin{align}
\bigg|\sigma|_X(\hat{P},H_X)\bigg| \ge \sum_{k=1}^{\rho_Xn}\left|\sum_{j=1}^n(H_X)_{kj}a_j^{\ k}(\hat{P})\right|.
\end{align}
Note that the absolute values indicate that the outer sum is not over $\mathbb{F}_2$ since distinct $X$-layers are independent.
\end{proof}

\begin{remark}\label{rmk:rowspace}
Since $c^{\underline{k}} \approx a^{\underline{k}}$, it follows from Lemma~\ref{lem:layer_input_map} that $c^{\underline{k}}$ is trivial if and only if $a^{\underline{k}} \in \mathrm{row}(H_Z)$. 
\end{remark}

\begin{lemma}\label{lem:energy_main}
Fix a $\rho_Zn \times n$ matrix $H_Z$ and some $b \in \mathfrak{B}_f$. Let $\mathfrak{C}(b) = \partial^{-1}(b)$ be the space of all string indicators $c$ with boundary $b=\partial c$. For a given $H_X$, let $\mathfrak{C}(H_X)$ be the space of all $c$ such that 
\begin{align}
\sum_{k=1}^{\rho_Xn}\left|\sum_{j=1}^n(H_X)_{kj}a_j^{\ k}(c)\right| \le \frac{\rho_Xn}{8}.
\end{align}
Then we have
\begin{align}
\Pr_{H_X}\left(\mathfrak{C}(b) \cap \mathfrak{C}(H_X) \neq \emptyset \mid H_Z\right) \le 2^{-\rho_Xn/15}
\end{align}
for all sufficiently large $n$, where the probability is over all $\rho_Xn\times n$ matrices $H_X$ orthogonal to $H_Z$.
\end{lemma}

\begin{proof}
Let $c_1,c_2 \in \mathfrak{C}(b)$. By Lemma~\ref{lem:equiv}, it follows that $c_1^{\underline{k}} \approx c_2^{\underline{k}}$ for all $k$. Consequently, $(a_1-a_2)^{\underline{k}} \in \mathrm{row}(H_Z)$ for all $k$ by Remark~\ref{rmk:rowspace}. Since the rows of $H_Z$ are orthogonal to the rows of $H_X$, it follows that
\begin{align}
\sum_{j=1}^n(H_X)_{kj}(a_1-a_2)_{j}^{\ k} = 0
\end{align}
for all $k$. This implies that
\begin{align}
\sum_{k=1}^{\rho_Xn}\left|\sum_{j=1}^n(H_X)_{kj}(a_1)_j^{\ k}\right| = \sum_{k=1}^{\rho_Xn}\left|\sum_{j=1}^n(H_X)_{kj}(a_2)_j^{\ k}\right|
\end{align}
It follows that if $\mathfrak{C}(b) \cap \mathfrak{C}(H_X) \neq \emptyset$, then we actually have $\mathfrak{C}(b) \subseteq \mathfrak{C}(H_X)$.

Now, fix an arbitrary $c \in \mathfrak{C}(b)$ and consider the quantities
\begin{align}
\alpha^k(H_X) = \sum_{j=1}^n(H_X)_{kj}a_j^{\ k}(c)
\end{align}
as we vary over random $H_X$. Since $\mathfrak{C}(b) \cap \mathfrak{C}(H_X) \neq \emptyset$ if and only if $c \in \mathfrak{C}(H_X)$, if and only if
\begin{align}
\sum_{k=1}^{\rho_Xn}|\alpha^k(H_X)| = \sum_{k=1}^{\rho_Xn}\left|\sum_{j=1}^n(H_X)_{kj}a_j^{\ k}(c)\right| \le \frac{\rho_Xn}{8},
\end{align}
it follows that we have
\begin{align}
\Pr_{H_X}\left(\mathfrak{C}(b) \cap \mathfrak{C}(H_X) \neq \emptyset\mid H_Z\right) = \Pr_{H_X}\left(\sum_{k=1}^{\rho_Xn}|\alpha^k(H_X)| \le \frac{\rho_Xn}{8}\ \bigg|\  H_Z\right).\label{eq:hoeffding1}
\end{align}
Since $H_X$ is uniformly random in the orthogonal complement of $H_Z$, each row is also i.i.d. uniformly random in the orthogonal complement. It follows that the $\alpha^k$'s, each being a linear functional of the $k$-th row of $H_X$, are also independent. Each $\alpha^k$ is either identically zero if $a^{\underline{k}} \in \mathrm{row}(H_Z)$, or else it is uniformly random on $\{0,1\}$ if $a^{\underline{k}}\notin \mathrm{row}(H_Z)$. From Lemma~\ref{lem:equiv} and the fact that $b\in \mathfrak{B}_f$, it follows that $c^{\underline{k}}$, and hence $a^{\underline{k}}$, defines a non-trivial configuration for exactly $\floor{(\rho_X n-1)/2}$ values of $k$. Therefore
\begin{align}
\sum_{k=1}^{\rho_Xn}|\alpha^k(H_X)| = \sum_{k=1}^{\rho_Xn}\left|\sum_{j=1}^n(H_X)_{kj}a_j^{\ k}(c)\right|
\end{align}
is a sum of $\floor{(\rho_X n-1)/2}$ independent Bernoulli random variables. The probability in Eq.~\eqref{eq:hoeffding1} can therefore be bounded by above by Hoeffding's inequality as
\begin{align}
\Pr_{H_X}\left(\sum_{k=1}^{\rho_Xn}|\alpha^k(H_X)| \le \frac{\rho_Xn}{8}\ \bigg|\  H_Z\right) &= \Pr_{H_X}\left(\sum_{k=1}^{\rho_Xn}|\alpha^k(H_X)| - \frac{1}{2}\floor{\frac{\rho_Xn-1}{2}} \le \frac{\rho_Xn}{8}-\frac{1}{2}\floor{\frac{\rho_Xn-1}{2}}\ \bigg|\  H_Z\right)\\
&\le \Pr_{H_X}\left(\sum_{k=1}^{\rho_Xn}|\alpha^k(H_X)| - \frac{1}{2}\floor{\frac{\rho_Xn-1}{2}} \le -\frac{\rho_Xn}{9}\ \bigg|\  H_Z\right)\\
&\le \exp\left(-\frac{2(\rho_Xn/9)^2}{\floor{(\rho_X n-1)/2}}\right)\\
&\le \exp\left(-\frac{4\rho_Xn}{81}\right)\\
&\le 2^{-\rho_Xn/15},
\end{align}
where we subtract off the mean $\frac{1}{2}\floor{\frac{\rho_Xn-1}{2}}$ in the first line and apply Hoeffding's inequality in the third.
\end{proof}

Now we combine everything to get the main theorem for the energy barrier.

\begin{theorem}[Layer Code Energy Barrier]\label{thm:layer_energy}
Let $\scrL(C)$ be a layer code with random CSS input $C\sim \mathrm{CSS}_n(\rho_X,\rho_Z)$. Then there exists a constant $f > 0$ such that
\begin{align}
\Pr_{C}\left(\Delta_{\scrL(C)} \le \frac{fn}{\log n}\right) \le 2^{-\Omega(n)}.
\end{align}
In particular, $\scrL(C)$ has energy barrier $\Delta_{\scrL(C)} > fn/\log n$ with high probability as $n\rightarrow \infty$.
\end{theorem}
\begin{proof}
It suffices to prove the result for the $Z$ energy barrier. A symmetric argument will then imply that the $X$ energy barrier satisfies the same bound with high probability, and the full result follows by the union bound. We adopt the same notation used in the Lemma~\ref{lem:energy_main}.

First, let us fix some $H_Z$. Let $H_X$ be a random matrix orthogonal to $H_Z$. Let $C=(H_X,H_Z)$. We will write $\Delta_Z(\mathscr{L}(C))$ for the $Z$ energy barrier associated with $\scrL(C)$. Suppose that $\Delta_Z(\mathscr{L}(C)) \le fn/\log n$ for some constant $f$. Then there exists some Pauli path $\cP = \{P(t)\}_{t=0}^T$, ending on a non-trivial logical operator $P(T)$, such that $\Delta_Z(\cP) \le fn/\log n$. It follows by Lemma~\ref{lem:config} that there exists some $t' \in [T]$ such that $c=c(\tilde{P}(t'))$ satisfies $b = \partial c \in \mathfrak{B}_f$. It also follows from Lemma~\ref{lem:Xsyndrome} that
\begin{align}
\sum_{k=1}^{\rho_Xn}\left|\sum_{j=1}^n(H_X)_{kj}a_j^{\ k}(c)\right| \le \left|\sigma|_X(P(t'),H_X)\right| \le \Delta_Z(\cP) \le \frac{fn}{\log n}. 
\end{align}
In particular, for $n$ sufficiently large, we have $fn/\log n < \rho_Xn/8$, so that $\Delta_Z(\mathscr{L}(C))\le fn/\log n$ implies $\mathfrak{C}(b) \cap \mathfrak{C}(H_X) \neq \emptyset$. It follows that
\begin{align}
\Pr_{H_X}\left(\Delta_Z(\mathscr{L}(C)) \le \frac{fn}{\log n} \bigg|\ H_Z\right) \le \Pr_{H_X}\left(\exists b \in \mathfrak{B}_f:\ \mathfrak{C}(b) \cap \mathfrak{C}(H_X) \neq \emptyset \mid H_Z\right)
\end{align}
for sufficiently large $n$. Therefore, we have
\begin{align}
\Pr_{H_X}\left(\Delta_Z(\mathscr{L}(C)) \le \frac{fn}{\log n} \bigg|\ H_Z\right) &\le \Pr_{H_X}\left(\exists b \in \mathfrak{B}_f:\ \mathfrak{C}(b) \cap \mathfrak{C}(H_X) \neq \emptyset \mid H_Z\right)\\
&\le \sum_{b \in \mathfrak{B}_f}\Pr_{H_X}\left(\mathfrak{C}(b) \cap \mathfrak{C}(H_X) \neq \emptyset \mid H_Z\right)\\
&\le |\mathfrak{B}_f|\cdot 2^{-\rho_Xn/15}\\
&\le 2^{10fn}2^{-\rho_Xn/15},
\end{align}
where we apply a union bound on the second line, Lemma~\ref{lem:energy_main} on the third line, and Lemma~\ref{lem:Bsize} on the last line. It follows that if we choose $f$ sufficiently small so that $150f < \rho_X$, then
\begin{align}
\Pr_{H_X}\left(\Delta_Z(\mathscr{L}(C)) \le \frac{fn}{\log n} \bigg|\ H_Z\right) = 2^{-\Omega(n)}.
\end{align}
This result holds for any choice of $H_Z$, so randomizing over $H_Z$ gives
\begin{align}
\Pr_{H_X,H_Z}\left(\Delta_Z(\mathscr{L}(C)) \le \frac{fn}{\log n}\right) = 2^{-\Omega(n)},
\end{align}
which is our desired result. 
\end{proof}

\subsection{Partial Self-Correction of Random Layer Codes}\label{sec:partial_self_corr_HDPC}

We show that layer codes with random input codes are partially self-correcting with high probability (taken over the choice of input code). To do so, we first define a variant of the concatenated decoder (cf. section~\ref{sec:decoder_concatenated}) that is suitable for correcting against errors with low energy barrier which can then be used in Lemma~\ref{lem:BH} (Lemma~1 of Ref.~\cite{bravyi2011analytic}) to infer partial self-correction.

We define the decoder for $Z$ errors (with $X$ errors treated analogously). Let $\scrL(C)$ be a random layer code with input code $C\sim \mathrm{CSS}_n(\rho_X,\rho_Z)$. Let $R_Z$, $R_Q$, and $R_X$ denote the $Z$-, $Q$-, and $X$-layers of code $\scrL(C)$, respectively. Let $e_0$ be a $Z$-error on $\scrL(C)$ and $\sigma(e_0)=\sigma_0$ be its syndrome. The modified concatenated decoder then proceeds as follows:

\begin{enumerate}
	\item Match all $e$-excitations on $R_Z$ to the top boundary of each $Z$-layer (which condenses $e$-anyons). These matchings are done with straight vertical strings operators. This step may create additional excitations on $R_X$.
    
	\item Next, we match all $e$-excitations on $R_Q$ to the top boundary (which is also $e$-condensing) by straight vertical string operators. Again, this step might create new excitations on $R_X$, but none on $R_Z$.
    
	\item Let $\sigma_C\in \mathbb{F}_2^{n_X}$ be the indicator vector of all $X$-layers that now contain an odd number of $e$-excitations. Notice that $\sigma_C$ can be viewed as a syndrome to the input code $C$. Let $\operatorname{Dec}_{C,Y}^Z$ be a decoder of input code $C$ that outputs a correction $\hat f_C$ of minimal $Y$-weight (see Definition~\ref{def:YC}) such that $\sigma(\hat f_C)=\sigma_C$. For every qubit in $\hat f_C$, we apply a straight vertical string operator from the top to the bottom of the corresponding $Q$-layer of the layer code. As before, additional $e$-excitations may appear on $R_X$, but none on $R_Z$.
    
	\item Finally, apply MWPM on each $X$-layer to eliminate the remaining excitations.
\end{enumerate}

We now prove the correctness of the modified concatenated decoder for errors of sufficiently low energy barrier (cf. Lemma~\ref{lem:HDPCdecoderEnergyBarrier}). We begin with a preliminary result on the correctness of the input decoder $\operatorname{Dec}_{C,Y}^Z$.

\begin{lemma}\label{lem:DecY_dist}
    Let $C\sim \mathrm{CSS}_n(\rho_X,\rho_Z)$ be a random CSS code. There exists a constant $c_0>0$ such that with probability $1 - 2^{-\Omega(n)}$, there exists a decoder $\operatorname{Dec}_{C,Y}^Z$ that corrects all $Z$-error $e$ of $Y$-weight $|e|_Y\le c_0n/\log n$.
\end{lemma}
\begin{proof}
    This is an immediate corollary of Lemma~\ref{lem:f_bound} and the fact that the minimum $Y$-weight decoder corrects errors of $Y$-weight at most $(d_Y-1)/2$, where $d_Y = \min\{fn/\log n: B_C(f)\cap L_Z(C)\ne\emptyset\}$ is the $Y$-distance of the code.
\end{proof}

Next, we establish an analog of Lemma~\ref{lem:config} for errors of sufficiently low energy barrier.

\begin{lemma}\label{lem:ak_rowHz_atleasthalf}
Let $\mathscr{L}(C)$ be a random layer code. Let $f > 0$ be the energy barrier constant from Theorem~\ref{thm:layer_energy}. Then with probability $1 - 2^{-\Omega(n)}$, all Pauli-$Z$ operators with energy barrier $\Delta_Z(P)\le fn/\log n$ satisfies the property that $a^{\underline{k}}(\hat{P})\in \row(H_Z)$ for at least half of the values $k\in [\rho_X n]$.
\end{lemma}
\begin{proof}
Let $\mathfrak{P}$ denote the set of all Pauli $Z$ operators with energy barrier $\Delta_Z(P)\le fn/\log n$ and $a^{\underline{k}}\in \row(H_Z)$ for less than half the values of $k\in [\rho_Xn]$. We show that $\mathfrak{P}=\emptyset$ with high probability.

First, let $P\in \mathfrak{P}$ be arbitrary. By Remark~\ref{rmk:rowspace} it follows that $c^{\underline{k}}$ is non-trivial for more than half the values of $k$. Consider a Pauli path $\mathcal{P}=\{P(t)\}_{t=0}^T$ ending on $P$ and let $c(t)=c(\tilde{P}(t))$. By an argument identical to the proof of Lemma~\ref{lem:config}, there exists a time $t'$ such that $\partial c(t')\in \mathfrak{B}_f$. Let us write $c'=c(t')$ and $b'=\partial c'$.

By Lemma~\ref{lem:Xsyndrome}, it follows that
\begin{align}
\sum_{k=1}^{\rho_Xn}\left|\sum_{j=1}^n(H_X)_{kj}a_j^{\ k}(c')\right| \le \left|\sigma|_X(P(t'),H_X)\right| \le \Delta_Z(P) \le \frac{fn}{\log n}. 
\end{align}
For $n$ sufficiently large so that $fn/\log n < \rho_Xn/8$, it follows that $\mathfrak{P}\neq \emptyset$ implies that there exists $b'\in \mathfrak{B}_f$ such that $\mathfrak{C}(b')\cap \mathfrak{C}(H_X) \neq \emptyset$. Therefore
\begin{align}
\Pr_{H_X}\left(\mathfrak{P}\neq \emptyset \mid H_Z\right) \le \Pr_{H_X}\left(\exists b' \in \mathfrak{B}_f:\ \mathfrak{C}(b') \cap \mathfrak{C}(H_X) \neq \emptyset \mid H_Z\right)
\end{align}
for sufficiently large $n$. The latter probability is $2^{-\Omega(n)}$ by the proof of Theorem~\ref{thm:layer_energy}. Finally, randomizing over $H_Z$ gives the desired result.
\end{proof}

\begin{lemma} \label{lem:HDPCdecoderEnergyBarrier}
Let $\scrL(C)$ be a random layer code with input code $C\sim\mathrm{CSS}_n(\rho_X,\rho_Z)$. There is a constant $c>0$, such that with probability $1 - 2^{-\Omega(n)}$ over the choices of the input code $C$, the modified concatenated decoder is able to correct all Pauli $Z$ errors $e\in P_Z(\scrL(C))$ of energy barrier $\Delta_Z(e) \le cn/\log n$.
\end{lemma}
\begin{proof}
Let $c=\min (c_0/2,f)$, where $c_0$ and $f$ are the constants appearing in Lemmas~\ref{lem:DecY_dist} and \ref{lem:ak_rowHz_atleasthalf}, respectively. Let $e_0$ be an error with energy barrier $\Delta_Z(e_0) \le cn/\log n$. Let $a^{\ \underline{k}}_0=a^{\underline{k}}(\hat{e}_0)$, where $\hat e_0$ is defined as in Fig.~\ref{fig:excitation2}. Let $e_1$ be an error obtained from $e_0$ after matching a single excitation in $e_0$, either in a $Z$-layer or a $Q$-layer, to the top boundary. Let $a_1^{\ \underline{k}}=a^{\underline{k}}(\hat{e}_1)$. Let us compute how $a_1^{\ \underline{k}}$ differs from $a_0^{\ \underline{k}}$.

First, suppose the excitation was originally on a $Z$-layer, below the $k_0$-th slab boundary. Then $e_0$ and $e_1$ differ by a vertical string operator matching the excitation to the top boundary. This string operator modifies the $e$-configurations at or above the $k_0$-th slab boundary, and after propagating the configurations onto the $Q$-layers, the result becomes
\begin{equation}
a^{\ \underline{k}}_1 = \begin{cases}
a^{\ \underline{k}}_0 + y^{k}, & \text{ if } k\ge k_0,\\
a^{\ \underline{k}}_0, & \text{ if } k < k_0,
\end{cases}
\end{equation}
where $y^k\in Y_C$ is some vector which depends on the location of excitation. Note that the vertical string will spawn additional excitations on the $X$-layers, but these do not contribute to $y^k$.

Next, suppose the excitation is on the $i$-th $Q$-layer, below the $k_0$-th slab boundary. The matching string will again modify the $e$-configurations at or above the $k_0$-th slab boundary, this time simply by the $i$-th elementary basis vector $e_i$. If the $Q$-layer has non-trivial intersection with the $j$-th $Z$-layer, then we can write $e_i = y_{i}-y_{i-1}$, where $y_i\in Y_C^i$ and $y_{i-1}\in Y_C^{i-1}$ are truncated versions of the $j$-th stabilizer generator. Therefore
\begin{equation}
a^{\ \underline{k}}_1 = \begin{cases}
a^{\ \underline{k}}_0 + y_{i}-y_{i-1}, & \text{ if } k\ge k_0,\\
a^{\ \underline{k}}_0, & \text{ if } k < k_0.
\end{cases}
\end{equation}
Either way, the matching of any $e$-excitation to the top boundary modifies each $a_0^{\ \underline{k}}$ by the addition of at most two vectors in $Y_C$. Iterating this process, it follows that the final $e$-configuration after all excitations have been matched to the top boundary is given by
\begin{align}
a_{2}^{\ \underline{k}} = a^{\ \underline{k}}_0 + \sum_{y\in S_k} y,
\end{align}
where $S_k\subseteq Y_C$ with cardinality at most $|S_k|\le 2cn/\log n$. Note that $a_2^{\ \underline{k}}$ are the $e$-configurations of the remaining error after the completion of the first two steps of the modified concatenated decoder. At this point, there are no more excitations on the $Z$- or $Q$-layers, so it follows by Lemma~\ref{lem:equivconfigurations} that the  $a^{\underline{k}}_2$ are boundary-equivalent for all $k\in [\rho_Xn]$.

In Step 3 of the modified concatenated decoder, the syndrome $\sigma_C$ given to the input code is the indicator of the $X$-layers that contain an odd number of $e$-excitations. Note that $\sigma_C$ is the syndrome corresponding to  $a_2^{\ \underline{k}}$, regarded as an error of the input code (this is true for all $k$ since boundary-equivalent $e$-configurations correspond to stabilizer-equivalent input errors). If $a_0^{\ \underline{k}}$ is trivial for some $k=k'$, and if $\operatorname{Dec}_C^Z$ successfully corrects all errors of $Y$-weight at most $2cn/\log n$, then $a_2^{\ \underline{k}'}$ is correctable under $\operatorname{Dec}_C^Z$. The resulting correction $\hat f_C$ is therefore equal to $\sum_{y\in S_{k'}} y$ up to a $Z$ stabilizer of the input code.

Applying vertical strings on the $Q$-layers corresponding to the correction $\hat f_C$ therefore makes the $k'$-th $e$-configuration trivial. Since this process does not create excitations on $Z$- or $Q$- layers, all $e$-configurations remain boundary-equivalent by Lemma~\ref{lem:equivconfigurations}, and so all $e$-configurations on all slab boundaries are trivial. The correctness of $\operatorname{Dec}_C^Z$ implies that every $X$-layer must have an even number of excitations. Therefore we can apply MWPM in Step 4 to eliminate the remaining excitations on the $X$-layers. Since all $e$-configurations of the resulting operator are trivial, it follows by Lemma~\ref{lem:layer_slab_support} that the final result is a stabilizer, and we have successfully corrected the error.

Throughout this process, the only two assumptions made are:
\begin{enumerate}
\item The correctness of the input decoder $\operatorname{Dec}_C^Z$ on errors of $Y$-weight less than $2cn/\log n$.

\item The existence of some $k'\in [\rho_Xn]$ such that $a_0^{\ \underline{k}'}$ is trivial, i.e., $a_0^{\ \underline{k}'}\in \row(H_Z)$, for all errors $e_0$ with energy barrier $\Delta_Z(e_0) \le cn/\log n$. 
\end{enumerate}
Since $2c \le c_0$, the first assumption holds with probability $1 - 2^{-\Omega(n)}$ by Lemma~\ref{lem:DecY_dist}. Since $c\le f$, the second assumption also holds with probability $1 - 2^{-\Omega(n)}$ by Lemma~\ref{lem:ak_rowHz_atleasthalf}. A union bound implies that both assumptions hold with high probability, and the result follows.
\end{proof}

\begin{theorem}\label{thm:psc_random_layer}
    Let $C = \{C_n\}_{n=1}^\infty$ be a random family of CSS codes, where each $C_n\sim\mathrm{CSS}_n(\rho_X,\rho_Z)$ independently with $\rho_X,\rho_Z = \frac 1 2(1-1/\log n)$. With probability one, the layer code family $\scrL=\{\scrL(C_n)\}_{n=1}^\infty$ is partially self-correcting using the modified concatenated decoder $\widetilde\Phi$ with input code decoder $\operatorname{Dec}_{C,Y}^Z$. In particular, maximum memory time $t^*_\mathrm{mem} = \exp(\exp(\Omega(\beta)))$ can be achieved with cutoff length $L^*=\exp(\Theta(\beta))$ for sufficiently large $\beta$.
\end{theorem}
\begin{proof}
    Let $c$ be the constant in Lemma~\ref{lem:HDPCdecoderEnergyBarrier}, and $M_n$ be the event that the decoder $\widetilde\Phi$ can correct errors of energy barrier up to $m=cn/\log n$ for the code $\sL(C_n)$. Thus, $\Pr(M_n)\ge 1-2^{-\gamma n}$, where $\gamma$ is some positive constant.
    Because $\sum_{n=1}^\infty 2^{-\gamma n} < \infty$, the Borel-Cantelli lemma implies that with probability one, there exists some $n_0$ such that $M_n$ holds for all $n\ge n_0$.
    
    For sufficiently large $\beta$, the cutoff size $\exp(\beta/2)$ for the layer code family $\sL(C_n)$---which gives cutoff length $L^*=\exp(\Theta(\beta))$---is such that the input code size is at least $n_0$. Therefore, the decoder $\widetilde\Phi$ can correct errors with energy barrier up to the cutoff $m^*=\Theta(L^*/\log L^*)$. The result then follows from Corollary~\ref{cor:partial_self_corr_decoder}.
\end{proof}

\begin{remark}\label{rem:psc_random_layer_efficient}
Note that the decoder $\operatorname{Dec}_C^Z$ of the random input code is not necessarily an \emph{efficient} algorithm, by similar reasoning as the NP-hardness of maximum likelihood decoding for random codes~\cite{Berlekamp78,PhysRevA.83.052331}.
Although the modified concatenated decoder is more efficient than the naive energy barrier decoder, its complexity means that the result of Theorem~\ref{thm:psc_random_layer} only implies that $\scrL(C)$ is only a partially self-correcting system.
However, if we restrict to errors of at most $O(\log n)$ energy barrier and codes with $k = O(\log n)$ logical qubits, we obtain a polynomial-time decoder.
The result is a partially self-correcting memory with maximum memory time scaling exponential in $\beta^2$ (instead of doubly exponential in $\beta$).
This is comparable to the cubic code~\cite{bravyi2011analytic} but with more logical qubits.
\end{remark}

\begin{remark}\label{rem:LscalingHDPC}
    Similar to Remark~\ref{rem:Lscaling}, we may consider the memory time scaling with the linear system size $L$ up to $L^*=\exp(\Theta(\beta))$. Following the same arguments, we obtain that for sufficiently large $\beta$ and $L$, the memory time scales as $t_{\mathrm{mem}} = \exp(\Omega(\beta \log L))$ or $t_{\mathrm{mem}} = \exp(\Omega(\beta L/\log L))$ for the modified concatenated decoder that is required or not required to be efficient, respectively.
\end{remark}

A summary of our results on the partially self-correcting quantum systems or memories based on layer codes is presented in Table~\ref{tab:results}.

\begin{table}[htpb]
    \centering
    \begin{tabular}{|c|c|c|c|c|}
        \hline
        \textbf{Input code} & $\boldsymbol k$ & $\boldsymbol{t^*_{\mathrm{mem}}}$ & \textbf{Decoder} & \textbf{Efficient?}\\\hline
        Quantum Tanner & $\Theta(L)$ & $\exp(\exp(\Omega(\beta)))$ & Concatenated & Yes\\\hline
        Random & $\Theta(L/\log L)$ & $\exp(\exp(\Omega(\beta)))$ & Modified Concatenated & No\\\hline
        Random & $\Theta(\log L)$ & $\exp(\Omega(\beta^2))$ & Modified Concatenated & Yes\\\hline
    \end{tabular}
    \caption{
    Overview of partially self-correcting layer codes with different types of input codes.  All results have cutoff length $L^*=\exp(\Theta(\beta))$. An efficient decoder implies that the code is also a SCQM.}
    \label{tab:results}
\end{table}

\section{Numerical Simulations}\label{sec:numerics}

To complement the analytical results we also numerically study the performance of layer codes. First, in Appendix~\ref{sec_est_mem_time}, we describe the simulations that we perform to estimate the memory time.
Then, in Appendix~\ref{sec:rand_layer_codes}, we choose a family of random layer codes and study their memory times.
We provide estimates of the maximum memory times $t_{\text{mem}}^*$ and cutoff system sizes $n^*$.
Finally, in Appendix~\ref{sec:num_3Dtoric}, we simulate the 3D toric code, which is self-correcting for $X$ errors, to contrast the behavior of partial self-correction of random layer codes.

\subsection{Estimating the Memory Time}
\label{sec_est_mem_time}

The Hamiltonian of a CSS code consists of pairwise commuting terms that decompose as $H_{\rm mem}=H_{X,\rm mem}+H_{Z,\rm mem}$, where $H_{X,\rm mem}$ and $H_{Z,\rm mem}$ comprise terms 
corresponding to the $X$- and $Z$-type stabilizer generators, respectively (cf. Eq.~\ref{eq:H_mem}).
The interaction of a CSS code with a heat bath of finite temperature can be studied for both $X$- and $Z$-type stabilizer generators separately. This allows us to use standard tools from statistical mechanics~\cite{QMFT,bravyi2011analytic}. In particular, the memory time can be simulated by a Monte Carlo algorithm with Glauber~\cite{GlauberDynamics} dynamics for $X$- and $Z$-type errors separately. Since the layer code construction treats Pauli $X$ and $Z$ errors symmetrically, it suffices to consider just $X$-type errors. We consider only $X$-type errors in the remainder of this appendix.

Consider the set $S_Z$ of $Z$-type stabilizer generators of a CSS code of block length $N$. We can associate a spin Hamiltonian $H:\{\pm 1\}^N\to \mathbb{R}$ to $S_Z$, defined as
\begin{align}
    H(\sigma) = -\frac12 \sum_{s\in S_Z}\prod_{i\in \supp(s)} \sigma_i  ,
\end{align}
which assigns energy to each spin configuration $\sigma=(\sigma_1,\ldots,\sigma_n)\in\{\pm 1\}^{N}$. 
The set of ground states for $H$ corresponds to the codewords of the classical code $C_Z=\ker H_Z$, where $H_Z$ is the parity check matrix associated with $S_Z$.

At time $t=0$ the state $\sigma$ is initialized in a ground state of the Hamiltonian $H$. The Monte Carlo simulation with Glauber dynamics models thermal noise due to the memory-bath interaction as a sequence of spin-flips, which corresponds to applications of the respective Pauli-$X$ errors. For each time step, a spin index $i\sim \mathrm{Uniform}([N])$ is drawn uniformly and independently at random. The spin $i$ is then flipped with probability depending on the energy change it causes, i.e., $\Delta E_i(\sigma)=H(\sigma_1,\ldots,\sigma_{i-1},-\sigma_i,\sigma_{i+1},\ldots, \sigma_N) - H(\sigma)$. We have
\begin{align}
    \sigma_i \mapsto \begin{cases}
        -\sigma_i, & \quad \text{ with probability } \frac{1}{1+\exp(\beta \Delta E_i(\sigma))},\\
        \sigma_i, & \quad \text{ otherwise} ,
    \end{cases}\label{eq:spinflip_acceptance}
\end{align}
where $\beta$ is the inverse temperature. This evolution then drives the state of the system to the Gibbs distribution.

Given a syndrome decoder for the CSS code, the memory time $t_{\text{fail}}$ is the first time at which the memory becomes corrupted, meaning that the encoded logical information cannot be recovered using the decoder. To determine $t_{\text{fail}}$, in each time step, we obtain the error syndrome from the spin configuration, attempt to decode it, and verify if the decoder introduces a logical error.

At low temperatures, the probability of accepting a spin flip in each round is  small, cf. Eq.~\eqref{eq:spinflip_acceptance}, which slows down the standard Monte Carlo approach. To address this, we employ the kinetic or $n$-fold-way Monte Carlo method~\cite{Young_1966,BORTZ1975}. This is a rejection-free algorithm which samples a spin $i$ with probability proportional to $(1+\exp(\beta\Delta E_i(\sigma)))^{-1}$, flips the resulting spin, and then advances the simulation time by a stochastic amount proportional to the waiting time in configuration $\sigma$. Here, the waiting time is the average time elapsed before a spin is flipped in the regular Monte Carlo method. This provides an effective algorithm which avoids numerous rejected moves due the low acceptance probability in Eq.~\eqref{eq:spinflip_acceptance}, thereby speeding up the simulation. A pseudocode for this approach is given in Algorithm~\ref{alg:nfoldwayMC}.

Since the decoder runtime is significantly slower than updating individual spins and we have to probe dynamics over exponentially-long periods of times, we perform decoding at geometrically-increasing time intervals. Specifically, we run the decoder if the simulation time has increased by at least $10\%$ relative to the time of the previous decoder run.

\begin{algorithm}[H]
    \caption{$n$-fold way Monte Carlo~\cite{BORTZ1975} memory time simulation using Glauber dynamics with geometric decoding schedule}\label{alg:nfoldwayMC}
    \begin{algorithmic}[1]
    \Require{$\operatorname{Dec}$ a syndrome decoder algorithm.}
        \State Initialize spins $\sigma=(\sigma_1,\ldots, \sigma_N)\gets (1,\ldots,1)$
        \State Initialize simulation time $t\gets 0$ and time until next decoding $t_{\rm dec}\gets 0$.
        \State Let $\Delta E(H)=\{\Delta E_i(\sigma) : i\in [N] \text{ and } \sigma\in \{\pm 1\}^N\}$ be all possible energy changes caused by a single flip.
        \State Initialize sets $\{M_j\}_{j\in \Delta E(H)}$, where each $M_j$ contains spin indices of $\sigma$ whose flip cause energy change $j$.
        \State Compute rates $\{P_{j}\}_{j\in \Delta E(H)}$ as $P_{j}\gets \frac{1}{1+\exp(\beta j)}$
        \Repeat
        \State Compute $Q_k =\sum_{j \in \Delta E(H) : j\le k} P_j |M_j|$ for all $k \in \Delta E(H)$ 
        \State Draw $R$ uniformly at random from interval $[0,\max_k(Q_{k}))$ 
        \State Find $\ell$ such that $Q_{\ell-1}\le R<Q_{\ell}$, where we set $Q_{m-1}=0$ for $m=\min \Delta E(H)$
        \State Draw a spin index $i\sim \mathrm{Uniform}(M_\ell)$ uniformly at random

        \State Flip spin with index $i$, i.e., set $\sigma_i\gets -\sigma_i$
        \State Update $\{M_j\}_{j\in \Delta E(H)}$ so that it corresponds to the new spin configuration $(\sigma_1,\ldots,\sigma_{i-1},-\sigma_i,\sigma_{i+1},\ldots, \sigma_N)$
        \State Compute $\Delta t= R'/\max_k(Q_{k})$ for a random value $R'\sim \mathrm{Exp}(1)$
        \State Advance the simulation time $t\gets t+\Delta t$ and decrease the time until next decoding $t_{\rm dec}\gets t_{\rm dec}-\Delta t$
        \If {$t_{\rm dec}\le 0$}
        \State Compute the error vector $e=\tfrac{1}{2}((1,\ldots,1)-\sigma)\in \mathbb{F}_2^N$ and its syndrome $s\in\mathbb{F}_2^{|S_Z|}$
        \State Decode syndrome $s$ with algorithm $\operatorname{Dec}$.
        \If{the decoder results in a logical error}
            \State \Return $t$
        \EndIf
        \State Advance the next decoding time geometrically $t_{\rm dec}\gets 0.1 t$
        \EndIf
        \ENDREPEAT
    \end{algorithmic}
\end{algorithm}

\begin{remark}
    All simulations of threshold and memory time for layer codes were done with the original layer code construction~\cite{williamson2023layer}. This construction differs slightly from the construction used to obtain theoretical bounds in previous sections since the $X$- and $Z$-check layers do not stretch from the first to the last qubit but stretch only through the qubits the checks are supported on. We note that the difference for random input codes is minimal with high probability. Additionally, the distance between the surface code layers was set to $K=1$, which we do not expect to affect the result.
\end{remark}

\subsection{Random layer codes}\label{sec:rand_layer_codes}

In this section, we study the memory times of random layer codes. For each input code size $n$, we sample $20$ random layer codes as follows: Let us say that a CSS code $C$ is \emph{balanced} if $d_X=d_Z$. We sample $1000$ balanced CSS codes with $k=1$ and keep the first $20$ with the highest observed distance. Our layer codes are then constructed from these $20$ random input codes.

Let us briefly comment on the choice behind this random input code family. We choose $k=1$ to simplify the numerics, and we choose our CSS codes to be balanced to maintain symmetry between $X$ and $Z$, since we only perform the simulation for $X$-errors. We maximize the distance in order to avoid small size effects associated small distance input codes. For instance, we display the frequency of $(d_X,d_Z)$ and their marginals for $10000$ randomly sampled CSS codes with $n=11$ and $k=1$ in Fig.~\ref{fig:LC_HDPC_threshold_dist_distr}(b).
We observe that codes with distance $1$ forms a large fraction of random CSS codes for such small $n$. A layer code constructed with a distance $1$ input code will effectively act as a 2D surface code, i.e., it will possess a logical operator which consists of a single string, see Section~\ref{sec:LC:logicals}. To avoid degenerate scenarios such as this, we post-select on our sampled codes for maximum distance. 

We first numerically find the logical error rate as a function of the physical error rate for random layer codes. The results are shown in Fig.~\ref{fig:LC_HDPC_threshold_dist_distr}(a).

\begin{figure}[!ht]
    \centering
    (a)\includegraphics[width=0.46\linewidth]{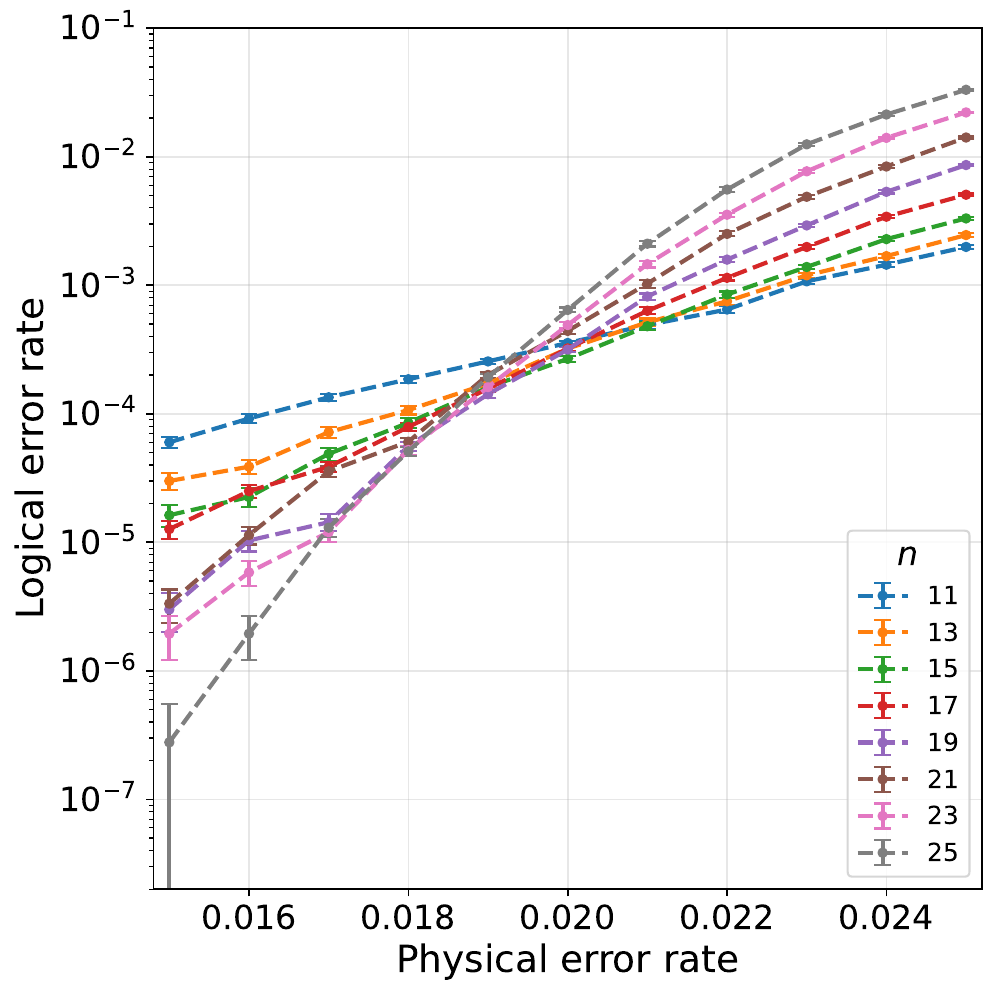}
    (b) \raisebox{2mm}[0pt][0pt]{\includegraphics[width=0.46\linewidth]{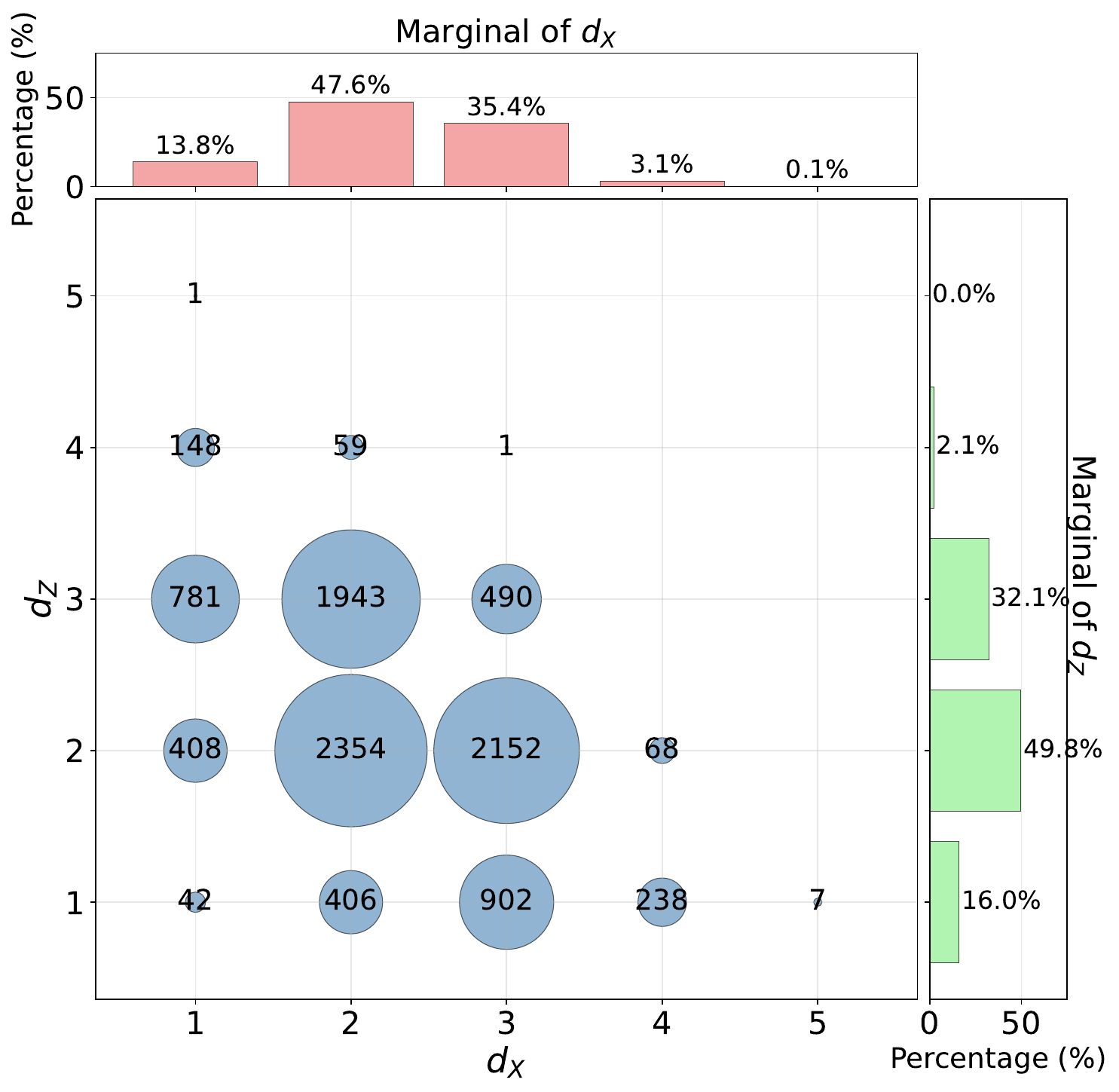}}
    \caption{(a) We numerically find the logical error rate achieved with cluster decoder as a function of the physical error rate for independent bit-flip errors for random layer codes. The error bars display standard error. (b) Frequencies of $(d_X,d_Z)$ distances and their marginals for random CSS codes of small length $n=11$ and $k=1$ obtained from $10000$ random samples. We can observe that the codes of distance $1$ have large representation.
    }
    \label{fig:LC_HDPC_threshold_dist_distr}
\end{figure}

We next performed memory time simulations on random layer codes subject to thermal noise as described in Appendix~\ref{sec_est_mem_time}. For each input size, bath temperature, and random layer code instantiation, we run $40$ simulations of the memory time. The memory time simulations are performed with the $n$-fold way Monte Carlo (see Algorithm~\ref{alg:nfoldwayMC}) using the cluster decoder. We then computed the sample average memory time $\langle t_{\rm fail}\rangle$ over all trials and codes for each pair $(n,\beta)$. The result is displayed in Fig.~\ref{fig_memory} of the main text for $\langle t_{\rm fail}\rangle$ as a function of $n$ for various $\beta$, and in Fig.~\ref{fig:coherence_time_random_LC_scaling_with_beta}(a) for $\langle t_\text{fail}\rangle$ as a function of $\beta$ for various $n$. The error bars in both figures depict the standard error of the mean (SEM), $\overline{\sigma}/\sqrt{m}$, where $\overline{\sigma}$ is the sample standard deviation and $m$ is the number of samples. 

We attempted to verify if the observed memory times $\langle t_{\mathrm{fail}}\rangle$ are consistent with the theoretical bounds obtained for random layer codes in Section~\ref{sec:partial_self_corr_HDPC}. We emphasize that the cluster decoder used to numerically simulate the memory time is different from the concatenated decoder for which we prove partial self-correction. We do not expect the cluster decoder to correct errors with large energy barrier. Moreover, since the analytical guarantees for the random layer codes are asymptotic in large $n$, finite size effects may be significant.

From the analytic bounds (see Theorem~\ref{thm:psc_random_layer} and Remark~\ref{rem:LscalingHDPC}) we expect a memory time scaling as
\begin{align}
    t_{\rm mem }\ge \exp(\Omega(\beta n/\log n)) \qquad \text{ for }\qquad n\le n^*=\exp(O(\beta)).
\end{align}
In contrast, it appears that a power law scaling of $\langle t_{\mathrm{fail}}\rangle$ with $n$ is a better fit to our numerical data obtained with cluster decoder. Consider the ansatz
\begin{align}
    \langle t_{\mathrm{fail}}\rangle = \Theta( n^{a\beta+b}) \qquad \text{ for }\qquad n\le n^*=\exp(O(\beta)), \label{eq:ansatz_poly}
\end{align}
which implies that $\langle t_{\mathrm{fail}}\rangle^*$ behaves quadratically. We determined a maximum size $n^*$ and the corresponding maximum memory time $\langle t_{\mathrm{fail}}\rangle^*$ for each inverse temperature $\beta$, and using this data, we fit a quadratic function to $\log  \langle t_{\mathrm{fail}}\rangle^*$ in Fig.~\ref{fig_memory}(c) of the main text. We obtain a best fit of $\log \langle t_{\mathrm{fail}}\rangle^*\approx 0.695\beta^2-7.112\beta+26.073$. Next, we also fit the cutoff size $\log n^*$ to a linear function in Fig.~\ref{fig_memory}(b) of the main text, obtaining a best fit function of $\log n^* \approx 0.448\beta-0.562$. Finally, we determine the constants $a$ and $b$ of the ansatz~\eqref{eq:ansatz_poly} as follows: For each $\beta$, we fit a linear function $\log \langle t_\mathrm{fail}\rangle \approx s(\beta)\log n+\mathrm{c}$ and extracted the slope $s(\beta)$. The linear fits for each $\beta$ are displayed in the Fig.~\ref{fig_memory} of the main text. We then fitted the slopes $s(\beta)$ to a linear function $s(\beta)\approx a\beta+b$ to extract the values of $a$ and $b$. The results are displayed in Fig.~\ref{fig:coherence_time_random_LC_scaling_with_beta}(b). We observe $a\approx 1.732$ and $b\approx -13.235$. The results appear to be consistent with the theoretical ansatz of Eq.~\eqref{eq:ansatz_poly}, which confirms partial self-correction as formally defined in Definition~\ref{def:pscqm} since we observe $\langle t_{\mathrm{fail}}\rangle^*=\exp(\Theta(\beta^2))=\exp(\omega(\beta))$. This scaling is similar to cubic code obtained in Ref.~\cite{bravyi2011analytic}. 

\begin{figure}[!ht]
    \centering
    (a)\includegraphics[width=0.45\linewidth]{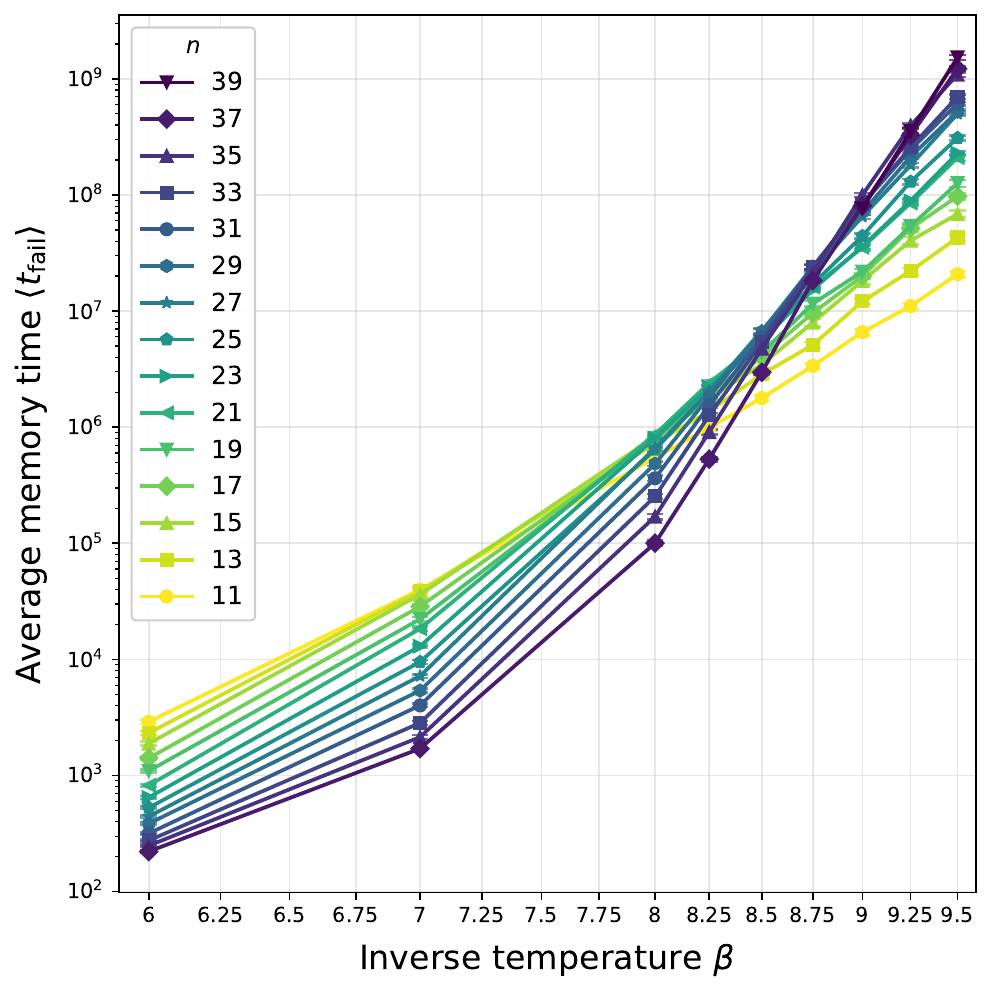}
    (b)\includegraphics[width=0.45\linewidth]{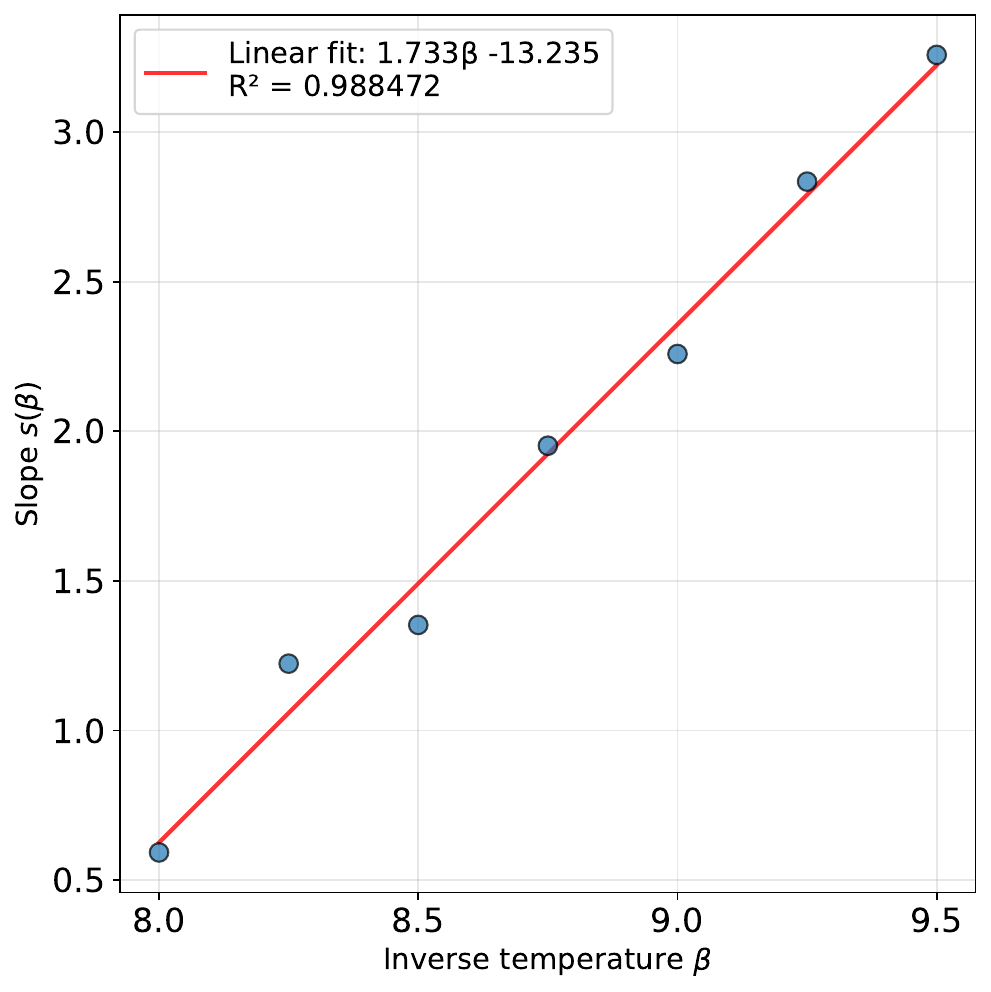}
    \caption{(a) Average memory time of random layer codes  as a function of inverse temperature $\beta$
    for various random input code lengths $n$.
    Error bars indicate the standard error of the mean (SEM). The solid lines are depicted for visual guidance.
    (b) Analysis of the power law exponent for memory time with ansatz $\log \langle t_{\rm fail}\rangle \approx (a\beta+b)\log n$. The obtained fit is $\log \langle t_{\rm fail}\rangle \approx (1.732\beta -13.235)\log n$.
    }\label{fig:coherence_time_random_LC_scaling_with_beta}
\end{figure}

\subsection{Comparison with the 3D Toric Code}\label{sec:num_3Dtoric}

To contrast the partial self-correcting behavior of random layer codes with a known self-correcting memory, we also simulate the memory time of the 3D toric code, which is a self-correcting memory for $X$-errors. Using the sweep decoder~\cite{KubicaPreskill,Vasmer2021}, we obtain the following results.
In Fig.~\ref{fig:coherence_time_3DtoriccodeZ}(a), we find the threshold of the sweep decoder under the bit-flip Pauli $X$ noise; our numerical estimate is in agreement with the previously reported values.
In Fig.~\ref{fig:coherence_time_3DtoriccodeZ}(b), we display the estimated memory time of the 3D toric code and clearly observe that the average memory time scales exponentially with the linear system size $L$, i.e., $\langle t_{\rm fail}\rangle=\exp(\Omega(\beta L))$, with no observed upper bounds. The memory time scaling is consistent with the analytical bounds of Ref.~\cite{alicki2008thermalstabilitytopologicalqubit}, demonstrating a qualitatively different behavior from that of the random layer codes depicted in Fig.~\ref{fig_memory} of the main text.

\begin{figure}[!ht]
    \centering
    (a)\includegraphics[width=0.46\textwidth]{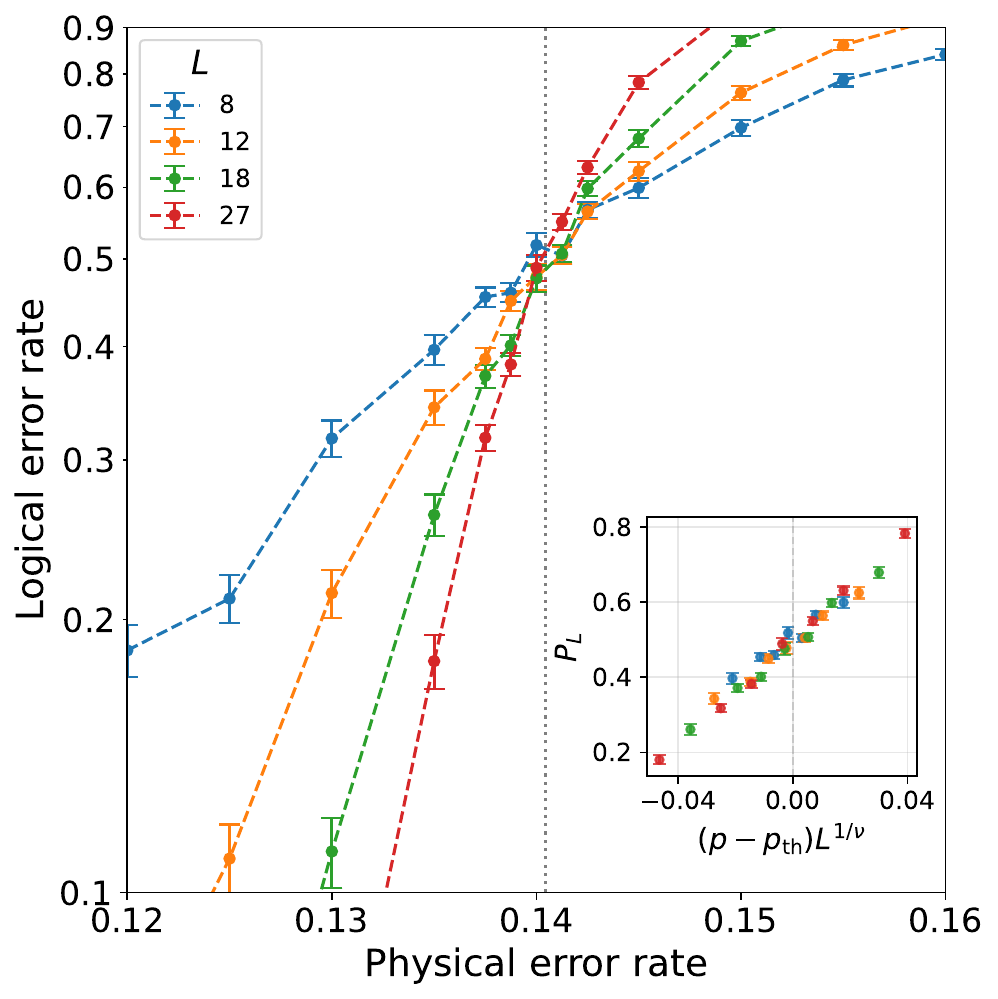}
    (b)\includegraphics[width=0.46\textwidth]{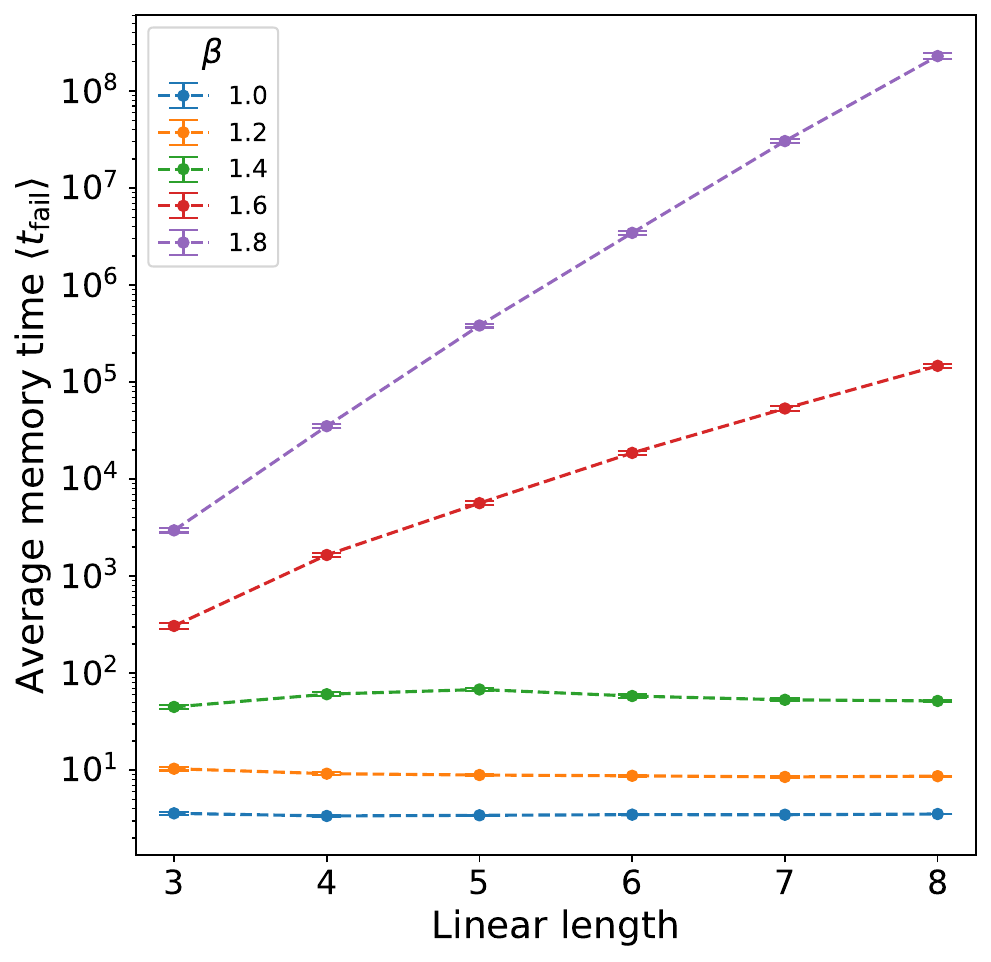}
    \caption{(a) The threshold of 3D toric code on the cubic lattice with periodic boundary under an independent and identically-distributed bit-flip noise using sweep decoder~\cite{KubicaPreskill,Vasmer2021} is estimated to be $p_{\text{th}}\approx 0.14$. In the inset, we display the same data using rescaled variable $(p-p_\text{th})L^{1/\nu}$ with linear length $L$ and fitting parameters $p_\text{th}$ and $\nu=1.535$.
    (b) Log plot of the sample average memory time of the $3D$ toric code for the bit-flip noise as a function of the linear system size $L$ for various inverse temperatures $\beta$.
    Each data point is obtained by taking $500$ samples.
    Error bars indicate standard error in both panels.
    The dashed lines are only displayed for visual guidance.
    }
    \label{fig:coherence_time_3DtoriccodeZ}
\end{figure}

\FloatBarrier
\bibliographystyle{alphaurl}
\bibliography{bib}

\end{document}